\documentclass[a4paper,twocolumn,superscriptaddress,11pt,accepted=2019-07-10]{quantumarticle}
\pdfoutput=1
\usepackage[utf8]{inputenc}
\usepackage[english]{babel}
\usepackage[T1]{fontenc}
\usepackage{amsmath}
\usepackage{hyperref}

\usepackage[numbers,sort&compress]{natbib}

\usepackage{tikz}
\usepackage{lipsum}

\usepackage{braket}
\usepackage{mathrsfs}

\usepackage{graphicx}

\usepackage{amsthm,amsmath,amssymb,bbm}
\allowdisplaybreaks
\usepackage[bb=boondox]{mathalfa}
\usepackage{enumitem}

\newcommand{\ii}{\mathrm{i}} % imaginary unit
\newcommand{\ee}{\mathrm{e}} % Euler number
\newcommand{\ie}{\emph{i.e.}, }
\newcommand{\eg}{\emph{e.g.}, }
\newcommand{\comm}[2]{\left[{#1},{#2}\right]} %Commutator
\DeclareMathOperator{\Tr}{Tr}
\newcommand{\mat}[1]{\mathbf{#1}}
\newcommand{\id}{\mathbb{1}} %identity
\newcommand{\pro}{\mathbb{P}}

\DeclareMathOperator{\ch}{ch}
\DeclareMathOperator{\sh}{sh}
\renewcommand{\cosh}{\ch}
\renewcommand{\sinh}{\sh}

\newtheorem{proposition}{Proposition}
\newtheorem{definition}{Definition}

\begin{document}

\title{Minimal energy cost of entanglement extraction}
\date{\today}
\author{Lucas Hackl}
\orcid{0000-0002-4172-0317}
\affiliation{Max Planck Institute of Quantum Optics, Hans-Kopfermann-Str. 1, 85748 Garching, Germany}
\affiliation{Munich Center for Quantum Science and Technology, Schellingstra{\ss}e 4, D-80799 M\"{u}nchen, Germany}
\email{lucas.hackl@mpq.mpg.de}
\author{Robert H. Jonsson}
\orcid{0000-0003-0295-250X}
\affiliation{QMATH, Department of Mathematical  Sciences,  University  of  Copenhagen, Universitetsparken  5,  2100  Copenhagen,  Denmark}
\email{robert.jonsson@math.ku.dk}

\maketitle

\begin{abstract}
We compute the minimal energy cost for extracting  entanglement  from the ground state of a bosonic or fermionic quadratic system. 
Specifically, we find the minimal energy increase in the system resulting from replacing an entangled pair of modes, sharing entanglement entropy $\Delta S$, by a product state,
and we show how to construct modes achieving this minimal energy cost.
Thus, we obtain a protocol independent lower bound on the extraction of pure state entanglement from quadratic systems.
Due to their generality, our results apply to a large range of physical systems, as we discuss with examples.
\end{abstract}

\section{Introduction}
Entanglement is a hallmark of quantum theory:
On a fundamental level, its very existence has deep implications for our understanding of nature, \eg as pertaining to its realism and localism \cite{einstein_can_1935,bell_einstein_1964,bell_einstein_2001,hensen_loophole-free_2015,giustina_significant-loophole-free_2015,shalm_strong_2015}.
From an applied point of view, entanglement is a key resource in quantum technologies and quantum information processing \cite{horodecki_quantum_2009,bennett_communication_1992,bennett_teleporting_1993,bennett_quantum_2014}.
And moreover, studying entanglement structures provided deep insights in various fields ranging from quantum many body systems, over quantum field theory to quantum gravity \cite{srednicki_entropy_1993,vidal_entanglement_2003,osborne_entanglement_2002,eisert_colloquium_2010,amico_entanglement_2008,islam_measuring_2015,ryu_holographic_2006}.

The fact that the ground states of local Hamiltonians are generally highly entangled gives rise to the idea that such systems could serve as a source of entanglement, \eg for applications in quantum information.
In particular, the idea of extracting entanglement from the ground state of a system has been explored in the context of the vacuum state of a quantum field.
Starting with the observation that the entanglement present in the vacuum fluctuations can violate Bell's inequalities \cite{summers_vacuum_1985,summers_bells_1987,summers_bells_1987a,valentini_non-local_1991,reznik_violating_2005},
various works studied the ability of the quantum field vacuum to entangle systems which couple locally to the quantum field. In particular, the Unruh-DeWitt particle detector model has been employed to investigate aspects ranging from the impact of relativistic motion \cite{salton_acceleration-assisted_2015}, spacetime curvature and topology \cite{steeg_entangling_2009,cliche_vacuum_2011,martin-martinez_spacetime_2016}, the nature of field states and interactions \cite{braun_creation_2002,braun_entanglement_2005,brown_thermal_2013,sachs_entanglement_2017,simidzija_general_2018,simidzija_harvesting_2018}, and potential implementations  \cite{olson_entanglement_2011,olson_extraction_2012,sabin_extracting_2012,pozas-kerstjens_entanglement_2016}.

Previous works on the relation between entanglement and energy content of a system focused mainly on the energy cost of creating entanglement and correlations, or looked at the  conversion of one into the other, \eg in the context of quantum thermodynamics \cite{galve_energy_2009,martin-martinez_sustainable_2013,huber_thermodynamic_2015,friis_energetics_2016,das_canonical_2017,chiribella_optimal_2017,vitagliano_trade-off_2018}.
Also it was shown that  it is possible to use the entanglement present in the vacuum of a quantum field to teleport energy from one spacetime region to another. In quantum energy teleportation one observer performs a local measurement of the field, classically communicates the result to another observer who then through local unitary interactions is able to extract energy from the field \cite{hotta_quantum_2008,hotta_controlled_2010,hotta_energy_2010,hotta_quantum_2014}.

However, the systematic study of the energy cost of entanglement extraction from complex quantum systems was only initiated recently by Beny et al. \cite{beny_energy_2018}. When entanglement is extracted from a system, its state is changed and thus its energy content may change as well. In particular, the extraction of entanglement from a system's \emph{unique} ground state is inevitably associated with  an energy increase injected into the system when changing its state. 

This interplay of entanglement extraction and energy cost is not only relevant with respect to potential implementations and quantum thermodynamics but, in fact, also connects to quantum field theory on curved spacetimes. For example, Jacobson \cite{jacobson_entanglement_2016}  recently showed that the semi-classical Einstein equations are equivalent to the property of the vacuum of a field to locally maximize the entanglement entropy. Attempting to increase a sphere's entanglement by varying the field state induces an increase in energy density which curves spacetime and shrinks the sphere's surface, thus reducing the entanglement contributions proportional to the surface area, and exactly cancelling out the attempted increase of entanglement.

In this article, we present a first result in the direction initiated by \cite{beny_energy_2018} which applies to large classes of systems of physical interest: We study the extraction of entanglement from the ground state of bosonic and fermionic modes governed by a quadratic Hamiltonian. This also includes spin systems that can be mapped to quadratic fermionic Hamiltonians via the Jordan-Wigner transform. In particular, we find the minimal energy cost for the extraction of mode pairs in a pure entangled state.

An important tool to our analysis is the  partner mode formula, which was developed a in series of recent papers \cite{hotta_partner_2015,trevison_pure_2018,trevison_spatially_2019}, mostly in the context of bosonic quantum field theory.
We generalize this construction, such that it applies to any bosonic or fermionic pure Gaussian state. This is achieved by phrasing the construction in terms of the complex structure associated with a pure Gaussian state~\cite{ashtekara._quantum_1975,wald_quantum_1994}, thus illuminating the geometric origin of the partner mode construction.

The partner mode construction is based on the well-known property of pure Gaussian states to factorize into pairs of entangled mode pairs \cite{botero_mode-wise_2003} (see also \cite{wolf_not-so-normal_2008}). If a system of many modes is in an overall pure Gaussian state, then for every mode in that system, which is not in a pure state on its own, there exists a unique partner mode, with which it shares all its entanglement. Thus, the mode and its partner are in a pure entangled state, that is in a product state with the rest of the system.

This property is the motivation for extracting entanglement by unitarily swapping the entangled state of such partner modes with an unentangled state of target modes in some laboratory system. Thus, the system acts as a source of an entangled state that is traded with an unentangled state from the target system. This framework was introduced in \cite{trevison_pure_2018} where the relationship between energy cost and entanglement extraction was discussed for a harmonic chain.

In this article, we generalize this framework to arbitrary systems of bosonic or fermionic modes with quadratic Hamiltonians. We find the minimal energy cost for the extraction of partner mode pairs from their ground state as analytical functions of the Hamiltonian's spectrum. Furthermore, we show how to explicitly construct  partner modes that achieve this minimal energy cost.

The structure of this article is as follows: In Section~\ref{sec:boson_review}, we review bosonic Gaussian states and quadratic Hamiltonians. In Section~\ref{sec:boson_extraction_wholesection}, we present our approach of extracting entanglement from partner modes and derive an analytical formula for the minimal energy cost. The following two Sections~\ref{sec:fermion_review} and~\ref{sec:fermion_extraction} mirror the previous two for fermionic systems, \ie we first review fermionic Gaussian states and quadratic Hamiltonians in Section~\ref{sec:fermion_review} and then derive the respective minimal energy cost for fermions in Section~\ref{sec:fermion_extraction}. Finally, Section~\ref{sec:applications} discusses applications of our results for both bosonic and fermionic models, before we conclude with a discussion in Section~\ref{sec:discussion}. Appendix~\ref{app:normalform} rederives the pairwise correlation structure of bosonic and fermionic Gaussian states in the newly developed language of linear complex structures. Appendix~\ref{app:spectrumproof} reviews relevant results in linear algebra on how the spectrum of a linear map changes under restrictions.

The article is structured to be largely self-contained. To the experienced reader Sections~\ref{sec:boson_review} and~\ref{sec:fermion_review} will mostly serve to establish notation for bosons and fermions respectively, and to present a covariant way of computing partner modes.

\section{Quadratic bosonic systems and their partner mode construction}\label{sec:boson_review}
The partner mode construction is a central tool in our derivation of the minimal energy of entanglement extraction. In this section, we introduce the partner mode formula for arbitrary pure Gaussian states of bosonic modes (specifically in Section \ref{sec:partnermode}). To be largely self-contained and to establish notation, we begin with a review of quadratic bosonic systems. We follow the conventions of~\cite{bianchi_squeezed_2016,hackl_entanglement_2018,hackl_aspects_2018}, while other comprehensive reviews include~\cite{weedbrook_gaussian_2012,eisert_colloquium_2010,plenio_entropy_2005,casini_entanglement_2009,woit_quantum_2017}.

\subsection{Phase space and observables}
For $N$ bosonic degrees of freedom, we consider a classical phase space $V\simeq\mathbb{R}^{2N}$ and its dual $V^*\simeq\mathbb{R}^{2N}$. We denote vectors $v^a\in V$ with an upper index and dual vectors $w_a\in V^*$ with a lower index. As phase space, both $V$ and $V^*$ are equipped with symplectic forms, namely $\Omega^{ab}: V^*\times V^*\to\mathbb{R}$ on $V^*$ and its inverse $\Omega_{ab}^{-1}: V\times V\to\mathbb{R}$ satisfying $\Omega^{ac}\Omega_{cb}^{-1}=\delta^a{}_b$.

We can characterize a vector $v^a\in V$ by its coordinate values with respect to a basis as
\begin{align}
    v^a\equiv\left(q_1(v),p_1(v),\cdots,q_N(v),p_N(v)\right)^\intercal\in \mathbb{R}^{2N} \,,\label{eq:def-v2}
\end{align}
where ``$\equiv$'' indicates that $v^a$ is represented by the column vector  with respect to the specified basis.

Typically, we choose coordinate functions   $q_i,p_i\in C^\infty(V)$ which are linear, \ie $q_i,p_i\in V^*$, and which consist of conjugate variables, \ie satisfying  $\Omega(q_i,p_j)=\delta_{ij}$.

With respect to such a basis, the symplectic form $\Omega^{ab}$  takes the standard Darboux form, \ie is represented by the block diagonal matrix
\begin{align}
    \Omega^{ab}&\equiv\mathbf{\Omega}=\left(\begin{array}{ccc}
    \mat{\Omega}_2 &    & 0 \\
    &\ddots &\\
    0 & &\mat{\Omega}_2
    \end{array}\right) \, ,
    \label{eq:staOmega2}
\end{align}
where we use   boldface symbols $\mathbf{\Omega}$ for the matrix representation with respect to a specific basis, and define the $2\times 2$-matrix
\begin{align}
    \mat{\Omega}_2=\begin{pmatrix}0&1\\-1&0\end{pmatrix}\, .
\end{align}

In the classical theory, observables correspond to smooth functions on the phase space $V$, \ie functions in $C^\infty(V)$. 
In particular, a covector $w\in V^*$  defines a linear observable through
\begin{align}
    \mathcal{O}_w:V\to \mathbb{R} : v^a\mapsto w_a v^a\,.
\end{align}
Similarly, a bilinear form $h_{ab}$ on $V$ describes uniquely a quadratic observable
\begin{align}
    \mathcal{O}_h: V\to \mathbb{R}: v^a\mapsto \frac{1}{2}h_{ab}v^av^b\,.
\end{align}

In quantum theory, observables correspond to linear operators on a Hilbert space $\mathcal H$. To define linear and multi-linear observables it is convenient to introduce the operator-valued vector $\hat\xi^a$ referred to as quantization map. It is subject to the constraint
\begin{align}
    [\hat{\xi}^a,\hat{\xi}^b]=\hat{\xi}^a\hat{\xi}^b-\hat{\xi}^b\hat{\xi}^a=\{\xi^a,\xi^b\}=\ii\Omega^{ab}\id
\end{align}
where $\id$ represents the identity operator on $\mathcal H$.
With respect to our initial basis where $\Omega$ takes the form of~\eqref{eq:staOmega2}, $\hat{\xi}^a$ is represented as
\begin{align}
    \hat{\xi}^a\equiv(\hat{q}_1,\cdots,\hat{q}_N,\hat{p}_1,\cdots,\hat{p}_N)^\intercal
\end{align}
with the familiar position and momentum operators. From here, we can construct the well-known creation and annihilation operators $\hat{a}_i=\frac{1}{\sqrt{2}}(\hat{q}_i-\ii\hat{p}_i)$ and $\hat{a}_i^\dagger=\frac{1}{\sqrt{2}}(\hat{q}_i+\ii\hat{p}_i)$.

The quantum observables associated to the linear form $w_a$ and quadratic form $h_{ab}$ are then
\begin{align}
    \hat{\mathcal{O}}_w=w_a\hat{\xi}^a\quad\text{and}\quad\hat{\mathcal{O}}_h=\frac{1}{2}h_{ab}\hat{\xi}^a\hat{\xi}^b\,.
\end{align}
In contrast to classical observables, a quantum observable constructed from a general tensor $t_{a_1\cdots a_n}$ depends on the tensor's non-symmetric part since the operators $\hat q_i$ and $\hat p_i$ do not commute.

Note that in the classical theory, the quantization map reduces to the vector-valued function 
\begin{align}
    \xi^a: V\to V: v^a\mapsto \xi^a(v)=v^a\,,
\end{align}
which is, in fact, just the identity map on $V$.
It can be used to represent linear and multi-linear observables, \eg through
\begin{align}
    \mathcal{O}_t: v\mapsto \frac{1}{n!}t_{a_1\cdots a_n} \xi^{a_1}(v)\cdots \xi^{a_n}(v)%= t_{a_1\cdots a_n} v^{a_1}\cdots v^{a_n}\,.
\end{align}
for a general tensor $t_{a_1\cdots a_n}$. In terms of a linear basis of $V$ it is represented as the tupel of functions
\begin{align}
    \xi^a\equiv(q_1,p_1,\cdots,q_N,p_N)^\intercal \label{eq:def-xi2}.
\end{align}
Because  each  component of $\xi^a$, \ie $q_i,p_i:V\to \mathbb{R}$, is a function on  phase space, the vector valued function  encodes the Poisson brackets of the classical theory through
\begin{align}
    \{\xi^a,\xi^b\}=\Omega^{ab}\,.
\end{align}
This is particularly convenient and allows for a clean notation when expressing Poisson brackets of observables abstractly.  For instance, for two linear observables $\mathcal{O}_w=w_a\xi^a$ and $\mathcal{O}_u=u_a\xi^a$, we find immediately
\begin{align}
    \{\mathcal{O}_w,\mathcal{O}_u\}=w_au_b\{\xi^a,\xi^b\}=w_au_b\Omega^{ab}\,.
\end{align}

\subsection{Bosonic Gaussian states}
Bosonic Gaussian states are an important class of quantum states because they can be described through powerful analytic methods~\cite{weedbrook_gaussian_2012}. They appear as ground states of quadratic Hamiltonians and their wave function representation is given by multi-dimensional complex Gaussian distributions. When representing them as quasi-probability distributions on classical phase space, they are the only states with positive Wigner functions~\cite{hudson_when_1974,soto_when_1983}.  In particular, their Wigner functions are multi-dimensional real Gaussians again. In the context of Bogoliubov transformations, Gaussian states are often identified as the vacuum associated to a set of annihilation operators satisfying bosonic commutation relations. Finally, their $n$-point correlation functions are completely determined from their $1$-point and 2-point correlation function via Wick's theorem~\cite{wick_evaluation_1950}.

We label a Gaussian state $\ket{G,z}$ by a displacement vector $z^a$ and a covariance matrix $G^{ab}$
\begin{align}
    z^a&=\bra{G,z}\hat{\xi}^a\ket{G,z}\,,\\ G^{ab}&=\bra{G,z}\hat{\xi}^a\hat{\xi}^b+\hat{\xi}^b\hat{\xi}^a\ket{G,z}-2z^az^b\,.
\end{align}
The covariance matrix $G^{ab}$ is a positive definite symmetric bilinear form on the dual phase space $V^*$, \ie $G$ equips $V^*$ with an inner product. 
In particular, for each covariance matrix there exists a basis of the phase space $V$ with respect to which the covariance matrix is represented by the identity matrix $G^{ab}\equiv\id$, while $\Omega$ takes the standard form~\eqref{eq:staOmega2}.

Furthermore, the covariance matrix of a pure Gaussian state and the inverse symplectic form satisfy the condition
\begin{align}
    G^{ab}\Omega^{-1}_{bc}G^{cd}\Omega^{-1}_{de}=-\delta^a{}_e\,.
\end{align}
In fact, we can use this property to define a linear map on $V$ 
\begin{align}\label{eq:boson_J_defn}
    J^a{}_b=-G^{ac}\Omega^{-1}_{cb}\,.
\end{align}
This map is referred to as linear complex structure on $V$ because it satisfies  
\begin{align}
J^2=(G\Omega^{-1})^2=-\mathbb{1}\,,
\end{align}
\ie it acts like the multiplication with the complex number $\ii$ on the real vector space. For a given covariance matrix and displacement vector, the Gaussian state $\ket{G,z}\in\mathcal{H}$ is then fully specified by the requirement that
\begin{align}
    \frac{1}{2}(\delta^{a}{}_{b}+\ii J^a{}_b)(\hat{\xi}^b-z^b)\ket{G,z}=0\,.\label{eq:Jproj}
\end{align}
Here, the first term corresponds to a projector onto the space spanned by those annihilation operators that annihilate $\ket{G,z}$. Describing Gaussian states in terms of linear complex structures was first introduced to describe unitarily inequivalent Fock space representations~\cite{ashtekara._quantum_1975,wald_quantum_1994,woit_quantum_2017}, but since then has been used to study entanglement production~\cite{bianchi_squeezed_2016,hackl_entanglement_2018}, typical entanglement of energy eigenstates~\cite{vidmar_entanglement_2017,vidmar_volume_2018,hackl_average_2019}  and to explore circuit complexity in free field theories~\cite{vidmar_volume_2018,chapman_complexity_2019}.

Wick's theorem provides an efficient way to compute arbitrary $n$-point functions
\begin{align}
    C^{a_1\cdots a_n}_{n}=\bra{G,z}(\hat{\xi}^{a_1}-z^{a_1})\cdots (\hat{\xi}^{a_n}-z^{a_n})\ket{G,z}.
\end{align}
where the subtraction of $z^a$ simplifies later expressions. Clearly, knowing $C^{a_1\cdots a_n}_{n}$ enables us to compute arbitrary correlations even without $z^a$. With the definitions of $G^{ab}$ and $z^a$, and the commutation relations $[\hat{\xi}^a,\hat{\xi}^b]=\ii \Omega^{ab}$, we find the 2-point function to be
\begin{align}
    C^{ab}_2=\frac{1}{2}(G^{ab}+\ii\Omega^{ab})\,.
\end{align}
With this expression, we can state Wick's theorem  compactly  as
\begin{align}
    C^{a_1\cdots a_n}_{n}&=\sum\text{(contractions of $C_2^{ab}$)}\\
    &=C^{a_1a_2}_2\cdots C^{a_{n-1}a_{n}}_2+\cdots\,,
\end{align}
where the sum over all contraction refers to writing $C^{a_1\cdots a_n}_n$ as sum of products of 2-point functions $C_2^{ab}$ with all possible inequivalent assignments of ordered indices $a_i$. This implies that all odd $n$-point functions vanish, because contractions require an even number of indices. Therefore, after the 2-point function, the next non-trivial correlation function is the four-point function given by
\begin{align}
    C_4^{abcd}=C_2^{ab}C_2^{cd}+C_2^{ac}C_2^{bd}+C_2^{ad}C_2^{bc}\,.
\end{align}
Note the importance of ordering the indices on each 2-point function $C^{a_ia_j}_{2}$  in the same way as on $C_n^{a_1\cdots a_n}$, \ie we require $i<j$.

\paragraph{Example.} In order to illustrate our methods and to fix conventions, let us look at the ground state $\ket{G,z}$ of the harmonic oscillator with Hamiltonian $\hat{H}=\frac{1}{2}(\hat{p}^2+\omega^2\hat{q}^2)$. 
Clearly, we have $z^a=0$, because the ground state has zero expectation value of position $\hat{q}$ and momentum $\hat{p}.$ We represent everything in the basis $\hat{\xi}^a\equiv(\hat{q},\hat{p})^\intercal$. Here, we have
\begin{align}
    \Omega^{ab}\equiv\begin{pmatrix}0 & 1\\ -1 & 0\end{pmatrix}\,\,\,\text{and}\,\,\,
    G^{ab}\equiv\begin{pmatrix}1/\omega & 0\\ 0 & \omega\end{pmatrix}\,.\label{eq:GofSHO}
\end{align}
We can use $\Omega^{ab}$ and $G^{ab}$ to compute the matrix representation of the linear complex structure
\begin{align}
    J^a{}_b=-G^{ac}\Omega^{-1}_{cb}\equiv\begin{pmatrix}0 & 1/\omega\\-\omega&0\end{pmatrix}\,.
\end{align}
With this, the  projection of equation~\eqref{eq:Jproj} reads
\begin{align}
    \frac{1}{2}(\delta^a{}_b+\ii J^a{}_b)\hat{\xi}^b\equiv\frac{1}{2}\begin{pmatrix}\hat{q}+\frac{\ii}{\omega}\,\hat{p}\\\hat{p}-\ii\omega \hat{q}\end{pmatrix}=\begin{pmatrix}\sqrt{\frac{1}{2\omega}}\,\hat{a}\\-\ii\sqrt{\frac{\omega}{2}}\hat{a}\end{pmatrix}\,,
\end{align}
\ie we project onto the subspace spanned by the annihilation operator $\hat{a}=\sqrt{\frac{\omega}{2}}(\hat{q}+\frac{\ii}{\omega}\,\hat{p})$.

\subsection{Quadratic bosonic Hamiltonians}
The most general quadratic Hamiltonian
\begin{align}
    \hat{H}=\frac{1}{2}h_{ab}\hat{\xi}^a\hat{\xi}^b+f_a\hat{\xi}^a\,,
\end{align}
contains a linear and a quadratic term, where we require $h_{ab}$ to be positive definite and symmetric. This ensures that $\hat{H}$ is bounded from below and its ground state is a Gaussian state $\ket{G,z}$. The covariance matrix $G^{ab}$ and the linear displacement $z^a$ can be computed from $h_{ab}$ and $f_a$:

\begin{itemize}
\item \textbf{Computation of $z^a$}\\
    The ground state minimizes the energy expectation value $E=\bra{G,z}\hat{H}{\ket{G,z}}$ given by
    \begin{align}
    E=\frac{1}{4}h_{ab}G^{ab}+\frac{1}{2}h_{ab}z^az^b+f_az^a\,.
    \end{align}
    This implies $z^ah_{ab}+f_b=0$  giving 
    \begin{align}
        z^a=-(h^{-1})^{ab}f_b\,, \label{eq:zformula}
    \end{align}
    with the inverse form $(h^{-1})^{ac}h_{cb}=\delta^a{}_b$, which exists since $h_{ab}$  is positive definite, thus non-degenerate.
\item \textbf{Computation of $G^{ab}$}\\
    We can compute $G^{ab}$ directly from $h_{ab}$ via
    \begin{align}
        G^{ab}=-K^a{}_c |K^{-1}|^{c}{}_d\Omega^{db}\,, \label{eq:Gfromh}
    \end{align}
    with $K^a{}_b=\Omega^{ac}h_{cb}$. $|K^{-1}|$ refers to 
    the linear map constructed by replacing the eigenvalues of $K^{-1}$ with their absolute values.
    %taking the absolute value of the eigenvalues of $K^{-1}$.
\end{itemize}
We can think of these formulas as a projection from the space of quadratic Hamiltonians onto the space of Gaussian states. Every Hamiltonian gives rise to a unique ground state, but for each state there are plenty of Hamiltonians that share the same ground state. This redundancy is visible in~\eqref{eq:Gfromh} where rescaling $h_{ab}$ with an overall factor (or individually in each frequency sector) leaves $G^{ab}$ unchanged.

The energy expectation value $E$ and its variance $\Sigma_E$ for Gaussian states of the form $\ket{G',z}$, \ie with the above $z^a$ but with arbitrary $G'$, are readily calculated as
\begin{align}
    E&=\bra{G',z}\hat{H}\ket{G',z}=\Tr(h G')/4+E_z\,,\\ 
    \Sigma_E&=\sqrt{\bra{G',z}\hat{H}^2\ket{G',z}-E^2}\\
    &=\sqrt{\Tr(h G'hG'+h\Omega h\Omega)}/4\,,
\end{align}
where $E_z=-\frac{1}{2}h_{ab}z^az^b=-\frac{1}{2}(h^{-1})^{ab}f_af_b$. 
This follows from~\eqref{eq:zformula} to rewrite
\begin{align}
    \hat{H}=\frac{1}{2}h_{ab}(\hat{\xi}-z)^a(\hat{\xi}-z)^b+E_z\,.
\end{align}

\paragraph{Example.} Let us consider the simple harmonic oscillator $\hat{H}=\frac{1}{2}(\hat{p}^2+\omega^2\hat{q}^2)$ from before. We find $f_a=0$ and with respect to $\hat{\xi}^a\equiv(\hat{q},\hat{p})$
\begin{align}
    h_{ab}\equiv\begin{pmatrix}\omega^2 &0\\0&1\end{pmatrix}\,.
\end{align}
Applying formula~\eqref{eq:Gfromh} yields
\begin{align}
    K^a{}_b\equiv\begin{pmatrix}0 &1\\-\omega^2&0\end{pmatrix}\,\,\Rightarrow\,\,
    G^{ab}\equiv\begin{pmatrix}1/\omega &0\\0&\omega\end{pmatrix}\,,
    %|K^{-1}|^a{}_b\equiv\begin{pmatrix}1/\omega &0\\0&1/\omega\end{pmatrix}
\end{align}
% leading to the ground state covariance matrix
% \begin{align}
%     G^{ab}=-K^a{}_c |K^{-1}|^{c}{}_d\Omega^{db}\equiv\begin{pmatrix}1/\omega &0\\0&\omega\end{pmatrix}\,,
% \end{align}
which agrees with the result from~\eqref{eq:GofSHO}. Here
\begin{align}
    |K^{-1}|\equiv \begin{pmatrix}\frac1\omega &0\\0&\frac1\omega\end{pmatrix}\,,
\end{align}
because  $K^{-1}$ has eigenvalues $\pm\frac\ii\omega$.

\subsection{Entanglement of bosonic Gaussian states}
In terms of a quantum system's Hilbert space $\mathcal{H}$, a bi-partition of the system corresponds to a tensor product $\mathcal{H}=\mathcal{H}_A\otimes\mathcal{H}_B$.
A bi-partition of a system of $N$ bosonic modes, into subsystems of $N_A$ and $N_B$ modes respectively, induces a decomposition of the phase space $V=A\oplus B$ into two symplectic complements $A$ and $B$.

Generally a Gaussian state on the full system $\ket{G,z}$ contains correlations, \ie entanglement between $A$ and $B$. However, there always exist bases 
$(q_1^A,p_1^A,\cdots,q^A_{N_A},p^A_{N_A})$ of $A$ and $(q_1^B,p_1^B,\cdots,q_{N_B}^B,p^B_{N_B})$ of $B$, such that the covariance matrix $G^{ab}$ takes the standard form \cite{botero_mode-wise_2003} re-derived in appendix~\ref{app:standardform-boson}
\begin{widetext}
\begin{align}
	G\equiv\left(\begin{array}{ccc|cccccc}
	\mat{ch}_1 & \cdots & 0 & \mat{sh}_1 & \cdots & 0 & 0 &\cdots & 0\\
	\vdots & \ddots & \vdots & \vdots & \ddots & \vdots & \vdots & \ddots & \vdots \\
	0 & \cdots & \mat{ch}_{N_A} & 0 & \cdots &\mat{sh}_{N_A} & 0 & \cdots & 0\\
	\hline
	\mat{sh}_1 & \cdots & 0 & \mat{ch}_1 & \cdots & 0 & 0 &\cdots & 0\\
	\vdots & \ddots & \vdots & \vdots & \ddots & \vdots & \vdots & \ddots & \vdots\\
	0 & \cdots & \mat{sh}_{N_A} & 0 & \cdots &\mat{ch}_{N_A} & 0 &\cdots & 0\\
	0 & \cdots & 0 & 0 & \cdots & 0 & \mathbb{1}_2 & \cdots & 0\\
	\vdots & \ddots & \vdots & \vdots & \ddots & \vdots & \vdots & \ddots & \vdots \\
	0 & \cdots & 0 & 0 & \cdots & 0 & 0 & \cdots & \mathbb{1}_2\\
	\end{array}\right)\,,\label{eq:sta}
\end{align}
\end{widetext}
built from the $2\times2$-matrices, with $r_i\geq0$,
\begin{align}
	\mat{ch}_i &=\left(\begin{array}{cc}
	\cosh{2r_i} & 0\\
	0 & \cosh{2r_i}
	\end{array}\right)\,,\\
	\mat{sh}_i &=\left(\begin{array}{cc}
	\sinh{2r_i} & 0\\
	0 & -\sinh{2r_i}
	\end{array}\right)\,.
\end{align}

The standard form highlights two important features of the state's entanglement structure. Firstly, the basis modes  bring the diagonal blocks, corresponding to $A$ and $B$ into diagonal form, \ie they provide a normal mode decomposition of the subsystems' partial states: The partial states of $A$, and $B$ respectively, factorize into a product state over the modes in each subsystem basis.

Secondly, all entanglement between $A$ and $B$ is contained in entangled mode pairs: If $r_i>0$, for a mode $(q_i^A,p_i^A)$ in $A$, then it is entangled  with the mode $(q_i^B,p_i^B)$ with which it is in a two-mode squeezed state.

The total von Neumann entropy $S_A(\ket{G,z})$ between $A$ and $B$ is given by the sum of the mode pairs' individual entropies~\cite{sorkin_entropy_2014,bombelli_quantum_1986} 
\begin{align}\label{eq:Sformula} 
    S_A(\ket{G,z})
    =\sum_{i=1}^{N_A}s_{b}(\cosh{2r_i})
\end{align}
with $s_{b}(x)=\left(\frac{x+1}{2}\right)\log\left(\frac{x+1}{2}\right)-\left(\frac{x-1}{2}\right)\log\left(\frac{x-1}{2}\right)$. Here, we chose the convention to compute all logarithms with respect to base $2$.

Note, that the total entropy can be expressed in a basis invariant way as~\cite{bianchi_squeezed_2016,hackl_aspects_2018}
\begin{align}
    S_A(\ket{G,z})=\mathrm{Tr}\left[\left(\frac{\mathbb{1}_A+\ii[J]_A}{2}\right)\log\left|\frac{\mathbb{1}_A+\ii[J]_A}{2}\right|\,\!\right]\,,\label{eq:SJbosons}
\end{align}
where $[J]_A$ is the restriction of the complex structure $J^a{}_b=-G^{ac}\Omega^{-1}_{cb}$ to the symplectic subspace $A$, \ie we restrict it to a $2N_A$-by-$2N_A$ matrix. The absolute value should be understood in terms of eigenvalues of $[J]_A$ which come in conjugate pairs $\pm \ii \cosh{2r_i}$.

\subsection{Bosonic partner mode construction}\label{sec:partnermode}
If we pick a single mode from a larger system in a Gaussian state $\ket{G,z}$, the single mode is in general entangled with the rest of the modes.
However, as implied by the standard form \eqref{eq:sta}, there exists a second mode, the so called partner mode, which shares all of the first mode's entanglement, \ie these two modes are in a product state with the rest of the system.

Recent works have developed formulae for the partner mode \cite{hotta_partner_2015,trevison_pure_2018,trevison_spatially_2019}. We here recast them in terms of the complex structure $J$ of the state, thus illuminating their structure, and generalizing them to arbitrary Gaussian states. In Appendix \ref{app:standardform-boson} we derive the partner mode construction in terms of the complex structure.

Let us assume that the two observables $\hat Q_A=x_a \hat\xi^a$ and $\hat P_A=k_a\hat\xi^a$ define a mode, \ie
\begin{align}
    \Omega^{ab} x_a k_b=1\,,
\end{align}
and that they yield the standard form of the mode according to \eqref{eq:sta}, \ie  for some $r>0$
\begin{align}
G^{ab} x_a x_b&=G^{ab}k_a k_b=\cosh 2r \,,\\
G^{ab} x_a k_b&=0 \,.
\end{align}
Then the partner mode of this mode is given by $\hat Q_{\bar{A}}=\bar{x}_a\hat\xi^a$ and $\hat P_{\bar{A}}=\bar{k}_a \hat\xi^a$ with
\begin{align}\label{eq:partnermodeformula}
\bar{x}_a&=\coth(2r) x_a+ \frac{1}{\sinh 2r}\left(J^\intercal\right)_a{}^c k_c\,,\\
\bar{k}_a&=- \coth(2r) k_a+  \frac{1}{\sinh 2r}\left(J^\intercal\right)_a{}^c x_c \,.
\end{align}

\subsubsection{Partner modes yield standard form for $G$}
This definition yields a normalized mode $\Omega^{ab} \bar{x}_a\bar{k}_b=1$, which commutes with the original mode and is correlated exactly as predicted by \eqref{eq:sta}, \ie as is straightforward to check using the identities  $(J^\intercal)^2=-\mathbb{1}$ and $G \Omega^{-1}G=-\Omega$, we have
\begin{align}
G^{ab}\bar{x}_a \bar{x}_b&=G^{ab}\bar{k}_a\bar{k}_b=\cosh2r\,,\\\
G^{ab}x_a\bar{x}_b&=-G^{ab}k_a\bar{k}_b=\sinh2r\,,\\
G^{ab}k_a \bar{x}_b&=G^{ab}x_a\bar{k}_b=G^{ab}\bar{x}_a\bar{k}_b=0\,.
\end{align}
Note  that repeating the partner mode construction gives back the original mode again, \ie the partner mode's partner is the original mode.

There is a close relationship between the correlation of other modes with the original mode, and the commutator of those other modes with the partner mode: Let $x'_a$ and $k'_a$ define another mode in the system, \ie  $\Omega^{ab}x'_a k'_b=1$. Then
\begin{align}
    \Omega^{ab}x'_a \bar{x}_b&=\coth(2r) \Omega^{ab}x'_ax_b+\frac{G^{ad} x'_a k_d}{\sinh 2r}\,, \\
    \Omega^{ab}k'_a \bar{x}_b&= \coth(2r)\Omega^{ab}k'_ax_b+\frac{G^{ad} k'_a k_d}{\sinh 2r} \,, \\
    \Omega^{ab}x'_a \bar{k}_b&= -\coth(2r) \Omega^{ab}x'_ak_b+\frac{G^{ad} x'_a x_d}{\sinh 2r} \,, \\
    \Omega^{ab}x'_a \bar{x}_b&= -\coth(2r)\Omega^{ab}k'_ak_b+\frac{G^{ad} k'_a x_d}{\sinh 2r} \,.
\end{align}
This means that if the mode commutes with the original mode $\Omega(x,x')=\Omega(x,k')=\Omega(k,x')=\Omega(k,k')=0$, then the commutator with the partner mode is proportional to the correlation with the original mode.
In particular, it follows that with respect to a basis of modes which contains the original mode and the partner mode as its first two modes, 
the covariance matrix $G^{ab}$ takes precisely the standard form.
\begin{align}\label{eq:G_standardform_partner}
	G\equiv\left(\begin{array}{cc|cc|c}
	\cosh{2r} & 0 & \sinh{2r} & 0 & 0 \\
	0 & \cosh{2r} & 0 & -\sinh{2r} & 0 \\
	\hline
	\sinh{2r} & 0 & \cosh{2r} & 0 & 0 \\
	0 & -\sinh{2r} & 0 & \cosh{2r} & 0 \\
	\hline
	0 & 0 & 0 & 0 &  \mathbb{1}_{N-2}
	\end{array}\right)
\end{align}
This shows that $\ket{G,z}$ is in a product state consisting of an entangled two-mode state (between original modes and its partner) and another Gaussian state describing the rest of the system. Note that the displacement vector $z^a$ does not affect the entanglement structure of the state.

\subsubsection{Unsqueezing the partner modes}
From the matrix representation of $G$ above we see that partner modes are in a pure two-mode squeezed state. Hence, they can be viewed as arising from squeezing two modes which are in a pure product state with each other. 
We obtain these two modes by acting with the inverse squeezing transformation on the two partner modes $(x_a\hat\xi^a,k_a\hat\xi^a)$ and $(\bar{x}_a\hat\xi^a,\bar{k}_a\hat\xi^a)$. This yields the  modes $(y_a\hat\xi^a,l_a\hat\xi^a)$ and $(z_a\hat\xi^a,m_a\hat\xi^a)$,
given by
\begin{align}\label{eq:unsqueezedmodes}
    y_a &= \cosh(r) x_a-\sinh (r) \bar{x}_a =  \frac{x_a- \left(J^\intercal\right)_a{}^c k_c }{2\cosh r}\,, \\ 
    l_a &= \cosh(r) k_a +\sinh(r) \bar{k}_a =  \frac{k_a +  \left(J^\intercal\right)_a{}^c x_c}{2\cosh r}\,, \\ 
    z_a &= -\sinh(r) x_a +\cosh(r) \bar{x}_a = \frac{x_a + \left(J^\intercal\right)_a{}^c k_c }{2\sinh r}\,, \\
    m_a &= \sinh(r) k_a +\cosh(r) \bar{k}_a = \frac{-k_a + \left(J^\intercal\right)_a{}^c x_c}{2\sinh r}\,.
\end{align}
Since the inverse squeezing is a symplectic transformation these \emph{unsqueezed} modes form a  normalized and commuting basis of the subsystem spanned by the two partner modes. In particular, with respect to them the covariance matrix is represented by the identity matrix $G^{ab}\equiv \id$.

\section{Entanglement extraction from quadratic bosonic systems}\label{sec:boson_extraction_wholesection}
In this section, we present the central result of this paper for bosonic systems. We begin by defining the entanglement extraction  procedure in Section~\ref{sec:bosonic_extraction_procedure}. Our proof then follows three steps:
\begin{enumerate}
    \item \textbf{Splitting the problem into two parts}\\
    In section~\ref{sec:energycostofEE}, we show that the problem of finding the  minimal energy cost can be broken up into first solving the problem for a system of two modes, and then optimizing how to embed these two modes used for entanglement extraction into the larger system.
    \item \textbf{Solving the two-mode problem}\\
    In Section~\ref{sec:twobosonicmodes}, we solve the two mode problem, \ie we answer the question: What is the minimal energy cost $\Delta E_{\min}$ of extracting two partner modes with %a given amount of 
    entanglement entropy $\Delta S$ from  a system consisting of two modes with a quadratic Hamiltonian? 
    For this, we first look at the case of a degenerate two-mode Hamiltonian  in Section~\ref{sec:degenerate_case}, before tackling the general case in Section~\ref{sec:two_mode_nondeg_boson}. The result is a general formula for the minimal energy cost as a function of the Hamiltonian's excitation energies $\epsilon_1$ and $\epsilon_2$.
    \item \textbf{Solving the embedding problem}\\
    In Section~\ref{sec:optimal_modes}, we derive how to select modes from a large system that optimize the energy cost. Not surprisingly, the lowest energy cost is achieved only if the two modes are chosen from the subsystem spanned by the two lowest energy eigenmodes of the Hamiltonian. Based on this insight, we construct the optimal partner modes and the associated excitation energies in the subsystem explicitly. Furthermore, we generalize the expressions and constructions to scenarios where the two partner modes are restricted to lie in a different two-mode subsystem.
\end{enumerate}

\subsection{Pure state entanglement extraction from bosonic modes}\label{sec:bosonic_extraction_procedure}
The general framework for entanglement extraction consists of two target systems and the source system $S$. Initially, the total state is assumed to be a product state
\begin{align}
    \rho_i=\sigma_1\otimes\tau_2\otimes\ket{\Psi}\bra{\Psi}_S\,,
\end{align}
where  $\ket{\Psi}_S$ is the initial state of the source system, and $\sigma_1$ and $\tau_2$ are the initial states of the first and second target system, respectively. Without loss of generality, the interaction between the source and target systems can be modeled by two unitary operations $U_1$ and $U_2$, which each couple one of the targets to the source. The final state of the two target systems is then
\begin{align}
    \rho_{12} &= \Tr_{S} \left( U_1 U_2\rho_{i} U_2^\dagger U_1^\dagger\right)\,.
\end{align}
In general, the entanglement content of this state depends on the type of source and target system and their realizable couplings. In particular, it also depends on the resources such as the energy which is available to implement the extraction.

The concrete type of source system we study in this section, is a system composed of bosonic modes with a quadratic Hamiltonian. This system is supposed to be in its ground state initially, \ie the state $\ket\Phi_S$ is Gaussian. Consequently, its entanglement structure can be analyzed using the formalism of partner modes introduced before. The two target modes are assumed to be single bosonic modes.

\begin{figure}[t]
\begin{center}
  \includegraphics[width=\linewidth]{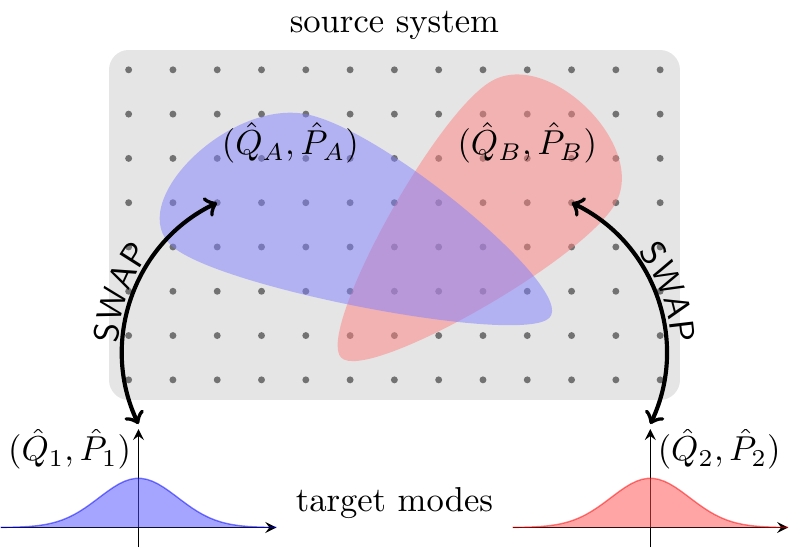}
\end{center}\vspace{-4mm}
\caption{General framework of entanglement extraction from a system of bosonic modes. Two target modes get entangled by swapping their states with a pair of entangled modes in the source system.}
\label{fig:entantlgement_extraction} 
\end{figure}

The model of interaction we use is that the target modes, with quadratures $(\hat Q_1,\hat P_1)$ and $(\hat Q_2,\hat P_2)$, each couple to one mode inside the source system, with quadratures $(\hat Q_A,\hat P_A)$ and $(\hat Q_B,\hat P_B)$. The unitary couplings $U_1$ and $U_2$ are assumed to swap the states of the  target modes and the source modes, \ie 
\begin{align}
\begin{split}
    U_1^\dagger \hat Q_1 U_1 &= \hat Q_A\,,\quad U_1^\dagger \hat Q_A U_1 = \hat Q_1\,,\\
    U_1^\dagger \hat P_1 U_1 &= \hat P_A\,,\quad\,\, U_1^\dagger \hat P_A U_1 = \hat P_1\,,\label{eq:Sswap}
\end{split}
\end{align}
and  $U_2$ swaps $(\hat Q_2,\hat P_2)$ and  $(\hat Q_B,\hat P_B)$, accordingly.

The only restriction imposed on the source modes is  that they commute according to
\begin{align}
    \comm{\hat Q_A}{\hat Q_B}\!=\!\comm{\hat Q_A}{\hat P_B}\!=\!\comm{\hat P_A}{\hat Q_B}\!=\!\comm{\hat P_A}{\hat P_B}\!=\!0.
\end{align}
This condition ensures that the entanglement of the two target modes was preexisting in the system, but not created by the couplings between target modes and system \cite{simidzija_general_2018}. Other than that, the modes can be chosen freely inside the source systems.

In particular, the source modes can be chosen to be partner modes, \ie $\hat Q_B=\hat Q_{\bar{A}}$ and $\hat P_B=\hat P_{\bar{A}}$ in the notation of Section~\ref{sec:partnermode}. As discussed, this implies that the ground state of the source system $S$ is a product state between the two modes and the rest of the source system
\begin{align}\label{eq:groundstate}
    \ket\Phi_S=\ket\psi_{A\bar{A}}\otimes\ket\phi_R\,,
\end{align}
where $\ket\psi_{A\bar{A}}$ is a two-mode squeezed state, and $\ket\phi_R$ is the ground state of the rest of the system. When the target modes and the source modes are swapped this puts the total system into the state 
\begin{align}\label{eq:finalstate}
    \rho_f= \ket\psi\bra\psi_{12} \otimes \sigma_A\otimes\tau_{\bar{A}}\otimes\ket\phi\bra\phi_R\,,
\end{align}
where now the target modes are in the entangled state $\ket{\psi}$, whereas the two source modes are in  the initial states of the target modes. The rest of the source system is unaffected by the  extraction and continues to be in its ground state $\ket{\phi_R}$.

Choosing partner modes in the source system is advantageous for two reasons: Firstly, it is the only way for the  target modes to end up in a pure state. If the system modes are not partners then after the swap the two target modes are correlated with the source system, \ie after tracing out the source system the partial state of target modes is mixed. Secondly, choosing the modes as partners maximizes the extracted entanglement. If the second party chooses a different mode, then the (mixed state) entanglement between the two modes is never larger than between the first mode and its partner. This is because the mixed state of the two modes is obtained by tracing out part of a larger state which included the  partner mode.

Motivated by this framework and these considerations we study the following concise question: If in the ground state of a quadratic Hamiltonian the state of two partner modes sharing entanglement entropy $\Delta S$ is replaced by a product state, by at least what amount $\Delta E$ does this increase the energy expectation value of the system.

While in the following we refer to $\Delta E$ as the energy cost of entanglement extraction,
in a realistic implementation of an extraction protocol additional energy costs are expected: For one, we do not take into account the energy content of the target modes outside of the source system. 
Furthermore, the implementation of the unitary interaction between targets and source may require  energy as well.
However, our reason to focus on the energy $\Delta E$  injected into the source system only, is that the other additional energy costs depend on the details and the  design of the specific extraction protocol implemented.

Another interesting aspect is the dynamics of the extraction process: In general, the source modes are not eigenmodes of the Hamiltonian, i.e., they have a non-trivial time evolution and their physical localization in the source system may change over time.
This dynamic needs to be taken into account in the design of the interaction of source and target modes. 
In the scope of this work, we assume that the swap of target and source modes happens within a time much shorter than the time scale of the time evolution of the target modes. Hence we neglect the dynamics and consider the interaction of target a
The vale modes to take part instantaneously. idity of this ind sourcdealization has to be reviewed with respect to any given implementation of an extraction protocol. We return to a discussion of this point at the end of the article.

\subsection{Energy cost of entanglement extraction}\label{sec:energycostofEE}
The energy cost of entanglement extraction, as introduced above, is given by the difference in expectation value of the source system's Hamiltonian before and after the extraction, \ie
\begin{align}
    \Delta E=\Tr\left(  \sigma_A\otimes\tau_{\bar A}\otimes\ket\phi\bra\phi_R \hat H- \ket\Psi\bra\Psi_S \hat H \right)\,.
\end{align}
The Hamiltonian $\hat H=\frac{1}{2}h_{ab}\hat\xi^a\hat\xi^b+f_a\hat\xi^a$ of the source system may be coupling all different modes of the system. Nevertheless, for the calculation of $\Delta E$,  only  the part acting on the two selected partner modes $A$ and $B=\bar{A}$ matter.

To see this, we split up the Hamiltonian into a sum of three terms
\begin{align}
\hat H=\hat H_{A\bar{A}}+\hat H_R+\hat H_{A\bar{A},R}
\end{align}
where $\hat H_{A\bar{A}}$ acts only on the two partner modes, $\hat H_R$ acts only on the rest of the source system's modes, and $\hat H_{A\bar{A},R}$ is the part which couples the partner modes and the rest of the system.\footnote{Put differently, the operators $\hat H_{A\bar{A}}$ and $\hat H_R$ are the restrictions of $\hat H$ onto the subspace spanned by the partner modes and the rest of the system, respectively. $\hat{H}_{A\bar{A},R}:=\hat{H}-\hat H_{A\bar{A}}-\hat H_R$ then contains all terms in $\hat H$ that act on both $A\bar{A}$ and $R$ simultaneously.}

The summand $\hat H_R$ does not contribute to the energy cost $\Delta E$, because  the rest of the system is in the same state before and after the extraction. Also the part $\hat H_{A\bar{A},R}$ does not contribute to $\Delta E$, because both the initial state \eqref{eq:groundstate} of the system and its final state \eqref{eq:finalstate} are product states with respect to the partition $A\bar{A}\oplus V$, hence
\begin{align}
 \Tr\left(  \sigma_A\otimes\tau_{\bar A}\otimes\ket\phi\bra\phi_R \hat H_{A\bar{A},R}\right)%=\Tr\left( \ket\Psi\bra\Psi_S \hat H_{A\bar{A},R} \right)
 =0\,.
\end{align}
Therefore, the energy cost of entanglement extraction is determined by $\hat{H}_{A\bar{A}}$ acting on the partner modes, only:
\begin{align}\label{eq:energycost}
 \Delta E=\Tr\left(  \sigma_A\otimes\tau_{\bar A}\hat H_{A\bar{A}}\right) -\Tr\left( \ket\psi\bra\psi_{A\bar{A}} \hat H_{A\bar{A}} \right) \,.
\end{align}
From this  follows which initial states of the target modes  minimize the energy cost: Because
\begin{align}\label{eq:product_in_energycost}
    \Tr\left(  \sigma_A\otimes\tau_{\bar A}\hat H_{A\bar{A}}\right) = \Tr\left(\sigma_A\hat H_A\right)+\Tr\left(\tau_{\bar A} \hat H_{\bar A}\right)\,,
\end{align}
the initial states $\sigma$ and $\tau$ of the target modes need to be the ground states of the restrictions $\hat H_A$ and $\hat H_{\bar A}$ of the Hamiltonian $\hat H$ onto the partner modes. These states are pure Gaussian states.

Equation \eqref{eq:energycost} effectively simplifies finding the minimal energy cost to a two-mode problem: If we know the minimal energy cost for the extraction of an amount $\Delta S$ of entanglement from a source system consisting of two modes, then applying this result to $\hat H_{A\bar A}$ and optimizing over different choices of $A$ yields the minimal cost for extracting $\Delta S$ from a larger system of $N$ modes.

\subsection{Energy cost for two bosonic modes}\label{sec:twobosonicmodes}
In this section, we analyze entanglement extraction from a  source system that consists of only two bosonic modes and has a quadratic Hamiltonian. We show that the minimal energy cost for the extraction from a pair of partner modes with a fixed amount of entanglement is a function only of the Hamiltonian's excitation energies. Our derivation  is constructive in the sense that we derive the symplectic transformation which maps the eigenmodes of the Hamiltonian to a pair of partner modes which minimize the energy cost.

Since the Hamiltonian %$\hat H=\frac{1}{2} h_{ab} \hat\xi^a\hat\xi^b+f_a$
is quadratic, we can choose a basis $\hat{\xi}^a\equiv(\hat q_1,\hat p_1,\hat q_2,\hat p_2)$ such that the quadratic part of the Hamiltonian $h_{ab}$  and the covariance matrix $G^{ab}$ of its ground state are represented by diagonal matrices $\mat{h}$ and $\mat G$ given by
\begin{align}\label{eq:h_diag}
	h_{ab}&\stackrel{1,2}\equiv 
	\mat h=\left(\begin{array}{cccc}
	\epsilon_1 & 0 & 0 & 0\\
	0 & \epsilon_1 & 0 & 0\\
	0 & 0 & \epsilon_2 & 0\\
	0 & 0 & 0 & \epsilon_2
	\end{array}\right)\,,\\
	G^{ab}&\stackrel{1,2}\equiv
	\mat G=\left(\begin{array}{cccc}
	1 & 0 & 0 & 0\\
	0 & 1 & 0 & 0\\
	0 & 0 & 1 & 0\\
	0 & 0 & 0 & 1
	\end{array}\right)\,. \label{eq:G12}
\end{align}
These modes $1$ and $2$ are not entangled and, therefore, cannot be used as modes $A$ and $\bar{A}$ for entanglement extraction. Instead, we need to find a pair of partner modes $(\hat Q_A,\hat P_A,\hat Q_{\bar{A}},\hat P_{\bar{A}})$ with respect to which $G^{ab}$ is represented in standard form \eqref{eq:G_standardform_partner} given by
\begin{align}
G^{ab}\stackrel{A,\bar{A}}\equiv\mat{\tilde G}=\begin{pmatrix}\cosh{2r}&0&\sinh{2r}&0\\0&\cosh{2r}&0&-\sinh{2r}\\ \sinh{2r}&0&\cosh{2r}&0\\ 0&-\sinh{2r}&0&\cosh{2r}\end{pmatrix}\label{eq:GAplusPARTNER}
\end{align}
for a fixed $r>0$, which we choose based on the amount of entanglement that we want to extract. There is an infinite number of choices for such $A$ and $\bar{A}$, but we will identify those that minimize the energy cost. This is, we find the symplectic transformations which map the energy eigenmodes $1$ and $2$ to those  partner modes $A$ and $\bar{A}$ which are optimal for entanglement extraction.

The two-mode squeezing transformation
\begin{align}
	\mat{M}_r=
	\left(\begin{array}{cccc}
	\cosh{r} & 0 & \sinh{r} & 0\\
	0 & \cosh{r} & 0 & -\sinh{r}\\ 
	\sinh{r} & 0 & \cosh{r} & 0 \\
	0 & -\sinh{r} & 0 & \cosh{r}
	\end{array}\right)
\end{align}
brings $\mathbf{G}$ into the form from~\eqref{eq:GAplusPARTNER} via
\begin{align}
    \mat{\tilde G}=\mat{M}_r \mat G\mat{M}^\intercal_r\,,
\end{align}
\ie it maps the eigenmodes to a pair of partner modes with squeezing parameter $r$. Under the symplectic transformation $\mathbf{M}_r$, the Hamiltonian matrix $\mathbf{h}$ transforms under the action of the inverse transform $\mathbf{M}^{-1}_r=\mathbf{M}_{-r}$, \ie
\begin{align}
    \mat{\tilde h}= \mat{M}_{-r}^\intercal \mat h\mat{M}_{-r}\,.
\end{align}
However, $\mat{M}_r$ is not the only transformation doing this: We can compose the squeezing $\mat{M}_r$ with any symplectic transformation $\mat{N}$  which leaves the covariance matrix invariant, \ie $\mat{G}=\mat{N} \mat{G} \mat{N}^\intercal$. This yields the most general form of symplectic transformation $\mat{T}_r=\mat{M}_r \mat{N}$ that transforms the modes $1$ and $2$ into $A$ and $\bar{A}$ for fixed $r$. All of these choices share the same amount of entanglement across $A$ and $\bar{A}$ because the covariance matrix takes the same form $\mat{\tilde G}=\mat{T}_r\mat{G} \mat{T}^\intercal_r$. However, the Hamiltonian matrix $\mathbf{h}$ is in general not left invariant by $\mat{N}$, \ie $\mat{h}\neq(\mat{N}^{-1})^\intercal \mat{h}\mat{N}^{-1}$ and thus the energy cost will be different.

\subsubsection{Degenerate two-mode Hamiltonian $(\epsilon_1\!=\!\epsilon_2)$}\label{sec:degenerate_case}
We first consider the case where $\epsilon:=\epsilon_1=\epsilon_2$. This case is the simplest to solve because the matrix representations $\mathbf{h}$ and $\mathbf{G}$ are proportional to each other in the eigenmode basis and, thus, in all bases. Indeed, any symplectic transformation $\mat N$ that leaves $\mat G$ invariant, \ie $\mat{N}\mat{G}\mat{M}^\intercal=\mat{G}$, also leaves $\mat{h}$ invariant, \ie $\mat{h}=(\mat{N}^{-1})^\intercal\mat{ h}\mat{N}^{-1}$. Therefore, for any pair of partner modes $A$ and $\bar{A}$ where the covariance matrix takes the standard form $\mat{\tilde G}$ for fixed $r$, we find
\begin{align}\label{eq:h_partnerbasis_degenerate}
	\mat{\tilde h}=\left(\begin{array}{cc|cc}
	\epsilon\cosh{2r} & 0 & -\epsilon\sinh{2r} & 0\\
	0 & \epsilon\cosh{2r} & 0 & \epsilon\sinh{2r}\\
	\hline
	-\epsilon\sinh{2r} & 0 & \epsilon\cosh{2r} & 0 \\
	0 & \epsilon\sinh{2r} & 0 & \epsilon\cosh{2r}
	\end{array}\right).
\end{align}

As discussed following equation \eqref{eq:product_in_energycost}, the entanglement extraction swaps the state of the two modes against the product state $\sigma_A\otimes\tau_{\bar A}$ of two Gaussian states $\sigma$ and $\tau$. Hence, denoting $\sigma_A=\ket{G_A'}\bra{G_A'}$ and $\tau_{\bar{A}}=\ket{G_{\bar{A}}'}\bra{G_{\bar{A}}'}$, with respect to partner mode basis, the covariance matrix after the extraction takes the block diagonal form
\begin{align}\label{eq:tildeGprime}
	\mat{\tilde G'}=
	\left(\begin{array}{c|c}
	\mat{\tilde G}'_A & 0 \\ [1mm]
	\hline \\ [-3mm]
	0 & \mat{\tilde G}'_{\bar{A}}
	\end{array}\right).
\end{align}
To minimize the energy cost, we choose the initial states of the target modes to be the ground states of the restricted single-mode Hamiltonians, \ie $\mat{\tilde{G}}'_A=\mat{\tilde{G}}'_{\bar{A}}=\id_2$.
In terms of this new covariance matrix, the energy cost \eqref{eq:energycost} of the entanglement extraction is
\begin{align}
    \Delta E&= \frac{1}{4}\Tr \left( \mat{\tilde h}^\intercal \mat{\tilde G'}\right)-\frac{1}{4}\Tr \left( \mat h^\intercal {\mat G}\right)\\
    &=\frac{1}{4}\Tr\left(\mat{\tilde h}_A\mat{\tilde G}_A+\mat{\tilde h}_{\bar{A}} \mat{\tilde G}_{\bar{A}}\right)-\epsilon,
\end{align}
where $\mat{\tilde h}_{A}=\mat{\tilde h}_{\bar{A}}=\epsilon \cosh{2r} \id_2$ are the diagonal blocks of $\mat{\tilde h}$ above.

With these initial states, the minimal energy cost for a pure two-mode squeezed state with entanglement entropy $\Delta S= s_b(\cosh{2r})$ (see \eqref{eq:Sformula}) is
\begin{align}\label{eq:energycost_degenerate}
    \Delta E_{\mathrm{min}}&= 2\epsilon \sinh^ 2{r}\,.
\end{align}
The final state of the source system after the swap is not an energy eigenstate anymore. Its energy variance is
\begin{align}
   \Sigma_E=\frac{\sqrt{\Tr\left(\mat{\tilde h}\mat{\tilde G}'\mat{\tilde h}\mat{\tilde G}'+\mat{\tilde h}\mat{\Omega}\mat{\tilde h}\mat{\Omega}\right)}}{4} =\frac{\epsilon \sinh{2r}}{\sqrt{2}}\,.
\end{align}

We notice that for the degenerate Hamiltonian with $\epsilon_1=\epsilon_2$, the energy cost and energy variance is actually the same for all partner modes with entanglement entropy $\Delta S$. In the next section, we will see that this is different for the case $\epsilon_1<\epsilon_2$.

\subsubsection{General two-mode Hamiltonian ($\epsilon_1\!<\!\epsilon_2$)}\label{sec:two_mode_nondeg_boson}
The case of a general Hamiltonian, with $\epsilon_1 \leq \epsilon_2$, is more involved to analyze because there exist symplectic transformations which leave the covariance matrix of the state invariant but do change the Hamiltonian matrix and vice versa.

Thus, before we apply the squeezing transformation $\mat{M}_r$, we can apply another symplectic transformation $\mat{N}=\exp\left(\theta \mat K\right)$ with\vspace{-2mm}
\begin{align}
   \mat{K}&=\left(\begin{array}{cccc}
0 & 0 & \cos{\phi} & -\sin{\phi}\\
0 & 0 & \sin{\phi} & \cos{\phi}\\
-\cos{\phi} & -\sin{\phi} & 0 & 0\\
\sin{\phi} & -\cos{\phi} & 0 & 0
\end{array}\right)\,,
\end{align}
which leaves $\mat{G}$ invariant, but possibly changes $\mat{h}$. Above general form can be derived by writing $\mat{N}$ as a matrix exponential $\mat{N}=\exp\left(\theta \mat K\right)$ of a generator $\mat{K}$ subject to the constraints $\mat{K}\mat{G}+\mat{G}\mat{K}^\intercal=0$ (leaving $\mat{G}$ invariant) and $\mat{K}\mat{\Omega}+\mat{\Omega}\mat{ K^\intercal}=0$ (being symplectic).

While  $\mat{N}=\exp\left(\theta \mat{K}\right)$ leaves $\mat{G}$ invariant, it changes $\mat{h}$. Consequently, $\mat{T}_r=\mat{M}_r\mat{N}$ provides the most general transformation that turns the covariance matrix from~\eqref{eq:G12} into $\mat{\tilde{G}}=\mat{T}_r\mat{G}\mat{T}_r^\intercal$ in the form of~\eqref{eq:GAplusPARTNER}.
Under this transformation, the Hamiltonian matrix transforms to
\begin{align}
	\mat{\tilde h}=(\mat{T}_r^{-1})^\intercal \mat{h}  \mat{T}_r^{-1}
	=\left(\begin{array}{c|c}
	\mat{h}_A & \mat{h}_{X}\\ [1mm]
	\hline \\ [-4mm]
	\mat{h}_X^\intercal & \mat{h}_{\bar{A}}
	\end{array}\right)\,, \label{eq:h_matrix_optimal_boson}
\end{align}
where the individual blocks are given by
\begin{widetext}
\begin{align}
	\mat{h}_A&=\left(\begin{array}{cc}
	\epsilon_+ \cosh{2 r}-\epsilon_- (\cos{2\theta}+\sin {2\theta}\cos{\phi} \sinh{2r}) & -\epsilon_- \sin{2\theta} \sin{\phi} \sinh{2r} \\
	-\epsilon_- \sin{2\theta} \sin{\phi} \sinh{2r} & \epsilon_+ \cosh{2 r}-\epsilon_- (\cos {2\theta}-\sin {2\theta}\cos{\phi}  \sinh{2r})
	\end{array}\right)\,,\\
	\mat{h}_{\bar{A}}& = \left(\begin{array}{cc}
	\epsilon_+ \cosh{2 r}+\epsilon_- (\cos{2\theta}-\sin {2\theta}\cos{\phi} \sinh{2r}) & \epsilon_- \sin{2\theta} \sin{\phi} \sinh{2r} \\
	\epsilon_- \sin{2\theta} \sin{\phi} \sinh{2r} & \epsilon_+ \cosh{2 r}+\epsilon_- (\cos {2\theta}+\sin {2\theta}\cos{\phi}  \sinh{2r})
	\end{array}\right)\,,\\
	\mat{h}_X&=\left(\begin{array}{cc}
	-\epsilon_+ \sinh{2r} +\epsilon_-\sin{2\theta}\cos{\phi}\cosh{2r} & -\epsilon_-\sin{2\theta}\sin{\phi}\cosh{2r}\\
	\epsilon_-\sin{2\theta}\sin{\phi}\cosh{2r} & \epsilon_+ \sinh{2r}+\epsilon_-\sin{2\theta}\cos{\phi}\cosh{2r}
	\end{array}\right)
\end{align}
\end{widetext}
with $\epsilon_\pm=\frac{1}{2}(\epsilon_2\pm\epsilon_1)$. The entanglement extraction swaps the state of the two modes against the initial state of the target modes. We denote its covariance matrix with respect $A$ and $\bar{A}$ by $\mat{\tilde G'}$ as in~\eqref{eq:tildeGprime}. The energy cost is then given by
\begin{align}
    \Delta E &=\frac{1}{4}\Tr\left(\mat{\tilde h}_A\mat{\tilde G'}_A+\mat{\tilde h}_{\bar{A}}\mat{\tilde G'}_{\bar{A}}\right)-\frac{\epsilon_1+\epsilon_2}2\,.
\end{align}
To minimize the energy  cost we choose the initial states to be the ground states of the single-mode Hamiltonians of $\mat{h_A}$ and $\mat{h_{\bar{A}}}$. This yields
\begin{align}\label{eq:boson_general_DeltaE}
	\Delta E =\frac{1}{2}(\epsilon_A+\epsilon_{\bar{A}}-\epsilon_1-\epsilon_2)\, ,
\end{align}
where the  excitation energies of the single-mode Hamiltonians are given by the symplectic eigenvalues 
of $\mat{h}_A$ and $\mat{h}_{\bar{A}}$, \ie
\begin{widetext}
\begin{align}\label{eq:boson_epsilonA}
	\epsilon_A&=\frac{1}{2}\sqrt{4\epsilon_+^2\cosh^2{2r}+\epsilon_-^2\left(3+\cos{4\theta}-(1-\cos{4\theta})\cosh{4r}\right)-8\epsilon_+\epsilon_-\cos{2\theta}\cosh{2r}}\,,\\
	\epsilon_{\bar{A}}&=\frac{1}{2}\sqrt{4\epsilon_+^2\cosh^2{2r}+\epsilon_-^2\left(3+\cos{4\theta}-(1-\cos{4\theta})\cosh{4r}\right)+8\epsilon_+\epsilon_-\cos{2\theta}\cosh{2r}}\,.
\end{align}
\end{widetext}
Since the energy cost is independent of  the parameter $\phi$,  we only need to compute the minimum with respect to $\theta$, in order to find partner modes which yield the minimal energy cost of entanglement extraction. 
$\Delta E$ only depends on $\theta$ through $\epsilon_A+\epsilon_{\bar{A}}$ and we find that the minimum is attained for $\theta=\frac{\pi}{4}+\frac{n\pi}{2}$ with $n\in\mathbb{Z}$. This value gives $\epsilon_A=\epsilon_{\bar{A}}=\frac{1}{2}\sqrt{\epsilon_1^2+\epsilon_2^2+2\epsilon_1\epsilon_2\cosh{4r}}$. Hence the minimal energy cost for the  extraction of a pure squeezed state with entanglement entropy $\Delta S=S(\cosh{2r})$ is given by
\begin{align}\label{eq:energycost_nondeg}
    \Delta E_{\mathrm{min}}&=\frac12\left(\sqrt{\epsilon_1^2+\epsilon_2^2+2\epsilon_1\epsilon_2\cosh{4r}}-\epsilon_1-\epsilon_2\right)\,.
\end{align}
The energy variance of the final system state is
\begin{align}
    \Sigma_E&=\sqrt{\frac{\epsilon_1^2\epsilon_2^2(1+\cosh{4r})}{\epsilon_1^2+\epsilon_2^2+2\epsilon_1\epsilon_2\cosh{4r}}}\,\sinh{2r}\,.
\end{align}

\begin{figure}[t]
\begin{center}
  \includegraphics[width=\linewidth]{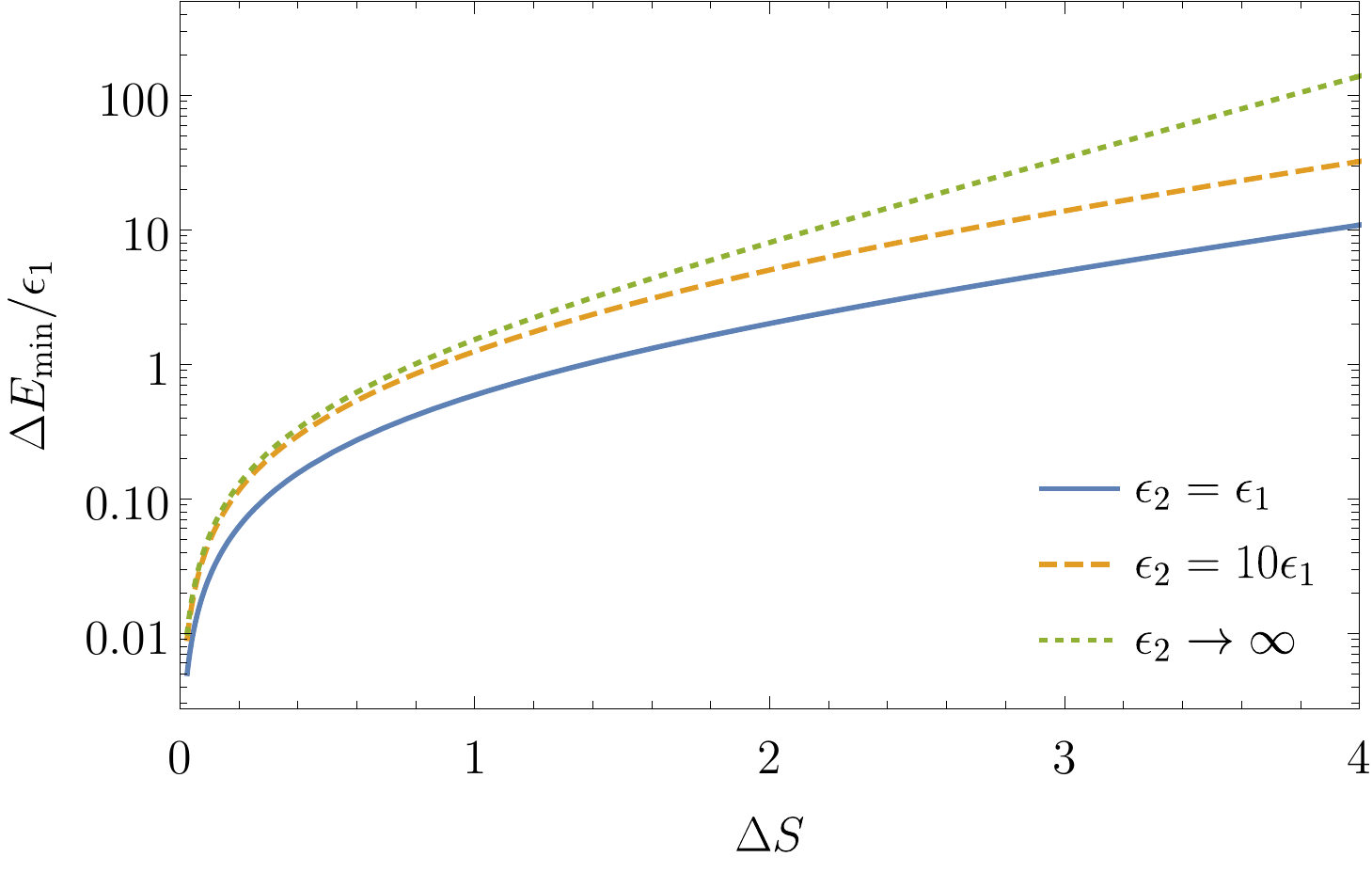}
\end{center}\vspace{-4mm}
\caption{Minimal energy cost $\Delta E_{\mathrm{min}}$ as a function of the extracted entanglement entropy $\Delta S$ from a bosonic, quadratic two-mode Hamiltonian with excitation energies $\epsilon_1$ and $\epsilon_2$.
} 
\label{fig:DeltaEvsDeltaSboson} 
\end{figure}

The minimal energy cost has a couple of interesting properties, which are seen in Figure~\ref{fig:DeltaEvsDeltaSboson}. First, we note that the expression above correctly reduces to the degenerate case for $\epsilon_2=\epsilon_1$.
For large amounts of extracted entanglement, the energy cost grows exponentially as 
\begin{align}
    \Delta E_{\mathrm{min}}\sim \ee^{2r}\sim e^{\Delta S}\quad\text{as}\quad r\to\infty\,,
\end{align}
%$\Delta E_{\mathrm{min}}\sim \ee^{2r}\sim e^{\Delta S}$ as $r\to\infty$, 
where we used the asymptotics $\Delta S\sim 2 r$ as $r\to\infty$. For small $r$, we find
\begin{align}
    \Delta E_{\mathrm{min}}\sim \frac{4\epsilon_1\epsilon_2}{\epsilon_1+\epsilon_2}r^2+\mathcal{O}(r^4)\quad\text{as}\quad r\to 0\,,
\end{align}
which should be seen in the context of the asymptotics $\Delta S\sim r^2 \left(1-\log \left(r^2\right)\right)$ as $r\to 0$.

Furthermore, the energy cost is a monotonously increasing function in both excitation energies. This means, for example that $\Delta E_{\mathrm{min}}$ increases when $\epsilon_2$ is increased and $\epsilon_1$ is kept fixed. However, there exists an upper bound on the energy cost in this case, which is determined by the lower excitation energy
\begin{align}
    \lim_{\epsilon_2\to\infty} \Delta E_{\mathrm{min}} = \epsilon_1 \left(\sinh 2r\right)^2\,.
\end{align}
Comparing the minimal energy cost  for a non-degenerate Hamiltonian $\Delta E_{\mathrm{min}, \epsilon_2>\epsilon_1}$ with the cost for a degenerate Hamiltonian $\Delta E_{\epsilon_2=\epsilon_1}$, we find that the ratio between them asymptotes to
\begin{align}
    \lim_{r\to\infty}\frac{\Delta E_{\mathrm{min},\epsilon_2>\epsilon_1}}{\Delta E_{\epsilon_2=\epsilon_1}}=\sqrt{\frac{\epsilon_2}{\epsilon_1}}\,.
\end{align}
Another interesting figure of merit is to look at the ratio between a given amount of extracted entanglement and the minimal energy cost necessary for its extraction $\Delta S/\Delta E_{\mathrm{min}}$.
Since $\Delta S=S(\cosh{2r})\sim (1-2\log(r)) r^2$ as $r\to0$, the ratio diverges in the limit of small amounts of extracted entanglement, \ie we have
\begin{align}
    \lim_{r\to0}\frac{\Delta S}{\Delta E_{\mathrm{min}}}=\infty\,.\label{eq:limDeltaSperDeltaE}
\end{align}
As the extracted entanglement increases the ratio is strictly decreasing, and in the limit of large entanglement $r\to\infty$ it decays exponentially. This means that it requires more energy to extract an amount of entanglement $\Delta S$ from a system once, than twice time the energy required to extract half the amount $\Delta S/2$.

\begin{figure}[t]
\begin{center}
  \includegraphics[width=\linewidth]{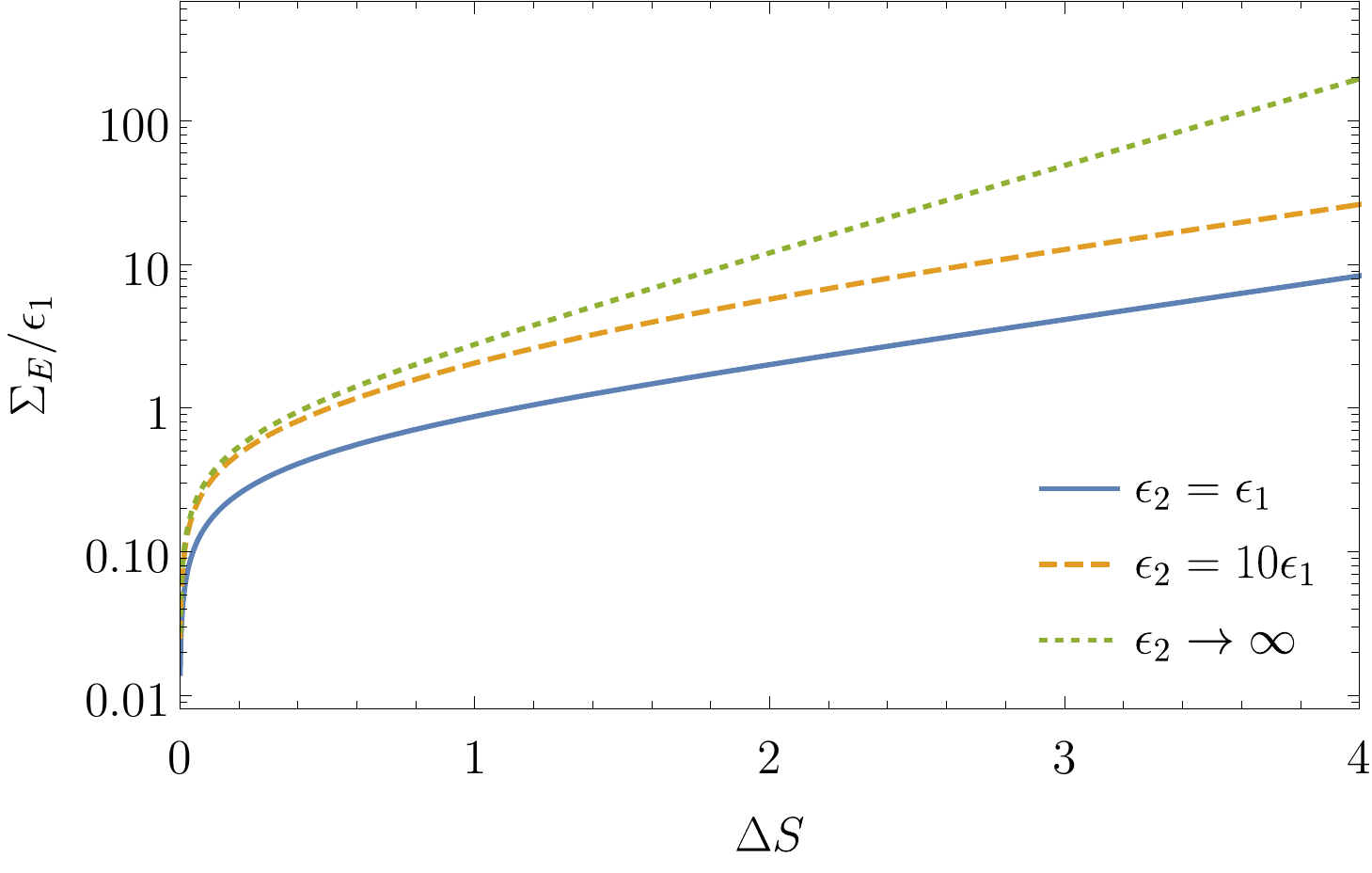}
\end{center}\vspace{-4mm}
\caption{
We show the energy variance $\Sigma_E$ as a function of the extracted entanglement entropy $\Delta S$ from a bosonic, quadratic two-mode Hamiltonian with excitation energies $\epsilon_1$ and $\epsilon_2$
}
\label{fig:SigmaEvsDeltaSboson} 
\end{figure}

\subsection{Optimal modes in large source systems}\label{sec:optimal_modes}
Having understood how to minimize the energy cost of entanglement extraction from a two-mode source system, allows us to also understand how to optimally extract entanglement from large source systems.

The partner modes from which entanglement can be extracted for the least possible energy cost are linear combinations of the two lowest energy eigenmodes of the system only.
No other choice of partner modes can have a lower energy cost, because the excitation energies $\epsilon_1$ and $\epsilon_2$ of the restricted two-mode Hamiltonian $\hat H_{A\bar A}$ are restricted by the excitation energies of $\hat H$ by
\begin{align}\label{eq:spec_restrictions_boson}
    \omega_1\leq\epsilon_1\leq\omega_{N-1}\,,\quad \omega_2\leq\epsilon_2\leq \omega_N
\end{align}
as discussed in Appendix~\ref{app:spectrumproof}.

These modes can be constructed according to above derivation. Take the two lowest energy eigenmodes $(\hat q_1,\hat p_1,\hat q_2, \hat p_2)$ of the system and mix them by acting on them with the transformation
\begin{align}
%\mat{N}_{\theta=\pi/4}=
\exp\left(\frac\pi4\mat K\right)=
\left(
\begin{array}{cccc}
 \frac{1}{\sqrt{2}} & 0 & \frac{\cos \phi}{\sqrt{2}} & -\frac{\sin \phi}{\sqrt{2}} \\
 0 & \frac{1}{\sqrt{2}} & \frac{\sin \phi}{\sqrt{2}} & \frac{\cos \phi}{\sqrt{2}} \\
 -\frac{\cos \phi}{\sqrt{2}} & -\frac{\sin \phi}{\sqrt{2}} & \frac{1}{\sqrt{2}} & 0 \\
 \frac{\sin \phi}{\sqrt{2}} & -\frac{\cos \phi}{\sqrt{2}} & 0 & \frac{1}{\sqrt{2}} \\
\end{array}
\right)
\end{align}
where $\phi$ can be chosen freely. The two obtained modes need to by squeezed by acting with $\mat M_r$ on them, in order to obtain the globally optimal partner modes of the source system.

In many scenarios the globally optimal modes may not be accessible to the experimenter, but they may be restricted to use a particular pair of partner modes $(\hat Q_A,\hat P_A,\hat Q_{\bar A},\hat P_{\bar A})$ as source modes. 
In this case, applying results of the previous section to the two-mode restriction $\hat H_{A\bar A}$ yields a lower bound on the energy cost of entanglement extraction, as discussed in Section \ref{sec:energycostofEE}.

However, the modes $(\hat Q_A,\hat P_A,\hat Q_{\bar A},\hat P_{\bar A})$ can only achieve the minimal energy cost $\Delta E_{\mathrm{min}}$ which  the spectrum of $\hat H_{A\bar A}$ allows for according to \eqref{eq:energycost_nondeg}, if they exactly correspond to the optimal modes with respect to $\hat H_{A\bar A}$ as we constructed them in the derivation of the previous section.

To check if this is the case it is sufficient to check if $\hat H_{A\bar A}$ acts symmetrically on the two modes $A$ and $\bar A$ because this is exactly what happens for the Hamiltonian matrix in \eqref{eq:h_matrix_optimal_boson} when $\theta=\frac\pi4+\frac{n\pi}2$ is chosen optimally.

The Hamiltonian matrix of $\hat H_{A\bar A}$ with respect to the basis $(\hat Q_A,\hat P_A,\hat Q_{\bar A},\hat P_{\bar A})$ is conveniently calculated using an inner product on $V^*$ which is induced by $h_{ab}$. To do this, we denote
\begin{align}
    \hat Q_A&=x_a\hat\xi^a\,,\qquad \hat P_A=k_a\hat\xi^a\,,\\
    \hat Q_{\bar{A}}&=\bar{x}_a\hat\xi^a,\qquad \hat P_{\bar{A}}=\bar{k}_a\hat\xi^a\,.
\end{align}
and assume, without loss of generality, that the modes are in standard form as in Section \ref{sec:partnermode}. With respect to this basis, the quadratic part of $\hat H_{A\bar A}$ is represented by 
\begin{align}\label{eq:hmatrix_partnermodes}
\left(\begin{array}{cc|cc} 
\check{h}(k,k)&-\check{h}(k,x) &\check{h}(k,\bar{k}) &-\check{h}(k,\bar{x}) \\
-\check{h}(x,k)& \check{h}(x,x)&-\check{h}(x,\bar{k}) & \check{h}(x,\bar{x})\\ \hline
\check{h}(\bar{k},k)&-\check{h}(\bar{k},x) &\check{h}(\bar{k},\bar{k}) &-\check{h}(\bar{k},\bar{x}) \\
-\check{h}(\bar{x},k)&\check{h}(\bar{x},x) & -\check{h}(\bar{x},\bar{k})&\check{h}(\bar{x},\bar{x})
\end{array}\right)\,,
\end{align}
with the inner product $\check h: V^*\times V^*\to\mathbb{R}$ with
\begin{align}
    \check{h}(w,u) = {h}_{ab}\Omega^{ac}\Omega^{bd}w_cu_d= {h}_{ab}G^{ac}G^{bd}w_cu_d\,,
\end{align}
induced by $h_{ab}$ of the full Hamiltonian. Denoting  $(J^\intercal w)_a=(J^\intercal)_a{}^b r_b$, the inner product satisfies
\begin{align}
\check{h}(J^\intercal w,J^\intercal u)= \check{h}(w,u)\,,
\end{align}
which just implies that $J^\intercal$ is orthogonal with respect to $\check{h}$ in the same way as $J$ is orthogonal with respect to $h$.
In particular, this implies $\check{h}(w,J^\intercal w)=0$, such that the matrix expressions simplify
\begin{align}
\check{h}(\bar{x},\bar{x})&= \check{h}(x,x)+ \Delta\,, \\
\check{h}(\bar{k},\bar{k})&=\check{h}(k,k)+ \Delta\,, \\
\check{h}(\bar{x},\bar{k})&=-\check{h}(x,k)\,, \\
\check{h}(x,\bar{x})&= \frac{\cosh(2r)\,\check{h}(x,x)+\check{h}(x,J^\intercal k)}{\sinh(2r)}\,, \\
\check{h}(k,\bar{k}) &= - \frac{\cosh(2r)\,\check{h}(k,k)\!+\!\check{h}(x,J^\intercal k)}{\sinh(2r)}\,, \\
\check{h}(x,\bar{k})&=-\frac{\cosh(2r)}{\sinh(2r)}\, \check{h}(x,k)=-\check{h}(k,\bar{x})\,,
\end{align}
where
\begin{align}\label{eq:Delta}
\Delta&=\frac{\check{h}(x,x)+\check{h}(k,k)+2\cosh(2r)\, \check{h}(x,J^\intercal k)}{\sinh^2(2r)}\,.
\end{align}
For $\Delta=0$, the Hamiltonian matrix is symmetric with respect to the two partner modes. This indicates that the partner modes are in fact optimal modes inside the two-mode subspace on which $\hat H_{A\bar A}$ acts and they achieve the minimal energy cost. Why this is the case is further illuminated by the following analysis of the case  where $\Delta$ does not vanish

To analyze the general case of $\Delta\neq0$ it is useful to reverse-engineer the derivation of the previous section. 
There, we showed that all partner modes in the subspace of $\hat H_{A\bar A}$ are connected to the eigenmodes of $\hat H_{A\bar A}$ through a transformation $\mat T_r=\mat M_r \mat N_\theta$ which is the product of an orthogonal transformation $\mat N_\theta=\exp\left(\theta\mat K\right)$ and a squeezing operation $\mat M_r$.

If we apply the first orthogonal transformation $\mat N_\theta$ to the energy eigenmodes we obtain a pair of modes which still diagonalize the covariance matrix $\mat G$. This pair of modes we also obtain by acting with the inverse squeezing transformation $\mat M_{-r}$ on the partner modes, \ie they are the unsqueezed modes $(y,l,z,m)$ as defined in \eqref{eq:unsqueezedmodes}.

Now we are able to express the matrix representing $\hat H_{A\bar A}$ in two different ways. On the one hand, we obtain it starting from its diagonal form $\mat h$ in \eqref{eq:h_diag} for the energy eigenmodes, and on the other hand we can express its entries using the form $\check h$ and the expressions \eqref{eq:unsqueezedmodes} for the unsqueezed modes.
\begin{align}
\begin{split}
   & \left(\mat N_\theta^{-1}\right)^\intercal \mat h \mat N_\theta^{-1}= \\
&
\left(\begin{array}{cc|cc} 
\check{h}(l,l)&-\check{h}(l,y) &\check{h}(l,m) &-\check{h}(l,z) \\
-\check{h}(y,l)& \check{h}(y,y)&-\check{h}(y,m) & \check{h}(y,z)\\ \hline
\check{h}(m,l)&-\check{h}(m,y) &\check{h}(m,m) &-\check{h}(m,z) \\
-\check{h}(z,l)&\check{h}(z,y) & -\check{h}(z,m)&\check{h}(z,z)
\end{array}\right)\,.
\end{split}
\end{align}
From this equation the values of the parameters $\phi$ and $\theta$ which yield the transformation $\mat N$ from the eigenmodes of $\hat H_{A\bar A}$ to the unsqueezed modes of $A$ and $\bar A$ as well as the spectrum of $\hat H_{A\bar A}$ can be determined. The latter is given by
\begin{align}
    \epsilon_+&=\frac{ \check h(x,x)+\check h(k,k) + \Delta  }{2 \cosh{2r}}\,, \\
    \epsilon_- &=\sqrt{\frac{4 \check h(x,k)^2+\left( \check h(x,x)-\check h(k,k)\right)^2}{4\sinh^2(2r)}+\frac{ \Delta^2}4}\,,
\end{align}
with $\epsilon_\pm=\frac12(\epsilon_2\pm\epsilon_1)$ as before. The parameter values are determined by
\begin{align}
    \tan\phi&=\frac{2\check h (x,k)}{\check h(x,x)-\check h(k,k)}\,,\\
    \tan2\theta&=\frac{\sqrt{4\epsilon_-^2-\Delta^2}}{\Delta}\,,
\end{align}
or, alternatively, $\Delta=2\epsilon_-\cos2\theta$.

With these parameters, we can calculate the minimal energy cost $\Delta E$ of entanglement extraction from a pair of given partner modes $(x,k\bar,x,\bar k)$ using \eqref{eq:boson_general_DeltaE}. The energy cost lies in the interval
\begin{align}
    \Delta E_{\theta=\pi/4}\leq\Delta E\leq\Delta E_{\theta=0}\,,
\end{align}
which is bounded by the minimal energy cost $\Delta E_{\min}$, achieved by $\theta=\frac\pi4+\frac{n\pi}2$,
\begin{align}
    \Delta E_{\theta=\pi/4}=
    \Delta E_{\mathrm{min}}=\sqrt{\epsilon_+^2\cosh^2{2r}-\epsilon_-^2\sinh^22r}
\end{align}
and the energy cost for a pair of partner modes obtained with the least favorable choice for the parameter $\theta=0+\frac{n\pi}2$, which is
\begin{align}
    \Delta E_{\theta=0}=2\epsilon_+\sinh^2{r}\,.
\end{align}

In fact, the last equality implies an upper bound on the energy cost of entanglement extraction with arbitrary partner modes of the system. Due to \eqref{eq:spec_restrictions_boson}, we have
\begin{align}\label{eq:upper_bound_cost_boson}
    \Delta E_{\theta=0}%\leq (\omega_N+\omega_{N-1})\sinh^2r
    \leq 2\omega_N \sinh^2r=:\Delta E_{\mathrm{max}}\,.
\end{align}
This means that for any pair of partner modes in the system the  energy cost $\Delta E$ (see \eqref{eq:energycost} and \eqref{eq:product_in_energycost}) of swapping their entangled state against the product of their one-mode restricted ground states is upper bounded by the minimal cost for a degenerate Hamiltonian (see \eqref{eq:energycost_degenerate}) with only the system's largest excitation energy $\omega_N$ in its spectrum.

Finally, we note that the analysis above allows us to construct from a given pair of partner modes $(\hat Q_A,\hat P_A,\hat Q_{\bar A},\hat P_{\bar A})$ the pair of modes which are optimal for extracting entanglement from the two-mode subspace of $\hat H_{A\bar A}$. For this we only need to correct for the mismatch between $\mat N_\theta$ corresponding to the given modes, and $\mat N_{\theta=\pi/4}$ which yields the optimal pair. So to obtain the ideal modes we only need to act with the transformation
\begin{align}
    \mat M_r \exp\left( (\frac\pi4-\theta)\mat K\right)\mat M_{-r}
\end{align}
on the modes $(\hat Q_A,\hat P_A,\hat Q_{\bar A},\hat P_{\bar A})$.

\section{Quadratic fermionic systems and their partner mode construction}\label{sec:fermion_review}
We review the most important definitions of quadratic fermionic systems and fermionic Gaussian states. More detailed reviews can be found in~\cite{eisert_colloquium_2010,plenio_entropy_2005,casini_entanglement_2009}, but we follow mostly the conventions introduced in~\cite{vidmar_entanglement_2017,hackl_aspects_2018}. In particular, defining fermionic Gaussian states using linear complex structures emphasizes the similarities to the bosonic case. The key difference between bosonic and fermionic Gaussian states is that $\Omega$ and $G$ exchange their roles. From a group theoretic perspective, this implies that the symplectic group $\mathrm{Sp}(2N,\mathbb{R})$ preserving $\Omega$ (governing commutation relations) is now replaced by the orthogonal group $\mathrm{O}(2N)$ preserving a positive definite metric $G$ (governing anti commutation relations).

\subsection{Classical and quantum theory}
As for bosons, we can formally define a classical phase space $V\simeq\mathbb{R}^{2N}$ and its dual $V^*\simeq\mathbb{R}^{2N}$ for a system with $N$ fermionic degrees of freedom. In contrast to bosons, $V$ and $V^*$ are now equipped with a positive-definite bilinear form $G^{ab}: V^*\times V^*\to\mathbb{R}$ and its inverse $G^{-1}_{ab}: V\times V\to\mathbb{R}$ satisfying $G^{ac}G^{-1}_{cb}=\delta^a{}_b$.

We can introduce an orthonormal basis of linear observables $q_i,p_i\in V^*\subset C^{\infty}(V)$, such that $G(p_i,p_j)=0$ and $G(q_i,q_j)=G(p_i,p_j)=\delta_{ij}$. As in the bosonic case,  $q_i$ and $p_i$ yield a component representation of a vector $v^a\in V$ as
\begin{align}
    v^a\equiv(q_1(v),p_1(v),\dots,q_N(v),p_N(v))\,.
\end{align}
With respect to this basis, the bilinear form $G^{ab}$ takes the standard form $G^{ab}\equiv\mathbb{1}$. We use $\xi^a: V\to V: v^a\mapsto \xi^a(v)=v^a$ in the same way as for bosons. However, the fermionic classical observables are not smooth functions on $V$, but rather the Grassmann algebra generated by $\xi^a$, \ie a formal series of anti-symmetrized powers in $\xi^a$.

When we quantized the bosonic system, we used the classical Poisson brackets encoded in $\Omega^{ab}$. For fermions, we implement the canonical anti-commutation relations
\begin{align}
    [\hat{\xi}^a,\hat{\xi}^b]_+=\hat{\xi}^a\hat{\xi}^b+\hat{\xi}^b\hat{\xi}^a=\{\xi^a,\xi^b\}_{+}=G^{ab}\,,%\{\xi^a,\xi^b\}_+=G^{ab}\,,
\end{align}
where $\{\xi^a,\xi^b\}_+=G^{ab}$ represents the classical fermionic Poisson brackets.

\subsection{Fermionic Gaussian states}
Fermionic Gaussian states appear under various names in the literature, ranging from fermionic vacua to Slater determinant states. There is no Gaussian wave function representations of them, but they share many properties with their bosonic counterparts. They form the set of ground states of quadratic Hamiltonians and also---in contrast to bosonic Gaussian states---their higher excited eigenstates. In the context of fermionic Bogoliubov transformations, they appear as vacua of annihilation operators satisfying fermionic anti commutation relations. Finally, their $n$-point correlation functions are completely determined from their 2-point correlation function. 

We label a fermionic Gaussian state $\ket{\Omega}$ by its covariance matrix
\begin{align}
    \Omega^{ab}=\bra{\Omega}\hat{\xi}^a\hat{\xi}^b-\hat{\xi}^b\hat{\xi}^a\ket{\Omega}\,,
\end{align}
which is antisymmetric, rather than symmetric as in the bosonic case. Note that there is no linear displacement for physical fermionic Gaussian states, \ie we require $\bra{\Omega}\hat{\xi}^a\ket{\Omega}=0$.
We can use $\Omega$ to construct a linear complex structure
\begin{align}
    J^a{}_b=\Omega^{ac}G^{-1}_{cb}=-G^{ac}\Omega^{-1}_{bc}\,,
\end{align}
that satisfies $J^2=-\mathbb{1}$ just as for bosons. Moreover, we have the condition $\frac{1}{2}(\delta^a{}_b+\ii J^a{}_b)\hat{\xi}^b\ket{\Omega}=0$, analogously to~\eqref{eq:Jproj}.

In contrast to bosons,  there is no displacement vector $z^a$ for fermions. Hence,  fermionic $n$-point functions are directly defined as
\begin{align}
    C_n^{a_1\cdots a_n}=\bra{\Omega}\hat{\xi}^{a_1}\dots\hat{\xi}^{a_n}\ket{\Omega}\,.
\end{align}
Apart from this, the definition of the two-point function is unchanged:
\begin{align}
    C_2^{ab}=\frac{1}{2}(G^{ab}+\ii\Omega^{ab})\,.
\end{align}
Wick's theorem applies  in almost the same way
\begin{align}
    C^{a_1\cdots a_n}_{n}&=\sum\text{(contractions of $C_2^{ab}$)}\\
    &=\sigma_{(a_1,\dots,a_n)}C^{a_1a_2}_2\cdots C^{a_{n-1}a_{n}}_2+\cdots\,,
\end{align}
however, each contraction enters with a positive or negative sign, depending on the parity $\sigma_{(a_1,\dots,a_n)}\in\{-1,1\}$ of its index ordering. This subtlety needs to be taken into account carefully when applying  Wick's theorem to fermions.

\paragraph{Example.} Looking at the ground state $\ket{\Omega}$ of the fermionic oscillator $\hat{H}=\frac{\ii}{2}h_{ab}\hat{\xi}^a\hat{\xi}^b=\ii \hat{q}\hat{p}$. The covariance matrix $\Omega^{ab}$ and the metric $G^{ab}$ with respect to the basis $\hat{\xi}^a\equiv(\hat{q},\hat{p})^\intercal$ are given by
\begin{align}
    \Omega^{ab}\equiv\begin{pmatrix}
    0 & 1\\
    -1 & 0
    \end{pmatrix}\quad\text{and}\quad G^{ab}\equiv\begin{pmatrix}
    1 & 0\\
    0 & 1
    \end{pmatrix}\,.
\end{align}
Just as in the bosonic case, we can construct the linear complex structure
\begin{align}
    J^a{}_b=\Omega^{ac}G^{-1}_{cb}\equiv\begin{pmatrix}
    0 & 1\\
    -1 & 0
    \end{pmatrix}\,.
\end{align}
With this, the fermionic projector becomes
\begin{align}
    \frac{1}{2}(\delta^a{}_b+\ii J^a{}_b)\hat{\xi}^b\equiv\frac{1}{2}\begin{pmatrix}\hat{q}+\ii \hat{p}\\\hat{p}-\ii \hat{q}\end{pmatrix}=\frac{1}{\sqrt{2}}\begin{pmatrix}\hat{a}\\-\ii \hat{a}\end{pmatrix}\,,
\end{align}
\ie the condition $\frac{1}{2}(\delta^a{}_b+\ii J^a{}_b)\hat {\xi}^b\ket{\Omega}=0$ is the same as for bosons. The projector selects the subspace spanned by complex linear combinations of annihilation operators $\hat{a}=\frac{1}{\sqrt{2}}(\hat{q}+\ii\hat{p})$.

\subsection{Quadratic fermionic Hamiltonians}
We are considering the most general quadratic fermionic Hamiltonian given by
\begin{align}
    \hat{H}=\frac{\ii}{2}h_{ab}\hat{\xi}^a\hat{\xi}^b\,,
\end{align}
where  $h_{ab}$ is antisymmetric and non-degenerate. Note that $\hat{H}$ is always bounded from below.

The ground state of $\hat H$ is given by a Gaussian state $\ket{\Omega}$. As for bosons, we can compute the ground state covariance matrix from $h_{ab}$:
\begin{itemize}
    \item \textbf{Computation of $\Omega^{ab}$}\\
    Similar to~\eqref{eq:groundstate}, we compute
    \begin{align}\label{eq:ferm_groundstate}
        \Omega^{ab}=K^a{}_c|K^{-1}|^c{}_dG^{db}
    \end{align}
    with $K^a{}_b=G^{ac}h_{cb}$.
\end{itemize}

The energy expectation value $E=\bra{\Omega'}\hat H\ket{\Omega'}$ and its variance $\Sigma_E$ for an arbitrary Gaussian state $\ket{\Omega'}$ can be efficiently computed as \begin{align}
    E&=-\frac{1}{4}h_{ab}{\Omega'}^{ab}=-\frac{\Tr(h {\Omega'}^\intercal)}{4}\,,\\
    \Sigma_E&=\frac{\sqrt{-\Tr(hGhG+h {\Omega'}h {\Omega'})}}{4}\,.
\end{align}

Again, we note that the symplectic form $\Omega^{ab}$ and the positive definite metric $G^{ab}$ switched their roles with respect to the bosonic formulae in Section \ref{sec:boson_review}.

\paragraph{Example.} We consider the harmonic oscillator of a single fermionic mode $\hat{H}=\frac{\ii}{2}h_{ab}\hat{\xi}^a\hat{\xi}^b=\ii\hat{q}\hat{p}$. The bilinear form $h$ with respect to $\hat{\xi}^a\equiv(\hat{q},\hat{p})$ is
\begin{align}
    h_{ab}\equiv\begin{pmatrix}
    0 & \omega\\
    -\omega & 0
    \end{pmatrix}\,,
\end{align}
which gives via~\eqref{eq:ferm_groundstate} rise to
\begin{align}
    K^a{}_b\equiv\begin{pmatrix}
    0 & \omega\\
    -\omega &0
    \end{pmatrix}\,\,\Rightarrow\,\,\Omega^{ab}\equiv \begin{pmatrix}
    0 & 1\\
    -1 & 0
    \end{pmatrix}\,,
\end{align}
where we used $|K^{-1}|=\id/\omega$, just like in the bosonic case. The ground state energy is given by $E=\bra{\Omega}\hat{H}\ket{\Omega}=-\omega/2$ and expressing the Hamiltonian in creation and annihilation operators gives $\hat{H}=\frac{\omega}{2}(\hat{a}^\dagger\hat{a}-1)$.

\subsection{Entanglement of fermionic Gaussian states}
Given a fermionic Gaussian state $\ket{\Omega}$ and a system decomposition into two even dimensional orthogonal complements $V=A\oplus B$, we can always choose an orthonormal basis $(q_1^A,p_1^A,\dots,p_{N_A}^A,q^A_{N_A},q_1^B,p_1^B,\dots,p_{N_B}^B,q^B_{N_B})$, such that the covariance matrix takes the standard form derived in appendix~\ref{app:standardform-fermions}
\begin{widetext}
\begin{align}
	\Omega\equiv\left(\begin{array}{ccc|cccccc}
	\mat{cos}_1 & \cdots & 0 & \mat{sin}_1 & \cdots & 0 & 0 &\cdots & 0\\
	\vdots & \ddots & \vdots & \vdots & \ddots & \vdots & \vdots & \ddots & \vdots \\
	0 & \cdots & \mat{cos}_{N_A} & 0 & \cdots &\mat{sin}_{N_A} & 0 & \cdots & 0\\
	\hline
	-\mat{sin}_1 & \cdots & 0 & \mat{cos}_1 & \cdots & 0 & 0 &\cdots & 0\\
	\vdots & \ddots & \vdots & \vdots & \ddots & \vdots & \vdots & \ddots & \vdots\\
	0 & \cdots & -\mat{sin}_{N_A} & 0 & \cdots &\mat{cos}_{N_A} & 0 &\cdots & 0\\
	0 & \cdots & 0 & 0 & \cdots & 0 & \mathbb{A}_2 & \cdots & 0\\
	\vdots & \ddots & \vdots & \vdots & \ddots & \vdots & \vdots & \ddots & \vdots \\
	0 & \cdots & 0 & 0 & \cdots & 0 & 0 & \cdots & \mathbb{A}_2
	\end{array}\right)\,,\label{eq:sta_ferm}
\end{align}
\end{widetext}
where $\mathbb{A}_2$, $\mat{cos}_i$ and $\mat{sin}_i$ are given by
\begin{align}
    \mathbb{A}_2 &=\left(\begin{array}{cc}
	0 & 1\\
	-1 & 0
	\end{array}\right)\,,\\
	\mat{cos}_i &=\left(\begin{array}{cc}
	0 & \cos{2r_i}\\
	-\cos{2r_i} & 0
	\end{array}\right)\,,\\
	\mat{sin}_i &=\left(\begin{array}{cc}
	0 & \sin{2r_i}\\
	\sin{2r_i} & 0
	\end{array}\right)\,.\label{eq:standardform-fermions}
\end{align}
This decomposition is analogous to the bosonic one in~\eqref{eq:sta}, where the mode $(q_i^A,p_i^A)$ is entangled with the mode $(q_i^B,p_i^B)$. The parameters $r_i$ can always be chosen to lie in the interval $[0,\pi/4]$. Again, the amount of entanglement is encoded in the parameter $r_i$, but now given by~\cite{banuls_entanglement_2007}
\begin{align}
    S_A(\ket{\Omega})=\sum^{N_A}_{i=1}s_f(\cos{2r_i})
\end{align}
with $s_f(x)=-\frac{1+x}{2}\log\left(\frac{1+x}{2}\right)-\frac{1-x}{2}\log\left(\frac{1-x}{2}\right)$. 
As for bosons, there exists a compact formula to express the entanglement entropy in terms of the linear complex structure $J$, namely
\begin{align}
    \textstyle S_A(\ket{G,z})=-\mathrm{Tr}\left[\left(\frac{\mathbb{1}_A+\ii[J]_A}{2}\right)\log\left(\frac{\mathbb{1}_A+\ii[J]_A}{2}\right)\,\!\right]\,,\label{eq:SJfermions}
\end{align}
which closely resembles~\eqref{eq:SJbosons}, but where  the modulus is replaced by an overall minus sign.

\subsection{Fermionic partner mode construction}
The standard form of the fermionic covariance matrix introduced above implies that also for a system of fermionic modes in an overall pure state, each single mode has a unique partner with which it is entangled.

In fact the partner mode for bosonic modes can be almost directly carried over to the fermionic case. Only the hypergeometric functions need to be replaced by their trigonometric analogues. If $(x_a,k_a)$ define a mode in standard form
\begin{align}
    G^{ab}x_ax_b=G^{ab}k_ak_b=1\,, \\
    G^{ab}x_ak_b=0\,,\\
    \Omega^{ab}x_ak_b=\cos2r\,,
\end{align}
then its partner mode is given by
\begin{align}\label{eq:partnermode_formula_ferm}
    \bar x_a&=\cot(2r) x_a+\frac1{\sin2r} \left(J^\intercal\right)_a^c k_c\,, \\
    \bar k_a&= -\cot(2r) k_a+\frac1{\sin2r}\left(J^\intercal\right)_a^c x_c\,.
\end{align}
This means that the partner mode is in standard form itself
\begin{align}
G^{ab}\bar{x}_a\bar{x}_b=G^{ab}\bar{k}_a \bar{k}_b=1\,,\\
G^{ab}\bar x_a\bar k_b=0\,,\\
\Omega^{ab}\bar{x}_a \bar{k}_k=\cos 2r\,,
\end{align}
and the correlations with the original mode are as required by the standard form
\begin{align}
G^{ab}x_a \bar{x}_b=G^{ab}k_a \bar{k}_b=G^{ab}x_a \bar{k}_b= 0\,\\
\Omega^{ab} x_a \bar{k}_b = \Omega^{ab} k_a \bar{x}_b=\sin2r\,.
\end{align}

Just as in the bosonic case, also for fermionic modes the partner of the partner mode is the original mode itself.
Also, we observe that any other mode $(x',k')$ which anti-commutes with $(x,k)$, has an anti-commutatator with the partner mode proportional to the covariance with the first mode, \ie
\begin{align}
    G^{ab}x'_a\bar x_b= \cot2r\, G^{ab} x'_a x_b-\frac1{\sin2r} \Omega^{ab}x'_a k_b\,.
\end{align}
This means that any mode which anticommutes with both partner modes is not correlated with them. In particular, a basis of modes which contains the two partner modes brings the covariance matrix of the state into block diagonal form
\begin{align}
    \Omega\equiv \left(\begin{array}{c|c}
         \mat{\tilde \Omega}&0  \\ \hline
         0&\mat\Omega_R 
    \end{array}\right)
\end{align}
where
\begin{align}\label{eq:OmegaAplusPARTNER}
    \mat{\tilde\Omega}=\left( \begin{array}{cccc}
         0&\cos 2r &0 &\sin2r  \\ -\cos2r&0&\sin2r\\
         0&-\sin2r&0&\cos2r\\ -\sin2r&0&-\cos2r
    \end{array}\right)
\end{align}
takes the standard form of \eqref{eq:sta_ferm}.

Also fermionic partner modes can be unsqueezed to find a pair of modes $(y,l)$ and $(z,m)$ which are in a pure product state and from which $(x,k)$ and $(\bar x,\bar k)$ are obtained by squeezing. These modes are given by:
\begin{align}
    y_a&= \cos(r) x_a-\sin(r)\bar x_a=\frac{  x_a-(J^\intercal)_a{}^b k_b }{2\cos r}\,, \\
    l_a&= \cos(r) k_a+\sin(r) \bar k_a=\frac{ k_a+(J^\intercal)_a{}^b x_b }{2\cos r} \,,\\
    z_a&= \sin(r) x_a +\cos(r) \bar x_a=\frac{ x_a+(J^\intercal)_a{}^b k_b }{2 \sin r}\,,\\
    m_a&=-\sin(r) k_a +\cos(r)\bar k_a=\frac{ -k_a+(J^\intercal)_a{}^b x_b }{2 \sin r}\,.
\end{align}

\section{Entanglement extraction from quadratic fermionic systems}\label{sec:fermion_extraction}
In this section, we analyze the minimal energy cost of pure state entanglement extraction from the ground state of a system of fermionic modes with a quadratic Hamiltonian. 
Our approach is motivated by the same framework as discussed for bosonic modes in Section \ref{sec:boson_extraction_wholesection}. Again, after reviewing the extraction procedure in section~\ref{sec:fermionic_extraction_procedure}, our proof follows exactly the same three steps, split over the sections~\ref{sec:fermion_proof1},~\ref{sec:twofermionicmodes} and~\ref{sec:fermion_proof3}.

\subsection{Pure state entanglement extraction from bosonic modes}\label{sec:fermionic_extraction_procedure}
We assume that the source system consists of fermionic modes which are in the ground state $\ket\Psi_S$ of an arbitrary quadratic Hamiltonian $\hat H$, which is Gaussian.
Within the source system, we select two anti-commuting source modes from which we want to extract entanglement.
For this we assume that it is possible to swap the state of the two source modes onto a pair of target modes outside of the system.
We can always use Gaussian evolution, \ie quadratic coupling terms, to swap the entangled Gaussian state in the source system with the unentangled Gaussian state in the target system.
The target modes can be prepared in an arbitrary product state. Hence the joint initial state of target modes and source system is of the form
\begin{align}
    \rho_i=\sigma_1\otimes \tau_2\otimes \ket\Psi\bra\Psi_S.
\end{align}

Just as in the bosonic case, also for fermionic modes the partner mode structure of entanglement of Gaussian states implies that the only way to extract a pure entangled state in this way is to chose the two source modes to be partner modes which we denote by $A$ and $\bar A$. 
This is also the optimal choice in the sense that any entanglement measure of one mode with with a second non-partner mode is never larger than of the mode with its partner.

The ground state of the source system is a product state $\ket\Psi=\ket\psi_{A\bar A}\otimes \ket\phi_R$ between the partner modes and the remainder $R$ of the system.
After the entanglement extraction, \ie after the states of the partner modes and the target modes have been swapped, the source system is in the state
\begin{align}
    \rho'_S=\sigma_A\otimes\tau_{\bar A}\otimes \ket\phi\bra\phi_R
\end{align}
where now the partner modes are in the product state given by the initial state of the target modes.

\subsection{Energy cost of entanglement extraction}\label{sec:fermion_proof1}
By the minimal energy cost of entanglement extraction we are referring to the minimum value by which the expectation value of $\hat H$ is increased, \ie by how much energy on average is injected into the system when the modes are swapped.
Since we have used the same notation, the discussion and formulae from Section~\ref{sec:energycostofEE} apply equally to the fermionic case, as considered here. In particular, the energy cost for the extraction of a fermionic pair of partner modes is given by
\begin{align}\label{eq:energycostfermion}
    &\Delta E=\Tr\left(\sigma_A\otimes\tau_{\bar A} \hat H_{A\bar A} -\ket\psi\bra\psi_{A\bar A}\hat H_{A\bar A}\right)\,.
\end{align}
It results from the restriction $\hat H_{A\bar A}$ of the full Hamiltonian onto the space spanned by the two partner modes. 
Thus, again the problem of finding the minimal energy cost reduces to a two-mode problem.

Also as before, we find that the final energy expectation value of the system is
\begin{align}
\Tr\left(\sigma_A\otimes\tau_{\bar A} \hat H_{A\bar A}\right)=\Tr\left(\sigma_A\hat H_A\right)+\Tr\left(\sigma_{\bar A}\hat H_{\bar A}\right)
\end{align}
which  implies that in order to minimize the energy cost, the target modes need to be initialized in the ground states of the single-mode restrictions $\hat H_A$ and $\hat H_{\bar A}$ of $\hat H$ onto the individual partner modes.
Also for fermionic modes  these states are Gaussian states, as discussed in the previous section.
Thus, we can perform our entire analysis within the framework of Gaussian states.

\subsection{Energy cost for two fermionic modes}\label{sec:twofermionicmodes}
This section presents the fermionic analogue to \ref{sec:twobosonicmodes}. We analyze entanglement extraction from a source system that consists of only two fermionic modes and has a quadratic Hamiltonian.

Since the Hamiltonian %$\hat H=\frac{\ii}{2} h_{ab} \hat\xi^a\hat\xi^b$
is quadratic, we can choose a basis $\hat{\xi}^a\equiv(\hat q_1,\hat p_1,\hat q_2,\hat p_2)$ of energy eigenmodes such that the quadratic part of the Hamiltonian $h_{ab}$  and the covariance matrix $\Omega^{ab}$ of its ground state are both represented by block diagonal matrices $\mat{h}$ and $\mat \Omega$ given by
\begin{align}
	h_{ab}&\stackrel{1,2}\equiv 
	\mat h=\left(\begin{array}{cccc}
	0 & \epsilon_1 & 0 & 0\\
	-\epsilon_1 & 0 & 0 & 0\\
	0 & 0 & 0 & \epsilon_2\\
	0 & 0 & -\epsilon_2 & 0
	\end{array}\right)\,,\\
	\Omega^{ab}&\stackrel{1,2}\equiv
	\mat{\Omega}=\left(\begin{array}{cccc}
	0 & 1 & 0 & 0\\
	-1 & 0 & 0 & 0\\
	0 & 0 & 0 & 1\\
	0 & 0 & -1 & 0
	\end{array}\right)\,.
\end{align}
The energy eigenmodes are not entangled and, therefore, cannot be used for entanglement extraction. Instead, we need to find a pair of partner modes $(\hat Q_A,\hat P_A,\hat Q_{\bar{A}},\hat P_{\bar{A}})$ with respect to which $\Omega^{ab}$ is represented by the matrix corresponding to the standard form \eqref{eq:OmegaAplusPARTNER}
\begin{align}
&\Omega^{ab}\stackrel{A,\bar{A}}\equiv
\mat{\tilde \Omega}
\end{align}
for some fixed $r\in[0,\pi/4]$, which depends on the amount of entanglement we want to extract. Among all choices of partner modes $A$ and $\bar{A}$, we want to identify those that minimize the energy cost associated to the entanglement extraction. Consequently, we will search for the orthogonal transformation which maps the energy eigenmodes $1$ and $2$ to those optimal partner modes $A$ and $\bar{A}$ for entanglement extraction.

The two-mode squeezing transformation
\begin{align}
	\mat{M}_r=
	\left(\begin{array}{cccc}
	\cos{r} & 0 & \sin{r} & 0\\
	0 & \cos{r} & 0 & -\sin{r}\\ 
	-\sin{r} & 0 & \cos{r} & 0 \\
	0 & \sin{r} & 0 & \cos{r}
	\end{array}\right)
\end{align}
brings $\mathbf{\Omega}$ into the form from~\eqref{eq:OmegaAplusPARTNER} via
\begin{align}
    \mat{\tilde \Omega}=\mat{M}_r \mat \Omega\mat{M}^\intercal_r\,,
\end{align}
\ie it maps the eigenmodes to a pair of partner modes with squeezing parameter $r$. Under the orthogonal transformation $\mathbf{M}_r$, the Hamiltonian matrix $\mathbf{h}$ transforms under the action of the inverse transform $\mathbf{M}^{-1}_r=\mathbf{M}_{-r}$, \ie
\begin{align}
    \mat{\tilde h}= \mat{M}_{-r}^\intercal \mat h\mat{M}_{-r}\,.
\end{align}
However, $\mat{M}_r$ is not the only transformation doing this: We can compose the squeezing $\mat{M}_r$ with any orthogonal transformation $\mat{N}$  which leaves the covariance matrix invariant, \ie $\mat{\Omega}=\mat{N} \mat{\Omega} \mat{N}^\intercal$, to obtain the most general transformation $\mat{T}_r=\mat{M}_r \mat{N}$ that transforms the modes $1$ and $2$ into $A$ and $\bar{A}$ for fixed $r$. All of these choices share the same amount of entanglement across $A$ and $\bar{A}$ because the covariance matrix $\Omega$ takes the same form $\mat{\tilde \Omega}=\mat T_r \mat\Omega \mat T_r^\intercal$ with respect to them. However, the Hamiltonian matrix $\mathbf{h}$ is in general not left invariant by $\mat{N}$, \ie $\mat{h}\neq(\mat{N}^{-1})^\intercal \mat{h}\mat{N}^{-1}$ and thus the energy cost can be different.

\subsubsection{Degenerate two-mode Hamiltonian $(\epsilon_1\!=\!\epsilon_2)$}
We first consider the case where $\epsilon:=\epsilon_1=\epsilon_2$. This case is the simplest to solve because the matrix representations $\mathbf{h}$ and $\mathbf{\Omega}$ are proportional to each other in the eigenmode basis and, thus, in all bases. Indeed, any orthogonal transformation $\mat N$ that leaves $\mat \Omega$ invariant, \ie $\mat{N}\mat{\Omega}\mat{M}^\intercal=\mat{\Omega}$, also leaves $\mat{h}$ invariant, \ie $\mat{h}=(\mat{N}^{-1})^\intercal\mat{ h}\mat{N}^{-1}$. Therefore, for any pair of partner modes $A$ and $\bar{A}$ where the covariance matrix takes the standard form $\mat{\tilde \Omega}$ for fixed $r$, we find
\begin{align}\label{eq:h_ferm_partnerbasis_degenerate}
	\mat{\tilde h}=\left(\begin{array}{cc|cc}
	0 & \epsilon\cos{2r}  & 0 & -\epsilon\sin{2r}\\
	-\epsilon\cos{2r} & 0 & -\epsilon\sin{2r} & 0\\
	\hline
	0 & \epsilon\sin{2r} & 0 & \epsilon\cos{2r}\\
	\epsilon\sin{2r} & 0 & -\epsilon\cos{2r} & 0
	\end{array}\right).
\end{align}
As discussed before, the entanglement extraction swaps the  state of the chosen partner modes against a product state $\sigma_{\bar{A}}\otimes\tau_{\bar{A}}$.
This brings the covariance matrix into block diagonal form with respect to the partner modes:
\begin{align}
	\mat{\tilde\Omega}'=
	\left(\begin{array}{c|c}
	\mat{\tilde \Omega}_A & 0 \\ [1mm]
	\hline \\ [-3mm]
	0 & \mat{\tilde \Omega}_{\bar{A}}
	\end{array}\right).
\end{align}
In terms of this new covariance matrix, the energy cost \eqref{eq:energycostfermion} of the entanglement extraction is
\begin{align}
    \textstyle\Delta E&=-\frac{1}{4}\Tr \left(\mat{\tilde h}\, \mat{\tilde\Omega}^\intercal\right)+\frac{1}{4}\Tr \left(\mat h\,{\mat\Omega}^\intercal\right)\\
    &=\epsilon-\frac{1}{4}\Tr\left(\mat{\tilde h}_A\mat{\tilde \Omega}_A+\mat{\tilde h}_{\bar{A}} \mat{\tilde \Omega}_{\bar{A}}\right)\,,
\end{align}
where $\mat{\tilde h_{A}}=\mat{\tilde h_{\bar{A}}}=\epsilon \cos{2r} \,\mathbb{A}_2$ are the Hamiltonian matrices of the restriction of $\hat H$ onto the single partner modes respectively, \ie the diagonal blocks of $\mat{\tilde h}$ in \eqref{eq:h_ferm_partnerbasis_degenerate}. To  minimize $\Delta E$, the target modes need to be prepared in the ground state of $\mat{\tilde h_{A}}$ and $\mat{\tilde h_{\bar{A}}}$, \ie $\mat{\tilde\Omega_A}'=\mat{\tilde\Omega_{\bar{A}}}'=\mathbb{A}_2$. 

With these initial states, the minimal energy cost for the extraction of a pure two-mode squeezed state with entanglement entropy $\Delta S= S_f(\cos2r)$ is
\begin{align}\label{eq:energycostfermion_deg}
    \Delta E&= 2\epsilon \sin^2{r}\,.
\end{align}
In addition to this elevation of the energy expectation value, the final state of the source system has an energy variance of
\begin{align}
    \textstyle \Sigma_E&=\frac{\sqrt{-\Tr\left(\mat{\tilde h}\mat{\tilde G}\mat{\tilde h}\mat{\tilde G}+\mat{\tilde h}\mat{\tilde\Omega}'\mat{\tilde h}\mat{\tilde\Omega}'\right)}}{4}\\
    &=\frac{\epsilon}{\sqrt{2}}\sin{2r}\,.
\end{align}

\subsubsection{General two-mode Hamiltonian ($\epsilon_1\!\leq\!\epsilon_2$)}
The case of a general Hamiltonian, with $\epsilon_1 \leq \epsilon_2$, is slightly more involved than the degenerate case. Interestingly, there exists a critical value of extracted entanglement $\Delta S$ below which the minimal energy cost is independent of the choice of partner modes. Above the critical amount, it is optimal to chose partner modes obtained from squeezing the eigenmodes of $\hat H$ directly.

In the non-degenerate case the \emph{orthogonal} transformations which map the energy eigenmodes to any pair of partner modes is of the general form $\mat T_{r}=\mat M_r\mat N$. This means that before the squeezing transformation $\mat{M}_r$  we can apply a transformation $\mat{N}$ which can change $\mat h$ but leaves $\mat{\Omega}$ invariant.
To find all possible forms of $\mat{N}$ we write it as a matrix exponential $\mat{N}=\exp\left(\theta \mat K\right)$ of the generator $\mat{K}$ subject to the constraints $\mat{K}\mat{G}+\mat{G}\mat{K}^\intercal=0$ (leaving $\mat{G}$ invariant, \ie being orthogonal) and $\mat{K}\mat{\Omega}+\mat{\Omega}\mat{ K^\intercal}=0$ (leaving $\mat \Omega$ invariant, \ie being symplectic).\footnote{These are indeed the same conditions as for bosons. Mathematically speaking, the intersection of the orthogonal algebra $\mathrm{so}(2N)$ and the symplectic algebra $\mathrm{sp}(2N)$ is given by the unitary algebra $\mathrm{u}(N)$.} From this, we find the most general form
\begin{align}
   \mat{K}&= \left(\begin{array}{cccc}
0 & 0 & \cos{\phi} & -\sin{\phi}\\
0 & 0 & \sin{\phi} & \cos{\phi}\\
-\cos{\phi} & -\sin{\phi} & 0 & 0\\
\sin{\phi} & -\cos{\phi} & 0 & 0
\end{array}\right)\,,
\end{align}
giving rise to $\mat{N}=\exp(\theta \mat{K})$.

Consequently, $\mat{T}_r=\mat{M}_r\mat{N}$ provides the most general transformation which brings the covariance matrix $\mat{\tilde{\Omega}}=\mat{T}_r\mat{\Omega}\mat{T}_r^\intercal$ in the form of~\eqref{eq:OmegaAplusPARTNER}.
Such a transformation brings the Hamiltonian matrix into the general form
\begin{align}\label{eq:htilde_fermion}
	\mat{\tilde h}=(\mat{T}_r^{-1})^\intercal \mat{h}  \mat{T}_r^{-1}
	=\left(\begin{array}{c|c}
	\mat{h}_A & \mat{h}_{X}\\ [1mm]
	\hline \\ [-4mm]
	-\mat{h}_X^\intercal & \mat{h}_{\bar{A}}
	\end{array}\right)\,, 
\end{align}
with blocks that read, using  $\epsilon_\pm=\frac{1}{2}(\epsilon_2\pm\epsilon_1)$,
\begin{widetext}
\begin{align}
	\mat{h}_A&=\begin{pmatrix}
	0 & \epsilon_+\cos{2r}-\epsilon_-\cos{2\theta}\\
	-\epsilon_+\cos{2r}+\epsilon_-\cos{2\theta} &0
	\end{pmatrix}\,,\\
	\mat{h}_{\bar{A}}&=\begin{pmatrix}
	0 & \epsilon_+\cos{2r}+\epsilon_-\cos{2\theta}\\
	-\epsilon_+\cos{2r}-\epsilon_-\cos{2\theta} &0
	\end{pmatrix}\,,\\
\mat{h}_X&=\begin{pmatrix}
	\epsilon_-\sin2\theta \sin\phi& \epsilon_+ \sin 2r+\epsilon_-\cos\phi\sin2\theta\\
	\epsilon_+\sin2r-\epsilon_-\cos\phi\sin2\theta &\epsilon_-\sin2\theta\sin\phi
	\end{pmatrix}\,.
\end{align}
\end{widetext}
The excitation energies of the restricted one-mode Hamiltonian $\hat{H}_A$ and $\hat{H}_{\bar{A}}$ are thus
\begin{align}
    \epsilon_{A}=|\epsilon_+\cos{2r}-\epsilon_-\cos{2\theta}|\,,\\
    \epsilon_{\bar{A}}=|\epsilon_+\cos{2r}+\epsilon_-\cos{2\theta}|\,.
\end{align}
As discussed above, the entanglement extraction brings the partner modes $A$ and $\bar A$ into a pure Gaussian product state. 
Since $A$ and $\bar A$ are fermionic modes, there are only four different product states. Their covariance matrices are
\begin{align}
	%\mat{\tilde\Omega}'=%\mat{ \Omega}^{(k,l)}=
	\mathbb{A}^{(k,l)}=
	\left(\begin{array}{c|c}
	(-1)^k \mathbb{A}_2 & 0 \\ [1mm]
	\hline \\ [-3mm]
	0 &(-1)^l\mathbb{A}_2
	\end{array}\right)\,,
\end{align}
with $k,l\in\{0,1\}$. 
Which of these states minimizes the energy cost, \ie is   the product of the ground states of $\mat h_A$ and $\mat h_{\bar A}$, depends on spectrum of the Hamiltonian and the squeezing parameter $r$.
This is because the energy cost arising from the different $\mathbb{A}^{(k,l)}$ is
\begin{align}\label{eq:ferm_energycost_productsates}
    &\Delta E=\epsilon_+-\frac14 {\Tr\left( \mat{\tilde h}\, \mathbb{A}^{(k,l)}\right)}\nonumber\\
    &=\left\{ \begin{array}{cl}
        \epsilon_+ (1-\cos2r), & \text{if }(k,l)=(0,0) \\
        \epsilon_+-\epsilon_- \cos2\theta, & \text{if } (k,l)=(1,0)\\
        \epsilon_++\epsilon_- \cos2\theta, &\text{if } (k,l)=(0,1)\\
        \epsilon_+ (1+\cos2r), & \text{if }(k,l)=(1,1) 
    \end{array}\right.\,.
\end{align}

The fact that for two of the states the energy cost is independent of $\theta$ versus it being independent of $r$ for the other two states, has an interesting consequence:
It means that for low values of $r$, \ie for low amounts of entanglement, the minimal energy cost is independent of $\theta$ and given by $\Delta E_{\text{min}}=2\epsilon_+\sin^2r$ which is achieved by the state $\mathbb{A}^{(0,0)}$.
However, as soon as $\cos2r\leq \frac{\epsilon_-}{\epsilon_+}$ is stisfied, it is more beneficial to bring the modes into the final state $\mathbb{A}^{(1,0)}$ which results in an energy cost $\Delta E=\epsilon_+-\epsilon_-\cos2\theta$. In this case, the energy cost is minimized by the parameter value $\theta=n\pi$ with $n\in \mathbb{Z}$, \ie if the partner modes are obtained by squeezing directly the energy eigenmodes.

In summary, the minimal energy cost for the extraction of a pair of partner modes sharing entanglement entropy $\Delta S=S_f\left(\cos2r\right)$ from a  fermionic quadratic Hamiltonian with excitation energies $\epsilon_1\leq\epsilon_2$, shown in Figure~\ref{fig:DeltaEvsDeltaSfermion}, is
\begin{align}\label{eq:DeltaEmin_fermion}
	\Delta E_{\mathrm{min}}=\begin{cases}
	(\epsilon_1+\epsilon_2) \sin^2{r}, & \text{if  }\cos{2r} \geq\frac{\epsilon_2-\epsilon_1}{\epsilon_2+\epsilon_1} \\
	\epsilon_1, &\text{if  }\cos{2r} \leq\frac{\epsilon_2-\epsilon_1}{\epsilon_2+\epsilon_1}
	\end{cases}\,.
\end{align}
In particular, the minimal energy cost is bounded from above by the lower excitation energy of the Hamiltonian $\epsilon_1$.
This is because above the critical squeezing parameter with $\cos2r\leq\epsilon_-/\epsilon_+$, the product state $\mathbb{A}^{(1,0)}$ in which the partner modes are left after the extraction is exactly the first excited state of the Hamiltonian $\hat H$. In fact, the covariance matrix $\mathbb{A}^{(0,1)}$ is left invariant by the squeezing operation $\mat M_r  \mathbb{A}^{(1,0)} \mat M_r^\intercal$, \ie the first excited state of $\hat H$ has the same covariance matrix with respect to the energy eigenmodes and with respect to the partner modes $A$ and $\bar A$ which are obtained by squeezing the energy eigenmodes.
\begin{figure}[t]
\begin{center}
  \includegraphics[width=\linewidth]{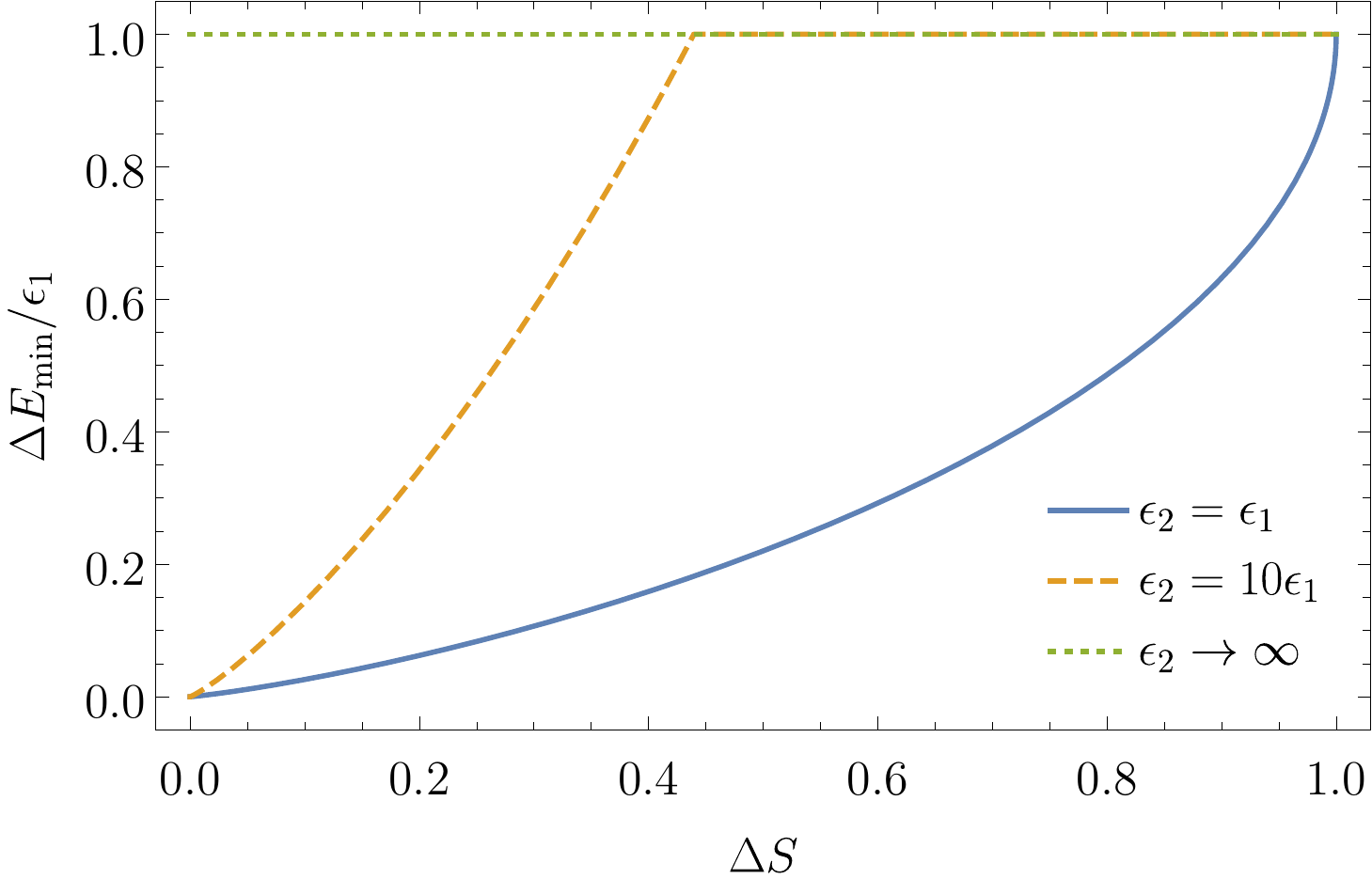}
\end{center}\vspace{-4mm}
\caption{Minimal energy cost $\Delta E_{\mathrm{min}}$ of extracting entanglement entropy $\Delta S$ from a fermionic, quadratic two-mode Hamiltonian with excitation energies $\epsilon_1$ and $\epsilon_2$.}
\label{fig:DeltaEvsDeltaSfermion} 
\end{figure}

As a consequence, in the regime $\cos{2r}\leq\epsilon_-/\epsilon_+$, the energy variance $\Sigma_E$ of the final system state vanishes because the system is in an energy eigenstate. For lower amounts of entanglement, where this is not the case, the variance of the state which minimizes the energy cost depends on the squeezing parameter. It is plotted in Figure \ref{fig:SigmaEvsDeltaSfermion} and reads
\begin{align}
    \Sigma_E&=\left\{\begin{array}{lcl}
	\frac{\epsilon_+}{\sqrt{2}}\sin{2r}, &&  \text{if  }\cos{2r} >\frac{\epsilon_2-\epsilon_1}{\epsilon_2+\epsilon_1}\\
	0, &&  \text{if  }\cos{2r} \leq \frac{\epsilon_2-\epsilon_1}{\epsilon_2+\epsilon_1}
	\end{array}\right.\,,
\end{align}
Note that for $\cos{2r} >\frac{\epsilon_2-\epsilon_1}{\epsilon_2+\epsilon_1}$, it is possible to leave the source system in an energy eigenstate such that $\Sigma_E=0$. The energy cost for this is $\epsilon_1$.

Another relevant figure of merit is to look at the ratio between a given amount of extracted entanglement and its minimal energy cost $\Delta S/\Delta E$. Since $\Delta S=S_f(\cos{2r})\sim (1-2\log{r})\,r^2$ as $r\to0$, the ratio diverges as
\begin{align}
    \lim_{r\to0}\frac{\Delta S}{\Delta E_{\text{min}}}=\infty
\end{align}
for small amounts of extracted entanglement, just as we found for bosons in~\eqref{eq:limDeltaSperDeltaE}. As the extracted entanglement increases the ratio is strictly decreasing and reaches its minimum $\log(2)/\epsilon_1$ for $r=\pi/4$. Similar to the bosonic case, we thus find that the price per entanglement bit increases as we are trying extract a larger total amount.

\begin{figure}[t]
\begin{center}
  \includegraphics[width=\linewidth]{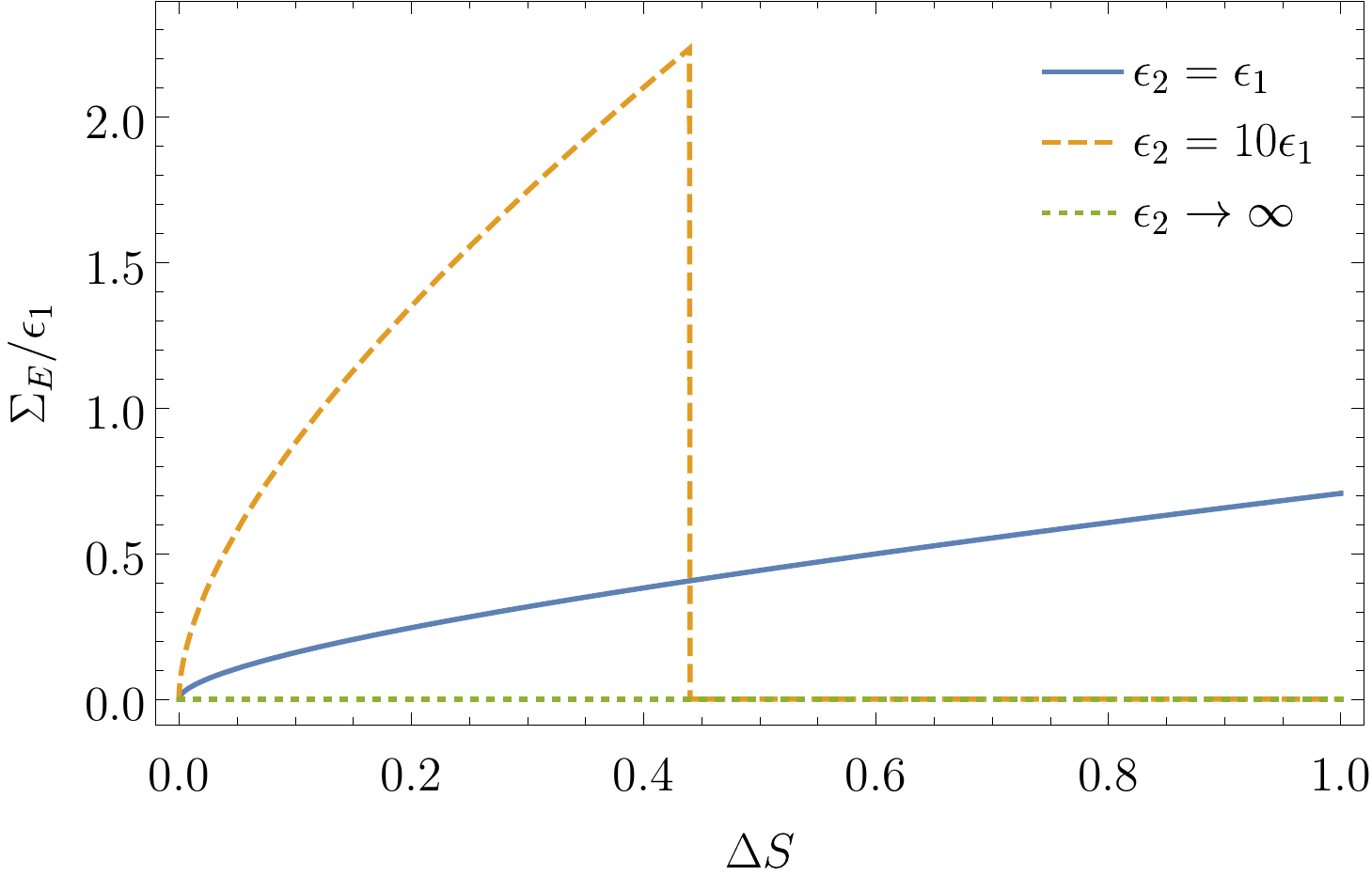}
\end{center}\vspace{-4mm}
\caption{Energy variance $\Sigma_E$ of the state minimizing the energy cost of extracting partner modes with entanglement entropy $\Delta S$ from a fermionic, quadratic two-mode Hamiltonian with excitation energies $\epsilon_1$ and $\epsilon_2$.}
\label{fig:SigmaEvsDeltaSfermion} 
\end{figure}

\subsection{Choosing optimal partner modes}\label{sec:fermion_proof3}
The analysis of the previous sections explain how the increase of energy expectation value of the source system can be minimized when extracting a pair of partner modes sharing a given amount of entanglement: To keep the energy cost as low as possible the modes need to be constructed from the two lowest energy eigenmodes of the Hamiltonian. Specifically, the partner modes should be obtained by directly squeezing the two energy eigenmodes, such that always the lowest possible energy cost according to \eqref{eq:DeltaEmin_fermion} is achieved. Again, as discussed in Appendix \ref{app:spectrumproof}, no other choice of partner modes yields a lower energy cost because the energy spectrum of the restricted Hamiltonian is bounded by the spectrum of the full Hamiltonian.

For a general pair of partner modes $A$ and $\bar A$ the minimal energy cost is determined by the spectrum of the restricted two-mode Hamiltonian $\hat H_{A\bar A}$, which enters the formula \eqref{eq:DeltaEmin_fermion}.
In contrast to the fermionic case, an arbitrarily chosen pair of partner modes will in fact achieve this minimal energy price if the amount of extracted entanglement is below the critical value, \ie if $\cos{2r}\geq\frac{\epsilon_-}{\epsilon_+}$.
If the squeezing is above this value, according to the analysis of the previous section only the partner modes obtained by directly squeezing the energy eigenmodes of $\hat H_{A\bar A}$ do achieve the possible minimal energy cost of $\Delta E_{\text{min}}=\epsilon_1$.

To perform this analysis for a given choice of partner modes, it is necessary to express the Hamiltonian matrix of $\hat H_{A\bar A}$ with respect to the partner mode operators $\left(\hat Q_A,\hat P_A,\hat Q_{\bar A},\hat P_{\bar A}\right)$. As before, we denote these mode operators as
\begin{align}
    \hat Q_A&=x_a\hat\xi^a\,,\qquad \hat P_A=k_a\hat\xi^a\,,\\
    \hat Q_{\bar{A}}&=\bar{x}_a\hat\xi^a,\qquad \hat P_{\bar{A}}=\bar{k}_a\hat\xi^a\,.
\end{align}
In this basis, the Hamiltonian matrix representing the restriction $\hat H_{A\bar A}$ of a general quadratic Hamiltonian $\hat H=\frac\ii2 h_{ab}\hat\xi^a\hat\xi^b$ is given by
\begin{align}
\left(\begin{array}{cc|cc}
0& \check h(x,k) & \check h(x,\bar x) &\check h(x,\bar k)\\
\check h(k,x) & 0 & \check h(k,\bar x) & \check h(k,\bar k)\\ \hline
\check h(\bar x,x) & \check h(\bar x,k) & 0& \check h(\bar x, \bar k)\\
\check h(\bar k,x) & \check h(\bar k,k) & \check h(\bar k,\bar x) & 0
\end{array}\right)
\end{align}
where we now define $\check h:V^*\times V^*\to\mathbb{R}$ by
\begin{align}
\check h(v,w)=  G^{ab}  G^{cd}v_ah_{bc} w_d\,,
\end{align}
which for fermionic modes yields a  bi-linear   anti-symmetric form:
\begin{align}
\check h(v,w)=-\check h(w,v)\,.
\end{align}
Using the fact that $K$ as defined in \eqref{eq:ferm_groundstate} commutes with $J$, it can be shown that
\begin{align}
    \check h(J^\intercal v,J^\intercal w)=\check h(v,w).
\end{align}
Which allows to calculate the terms invovling the partner mode $(\bar x,\bar k)$ as
\begin{align}
&\check h(x,\bar x)=\check h(k,\bar k)=\frac{ \check h(x,J^\intercal k)}{\sin2r} \,,\\
&\check h(x,\bar k)=-\cot2r \,\check h(x,k)+\frac{\check h(x,J^\intercal x)}{\sin2r}\,,\\
&\check h(k,\bar x)= -\cot2r \,\check h(x,k)+\frac{\check h(k,J^\intercal k)}{\sin2r}\,,\\
&\check h(\bar x,\bar k)=\frac{\cot2r}{\sin2r}\left(\check h(x,J^\intercal x)+\check h(k,J^\intercal k)\right)\nonumber\\*
&\qquad\qquad -(1+2\cot^2(2r)) \check h(x,k)\,.
\end{align}
Then, by equating the obtained matrix with \eqref{eq:htilde_fermion}, we obtain the spectrum of $\hat H_{A\bar A}$ with $\epsilon_\pm=\frac12(\epsilon_2\pm\epsilon_1)$ as
\begin{align}
\epsilon_+^2& \!=\!\frac{\left( \check h(\bar x,\bar k)+\check h(x,k)\right)^2+ \left( \check h(x,\bar k)+ \check h(k,\bar x)\right)^2}4,\\
\epsilon_-^2& \!=\!\frac{\left( \check h(\bar x,\bar k)-\check h(x,k)\right)^2+ \left( \check h(x,\bar k)- \check h(k,\bar x)\right)^2}4\nonumber\\*
&\qquad+\frac{\check h(x,\bar x)^2 }4\,.
\end{align}
The two parameters $\phi$ and $\theta$ in the transformation $\mat T_r$ (mapping the eigenmodes of $\hat H_{A\bar A}$ to the partner modes $A$ and $\bar A$) are determined by
\begin{align}
&\tan\phi=\frac{2 \check h(x,\bar x)}{\check h(x,\bar k)-\check h(k,\bar x)}\,, \\
&\cos2\theta=\frac{\check h(\bar x,\bar k)-\check h(x,k)}{2 \epsilon_-}\,.
\end{align}

With this information on the spectrum and the parameters it is possible to check if a given pair of partner modes minimizes the energy cost: If we find that the extracted amount of entanglement is low enough such that $\cos{2r}\geq\frac{\epsilon_-}{\epsilon_+}$ then the pair achieves the minimal energy cost of $\Delta E_{\text{min}}=2\epsilon_+\sin^2r$ independent of the value of $\theta$. However, beyond this regime the minimal energy cost depends on the value of $\theta$ and according to \eqref{eq:ferm_energycost_productsates} can be lowered to
\begin{align}
    \Delta E_\theta=\epsilon_+-\epsilon_-|\cos2\theta|\,,
\end{align}
\ie only if $\theta\in\{0,\pi/2,\pi,\dots\}$ the minimal energy cost of $\Delta E_{\text{min}}=\epsilon_+-\epsilon_-=\epsilon_1$ is achieved by the chosen partner modes $A$ and $\bar A$.

Finally, also for fermionic modes we can give an upper bound on the energy cost of entanglement extraction from arbitrary partner modes of the system. Combining \eqref{eq:DeltaEmin_fermion} and \eqref{eq:spec_restrictions_boson} (which as shown in Appendix \ref{app:spectrumproof} applies to fermionic Hamiltonians in the same form), we find that the energy cost to replace the state of an arbitrary pair of partner modes by the product states of their one-mode restricted ground states
\begin{align}\label{eq:upper_bound_cost_ferm}
    \Delta E \leq 2\epsilon_+\sin^2r \leq 2\,\omega_N\sin^2r=:\Delta E_{\mathrm{max}}
\end{align}
is again upper bounded by the energy cost for a degenerate Hamiltonian with only the system's largest excitation energy $\omega_N$ in its spectrum.

\section{Applications}\label{sec:applications}
In this section, we apply the previously derived minimal energy cost to concrete physical models. The main goal is to present a proof of concept that our general results are easily related to concrete scenarios of extracting entanglement from modes that are accessible in a physical model. In particular, we present examples for both bosonic and fermionic systems, the latter being also related to spin systems via Jordan-Wigner transformation.

\subsection{Hamiltonian of dilute Boson gas}\label{sec:example_Einstein_Bose}
We consider the Hamiltonian of the form 
\begin{align}\label{eq:bosongasH}
    \hat H=\sum_{k\neq0} \omega_k a_k^\dagger a_k+ \gamma_k \left(a_k^\dagger a_{-k}^\dagger+a_k a_{-k}\right)\,,
\end{align}
with $\omega_k=\omega_{-k}$ and $\gamma_k=\gamma_{-k}$, which is well known to describe a weakly interacting dilute Boson gas (see, \eg \cite{schwabl_advanced_2008}).
The sum runs over all modes  with non-zero momentum, which have very low occupation numbers, but it excludes the zero momentum mode, which is macroscopically occupied. In this way, the Hamiltonian describes the interaction of excitations in the higher modes with the condensate of particles in  the zero-mode.

Typically, the squeezing term in the Hamiltonian is much smaller than the number operator term, \ie $\gamma_k\ll\omega_k$. However, as long as $\gamma_k<\omega_k/2$ the Hamiltonian is bounded from below, as we shall see below. 

Inserting $a_k=\frac1{\sqrt2}\left(\hat q_k+\ii\hat p_k\right)$, we can rewrite the Hamiltonian  as
\begin{align}
\hat H&=\sum_{k>0} \hat H_k -\sum_{k\neq0} \frac{\omega_k}2
\end{align}
with
\begin{align}
\hat{H}_k&= \frac{\omega_k}2\left(\hat q_k^2+\hat p_k^2+\hat q_{-k}^2+\hat p_{-k}^2\right)\\
&\qquad\qquad +2 \gamma_k \left( \hat q_k\hat q_{-k} - \hat p_k\hat p_{-k}\right)\,.
\end{align}
It is evident, that modes with opposite momentum, \ie pairs of the form $(\hat q_k,\hat p_k,\hat q_{-k},\hat p_{-k})$ are partner modes in the ground state of   $\hat H$ because the blocks $\hat H_k$ couple them pairwise.
Since we assumed $\omega_k=\omega_{-k}$ each block $\hat H_k$ is a degenerate two-mode Hamiltonian. In fact, the Hamiltonian matrix corresponding to $\hat H_k$ is
\begin{align}
    \mat{h}_k&= \begin{pmatrix}\omega_k &0 & 2\gamma_k &0\\ 0&\omega_k&0&-2\gamma_k \\ 2\gamma_k&0&\omega_k&0\\0&-2\gamma_k&0&\omega_k \end{pmatrix}\,,
\end{align}
which  matches  \eqref{eq:h_partnerbasis_degenerate}, with the squeezing parameter $r$ and the excitation energy of $\hat H_k$ being
\begin{align}
&\epsilon\cosh{2r}=\omega_k\,,\quad \epsilon\sinh{2r}=-2\gamma_k\,, \nonumber\\*
&\Leftrightarrow%\quad
    \epsilon=\sqrt{\omega_k^2-4\gamma_k^2},\quad 
    \cosh{2r}=\frac1{\sqrt{1-4\gamma_k^2/\omega_k^2}}\,. %\tanh(2r)=\frac{-2\gamma_k}{\omega_k}.
\end{align}

This means that in order to extract pure state entanglement from the ground state of the Hamiltonian $\hat H$, one needs to select two modes with opposite momentum $k$ and $-k$. Following  \eqref{eq:Sformula}, these modes share an entanglement entropy of 
\begin{align}
    \Delta S%=s_b\left(\cosh{2r}\right)
    =s_b\left(\left(1-4\gamma_k^2/\omega_k^2\right)^{-\frac12}\right)\,.
\end{align}
The minimal energy cost of entanglement extraction from these modes is, according to \eqref{eq:energycost_degenerate},
\begin{align}
    \Delta E= \epsilon \left(\cosh{2r}-1\right)
    = \omega_k \left(1-\sqrt{1-4\gamma_k^2/\omega_k^2}\right)\,.
\end{align}

\begin{figure}[t]
\begin{center}
  \includegraphics[width=\linewidth]{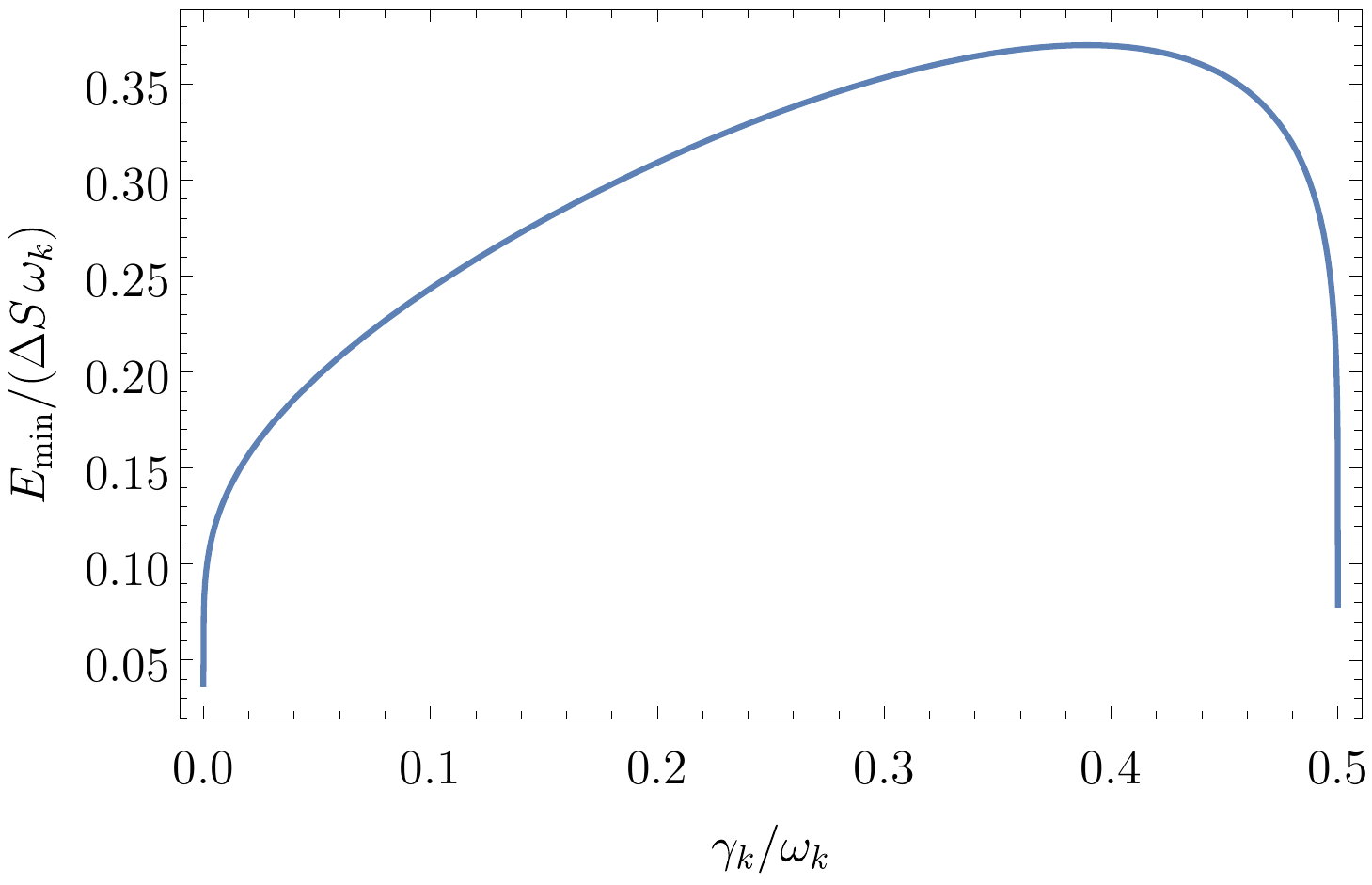}
\end{center}\vspace{-4mm}
\caption{Ratio of energy cost to extracted entanglement from the ground state of a dilute weakly interacting Boson gas (Hamiltonian of \eqref{eq:bosongasH}) as a function of $\gamma_k/\omega_k$. 
}
\label{fig:bosongas} 
\end{figure}

It is interesting to see that while the entropy between the partner modes diverges as $\gamma_k\to\omega_k/2$, the energy cost is upper bounded and approaches $\omega_k$. At first this may seem in contradiction to the results of Section \ref{sec:boson_extraction_wholesection}, but it is due to the excitation energy of $\hat H_k$ vanishing, $\epsilon\to0$, in this limit.
In particular, this leads to an interesting dependency of the ratio $E_{\mathrm{min}}/\Delta S$ between energy cost and extracted entanglement on the ratio $\gamma_k/\omega_k$ of the Hamiltoian parameters, seen in Figure \ref{fig:bosongas}.
The energy that needs to be invested per bit of extracted entanglement falls off drastically both for very small amounts of entanglement, $\gamma_k\to0$, and for extremely large amounts of entanglement as $\gamma_k\to\omega_k/2$.
In between the highest amount of energy per extracted entanglement entropy is required at $\gamma_k/\omega_k\approx 0.389368$, where the entanglement entropy is $\Delta S\approx 1.00679$ and the energy cost $E_{\mathrm{min}}\approx 0.372648 \, \omega_k$.

An interesting feature of the entanglement structure of this example system is that it exhibits partner modes which may be realistically accessible in experiment.
In other systems it seems much more challenging to access both modes in a pair of partner modes equally: For example, in the coupled harmonic chain generally (at least) one  of the partner modes are delocalized over the complete system.
The Hamiltonian \eqref{eq:bosongasH}, however, couples and entangles  modes of counter-propagating momentum symmetrically such that both partner modes should be equally accessible in experiment.

\subsection{XY spin model}

\begin{figure}[t]
\begin{center}
  \includegraphics[width=\linewidth]{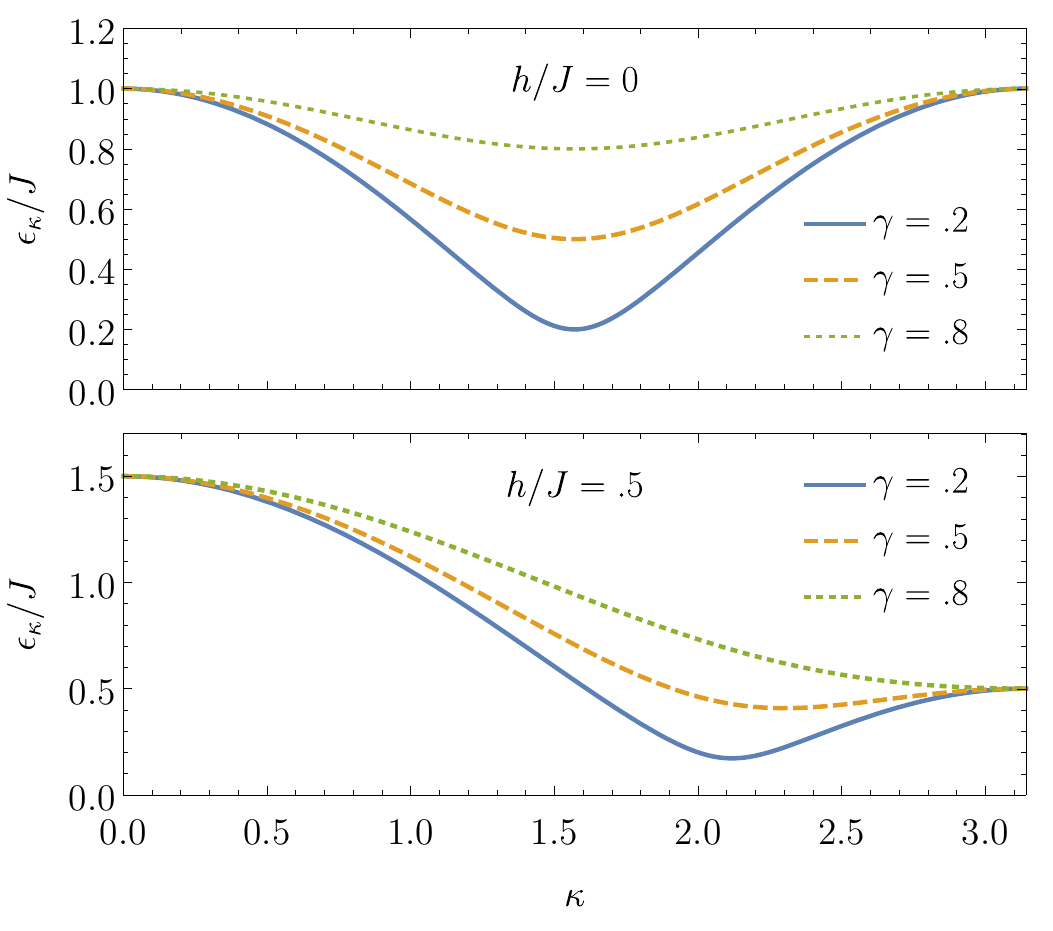}
\end{center}\vspace{-4mm}
\caption{Dispersion relation $\varepsilon_k$ of the XY model for different values of $h/J$ and $\gamma$.}
\label{fig:epsilonkXY} 
\end{figure}

The XY model presents a prime example of a spin model that can be mapped to a fermionic quadratic Hamiltonian via the Jordan-Wigner transformation~\cite{jordan_ueber_1993}. We consider the Hamiltonian with an external field $h$ given by~\cite{lieb_two_1961}
\begin{align}\label{eq:HXY1}
\begin{split}
 \hat{H}_{\mathrm{XY}}&=-\frac{J}{2}\sum^N_{j =1}\left[(1+\gamma)\hat{S}^{\mathrm{X}}_j \hat{S}^{\mathrm{X}}_{j +1}\right.\\&\qquad\quad\left.+(1-\gamma)\hat{S}^{\mathrm{Y}}_j \hat{S}^{\mathrm{Y}}_{j +1}\right]-h\hat{S}^{\mathrm{Z}}_j\,,
\end{split}
\end{align}
where $\hat{S}_j^\mathrm{X}$, $\hat{S}_j^\mathrm{Y}$ and $\hat{S}_j^\mathrm{Z}$ represent the spin-$\frac{1}{2}$ operators on site $j$. Ignoring boundary terms that are negligible in the thermodynamic limit, the XY model can be mapped to the quadratic fermionic Hamiltonian given by~\cite{holstein_field_1940,cazalilla_one_2011}
\begin{align}\label{eq:H_XYfermion}
\begin{split}
    \hat{H}_{\mathrm{XY}}=\sum_{k\in\mathcal{K}}\varepsilon_{\kappa}\ \left(\hat{\eta}_{\kappa}^\dagger\hat{\eta}_{\kappa}-\frac{1}{2}\right)\qquad\text{with}\qquad\\
    \varepsilon_{\kappa}\!=\!\sqrt{h^2\!+2h J\cos(\kappa)\!+J^2\!+(\gamma^2-1)J^2\sin(\kappa)^2}\,,
\end{split}
\end{align}
where we defined $\mathcal{K}=\{2\pi k/N|0\leq k<N\}$. 
Figure~\ref{fig:epsilonkXY} shows the dispersion relation $\epsilon_\kappa$ as a function of $\kappa$ for various parameter choices. The fermionic annihilation operators $\hat{\eta}_\kappa$ diagonalize the Hamiltonian and are related to on-site annihilation operators
\begin{align}
    \hat{f}_j=\frac{1}{\sqrt{N}}\sum_{\kappa\in\mathcal{K}}e^{-\ii\kappa }(u_\kappa \eta_{\kappa}-v_\kappa\eta_{-\kappa}^\dagger)
\end{align}
on site $j$, where the Bogoliubov coefficients $u_\kappa$ and $v_\kappa$ are given by
\begin{align}
\begin{split}
    u_\kappa=\frac{\varepsilon_\kappa+a_\kappa}{\sqrt{2\varepsilon_\kappa(\varepsilon_\kappa+a_\kappa)}}\,,\quad v_\kappa=\frac{\ii b_\kappa}{\sqrt{2\varepsilon_\kappa(\varepsilon_\kappa+a_\kappa)}}\,,\\
    a_\kappa=-J\cos(\kappa)-h\,,\quad\text{and}\quad b_\kappa=\gamma\, J\sin(\kappa)\,. 
\end{split}
\end{align}
Let us emphasize that the Jordan-Wigner transformation relating the spin Hamiltonian~\eqref{eq:HXY1} with the fermion Hamiltonian~\eqref{eq:H_XYfermion} is non-local. However, it preserve bi-partite entanglement~\cite{amico_entanglement_2008} between localized regions, \ie the entanglement entropy associated to a single fermionic site is equal to the respective entanglement entropy of the same site in the spin Hamiltonian~\eqref{eq:HXY1}. (The relationship for general non-local subsystems is more subtle.) Also, we note that the Hamiltonian in~\eqref{eq:H_XYfermion} is a paradigmatic model in its own right: Choosing $\gamma=1$, we obtain the transverse field Ising model~\cite{pfeuty_one-dimensional_1970}.

To find the minimal excitation energy, \ie the lowest eigenvalue of $\hat H_{XY}$, we need to find the minimum of the dispersion relation $\varepsilon_{\kappa}$. 
We find that the condition $\cos\kappa=\frac{h}{J(\gamma^2-1)}$ ensures that $\varepsilon^2_\kappa$ is minimal, which yields the lowest eigenvalue 
\begin{align}
    \epsilon_{\min}=\min_{\kappa}\varepsilon_{\kappa}=\sqrt{\gamma^2\left(J^2+\frac{h^2}{\gamma^2-1}\right)}\,.
\end{align}

The model becomes gapless for $\gamma=0$ and for $h=J\sqrt{1-\gamma^2}$. If the number of sites  $N$ is large enough, there are sufficiently many excitation energies close to the minimum that we can consider the spectrum as flat.
This means, that by extracting a pair of partner modes from the subspace of these lowest energy modes, we can achieve the minimal energy cost
\begin{align}
    \Delta E_{\mathrm{min}}=2\epsilon_{\min} \sin^2{r}\,,
\end{align}
with the partner modes' squeezing parameter $r$.

The ground state $\ket{\Omega}$ of this model is encoded by the complex structure~\cite{hackl_average_2019}
\begin{widetext}
\begin{align}
 J\equiv\mat{J}=-\frac{\ii}{N}\sum_{{\kappa}\in\mathcal{K}}\left(\begin{array}{c|c} |u_{\kappa}|^2e^{\ii {\kappa}(j -l)}-|v_{\kappa}|^2e^{\ii {\kappa}(l-j)} & u^*_{\kappa}v^*_{\kappa}\left(e^{\ii {\kappa}(l-j)}-e^{\ii {\kappa}(j -l)}\right)\\ \hline u_{\kappa}v_{\kappa}\left(e^{\ii {\kappa}(l-j )}-e^{\ii {\kappa}(j -l)}\right)& -|u_{\kappa}|^2e^{\ii {\kappa}(l-j )}+|v_{\kappa}|^2e^{\ii {\kappa}(j -l)} \end{array}\right)\,,
\end{align}
\end{widetext}
here expressed with respect to the local basis $\hat{\xi}^a\equiv(\hat{f}_1^\dagger,\dots,\hat{f}_N^\dagger,\hat{f}_1,\dots, \hat{f}_N)$. 
We can change to a hermitian basis defined by $\hat{q}_i=\frac{1}{\sqrt{2}}(\hat{f}_i^\dagger+\hat{f}_i)$ and $\hat{p}_i=\frac{\ii}{\sqrt{2}}(\hat{f}_i^\dagger-\hat{f}_i)$, such that $G\equiv\id$. In this basis, the only non-vanishing entries of the covariance matrix are given by
\begin{align}
\begin{split}
    \Omega(q_j,p_l)&=\frac{1}{L}\sum_{\kappa\in\mathcal{K}}\Big[(|v_\kappa|^2-|u_\kappa|^2)\cos{\kappa(j-l)}\\
    &\qquad\qquad\,\,+2\,\mathrm{Im}(u_kv_k)\sin{\kappa(j-l)}\Big]\,.
\end{split}
\end{align}
Similarly, we find the Hamiltonian matrix to be
\begin{align}
    \check{h}(q_j,p_l)=h\delta_{jl}+\frac{J+\gamma}{2}\delta_{j+1,l}+\frac{1-\gamma}{2}\delta_{j-1,l}\,,
\end{align}
where we have periodicity $N+1\equiv 1$.

We can now explore how the energy cost of randomly chosen partner modes relate to the minimal energy cost $\Delta E_{\mathrm{min}}$ and the upper bound $\Delta E_{\mathrm{max}}$ established above. To this end, we choose  a Haar random mode $(x_a,k_a)$ spanning a subsystem $A\subset V$ and then apply the following steps:
\begin{enumerate}
    \item Compute entanglement entropy $S_A(\ket{\Omega})$.
    \item Compute partner mode $(\bar{x},\bar{k})$.
    \item Compute the restriction $\hat{H}_{A\bar{A}}$.
    \item Compute the energy cost $\Delta E$ for replacing the partner mode state   by the ground state $\ket{\Omega_A'}\otimes \ket{\Omega'_{\bar{A}}}$ of $\hat{H}_A+\hat{H}_{\bar{A}}$.
\end{enumerate}
Selecting a random mode with respect to the Haar measure of the orthogonal group can be accomplished by generating a random $2N$-by-$2N$ matrix with Gaussian distributed entries and orthonormalize its column vectors with respect to the inner product induced by $G^{ab}$. From here, we can select the first two column vectors. This is equivalent to directly generating two vectors and orthonormalizing them.

\begin{figure}[t!]
\begin{center}
  \includegraphics[width=\linewidth]{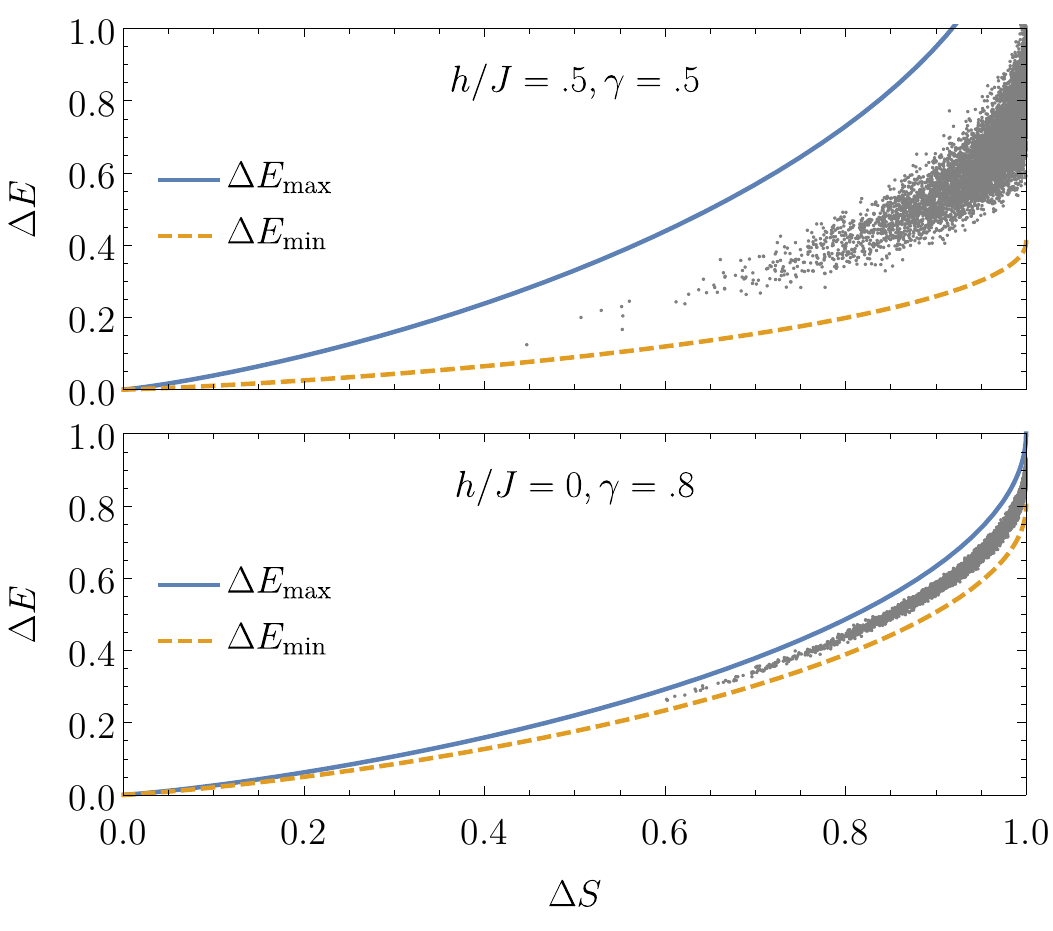}
\end{center}\vspace{-4mm}
\caption{Copmarison of  lower and upper bounds for the energy cost of entanglement extraction with a sample of $10,000$ Haar randomly generated partner modes.
We consider the XY model with $N=10$, and parameter choices $(h=.5,\gamma=.5)$ and $(h=0,\gamma=.8)$. 
}
\label{fig:randomXY} 
\end{figure}

As seen in Figure~\ref{fig:randomXY}, the Haar random sample is accurately bounded from below by our predicted minimal energy cost. The upper bound on the energy cost is due to the fact that the spectrum of the XY model is also bounded from above by
\begin{align}
    \epsilon_{\max}=\max_{\kappa}\epsilon_{\kappa}=\sqrt{J^2+h^2+2hJ}\,.
\end{align}
We see that our bounds become tighter for narrower excitation spectra. As seen from Figure~\ref{fig:epsilonkXY} our spectrum becomes flat for $h\to 0$ and $\gamma\to 1$. This is reflected in Figure~\ref{fig:randomXY} for $(h=0,\gamma=.8)$.

Our analysis of the energy cost for the XY model did not take into account locality in real space so far. In particular, the optimal partner modes that accomplish our minimal energy cost are spanned by the de-localized energy eigenmodes with lowest energy values. 
It is therefore interesting, to investigate the energy cost in a more physical scenario where one mode is  localized on a given site, while its partner mode is localized in the complement of the site.
In this framework the choice of modes is completely fixed, hence we can scan through the $(h,\gamma)$ parameter space of the XY model to compare the energy cost for this single site extraction to the global lower and upper bounds.

In Figure~\ref{fig:XYsingle}, we show the energy cost of entanglement extraction for a single-site mode and its partner. We scan through the parameter space along three different paths of points, which are constrained to yield the same lowest excitation energy $\epsilon_{\min}$ for the Hamiltonian $\hat H_{XY}$, \ie
\begin{align}\label{eq:XYfixedE}
    h=\pm\frac1{\gamma } {\sqrt{\left(1-\gamma ^2\right) \left(\gamma ^2-\epsilon_{\min}^2\right)}}\,.
\end{align}
We notice that the entanglement entropy of a single site with the rest of the system is close to maximal, \ie $\Delta S\approx 1$. This is not surprising, as it is known~\cite{vidmar_entanglement_2017}  that every local site becomes maximally entanglement in the thermodynamic limit $N\to \infty$. Furthermore, we see that for $\epsilon_{\min}=.9$, the actual energy cost from the single site extraction is comparably close to our lower bound. This indicates that for certain parameter choices our minimal energy cost could actually be achieved in this concrete physical scenario.

\begin{figure}[t]
\begin{center}
  \includegraphics[width=\linewidth]{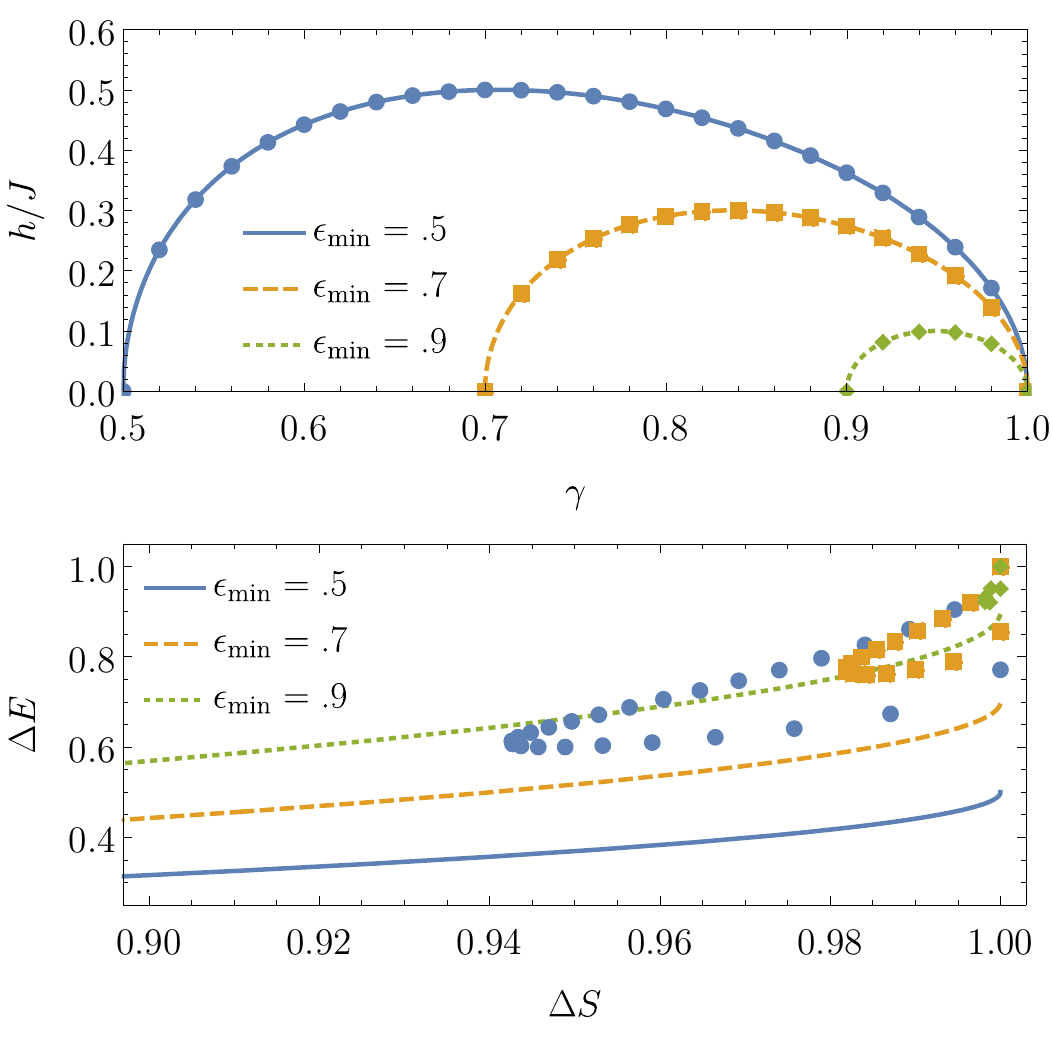}
\end{center}\vspace{-4mm}
\caption{We consider the XY model with $N=10$ and choose the model parameters $h/J$ and $\gamma$, such that the minimal excitation energy $\epsilon_{\min}$ is in $\{.5,.7,.9\}$ as indicated in the upper panel. For each point in the parameter space, we then compute the energy cost of entanglement extraction between a single site mode and its partner. We compare these sample points with our lower bounds for the given energy.}
\label{fig:XYsingle} 
\end{figure}

\section{Discussion and outlook}\label{sec:discussion}
In this work, we have found the minimal energy cost of extracting a pair of entangled partner modes from the ground state of a quadratic Hamiltonian, and we have characterized the modes which achieve this minimal energy cost.

The minimal energy cost is linear in the spectrum of the Hamiltonian restricted to the two-mode space spanned by the extracted partner modes. Hence, as one may intuitively expect, the globally lowest energy cost possible is achieved by partner modes constructed from  the two lowest energy eigenmodes of the Hamiltonian.
In some applications, access to such less localized energy eigenmodes may physically not be possible. However, our analysis of the XY spin model showed instances of modes that are localized on a single site but still came relatively close to the global minimum, with their delocalized partner mode.

Interestingly, the minimal energy cost of entanglement extraction can be used to derive an upper bound on a Hamiltonian's spectral gap: If it is possible to extract a pair of partner modes sharing entanglement $S$ from the ground state, while only increasing the system's energy by an amount $\Delta E$, then the Hamiltonian's lowest eigenvalue is at most as high as the value $\epsilon$ implied by the minimal energy cost relation derived above.

Apart from the question of the global minimal energy cost, as dictated by the spectrum of the full Hamiltonian, it is meaningful to ask whether a given pair of partner modes minimize the energy cost locally, \ie whether they achieve the minimal energy cost allowed for by the spectrum of the Hamiltonian restricted to their two-mode subspace. If not, then there exist symplectic transformations acting only on the given partner modes which yield another pair of partner modes achieving the local minimum.

We showed that for bosonic modes there are  only specific pairs of partner modes which achieve the local minimal energy cost. These modes are  obtained by squeezing the odd and even superpositions of the eigenmodes of the restricted Hamiltonian.
However, for fermionic modes, if the amount of extracted entanglement lies below a critical value, all partner modes achieve the local minimal energy cost. Only above the critical value, specific partner modes are required which are obtained by squeezing the eigenmodes of the restricted Hamiltonian. 
The latter is due to the fact that above the critical value, the lowest excited energy eigenstate of the two-mode subsystem then coincides with the product of the restricted ground states of the partner modes.

An intriguing question for future research will be to analyze the dynamics of the extraction process due to the time evolution of the source modes under the source system's Hamiltonian.
In certain cases this dynamic may be relatively simple. For example, the partner modes which minimize the energy cost lie within the subspace of only the two lowest energy eigenmodes of the system. Hence their time evolution is restricted to this subspace.
In general, for partner modes constructed from a range of energy eigenmodes, a more complex time evolution is expected. 
This would affect how the extraction from such source modes can be practically implemented, as well as it affects how the energy injected during the extraction process  dissipates inside the source system after the extraction process.
Indeed, the prospect that such complex dynamics potentially could be exploited in the design of entanglement extraction protocols motivates further study of the dynamics and time evolution of partner mode pairs.

\subsection{Outlook: Extraction from arbitrary modes}
In this work, we assumed that the system modes which are swapped onto the target modes  can be selected freely from the source system.
This led us to choose a pair of partner modes which is the only way to obtain a pure final state in the target modes.
Hence, we derived the energy cost of replacing the state of arbitrary partner modes in a system's ground state with the product state of their restricted ground states. In particular, we established an upper bound to this cost in terms of the systems largest excitation energy.

In practical implementations of entanglement extraction protocols, however, the access to the source system might be limited.
For example,  the modes could be restricted to be localized in a given subregion of the system, or to lie in the span of a limited set of modes. In particular, such restriction could make it impossible to access  pairs of partner modes. It could also be interesting to consider randomly chosen modes, see, e.g., \cite{dahlsten_entanglement_2014}.

This raises the important question what our results imply for more general extraction protocols that are not based on partner modes. 
We expect that the minimal energy cost for entanglement extraction from partner modes, that we derived here, also constitutes a lower bound for the minimal energy cost for swapping pairs of arbitrary single modes $A$ and $B$ out of the system that are not necessarily partners.

If the modes $A$ and $B$ are not partner modes then other entanglement monotones than the entanglement entropy $S(\rho_A)$ or $S(\rho_B)$ are necessary to measure the entanglement between them, such as  logarithmic negativity or entanglement of purification need to be used instead.
This is because the modes $A$ and $B$ share entanglement with the remainder of the system, such that their joint state $\rho_{AB}$, \ie the state of the subsystem $A\oplus B$, is mixed.

Whichever mixed state entanglement $\mathcal{N}_{AB}$ measure we employ, we may assume that the lower entanglement entropy of the two modes $S_{\min}=\min (S(\rho_A),S(\rho_B))$ implies  an upper bound on $\mathcal{N}_{AB}$ because the entanglement between a mode $A$ and its complement is larger than the entanglement between modes $A$ and $B$.
In particular, this means that a pair of partner modes sharing entanglement $S_{\min}$ is at least as valuable as the two-mode state $\rho_{AB}$.

However, we expect that the minimal energy cost for the extraction of a partner mode pair $A$ and $\bar{A}$ with entanglement $S_{\min}$ is always lower than the energy cost of extracting modes $A$ and $B$. In fact, this can be proven for the case of a Hamiltonian with degenerate excitation energies and we have numerical evidence for the general case~\cite{hackl_jonsson_2019b}.

In this way, the derived minimal energy cost from partner mode extraction actually applies much more broadly to essentially any entanglement monotone (of Gaussian states) and any two target modes $A$ and $B$.
Note, however, that this minimal energy cost will only be saturated if we are actually using partner modes, which are in this sense optimal.
Where this is prevented by additional constraints, \eg locality assumptions, we are likely to underestimate the actually required energy cost.

\subsection{Outlook: Mixed states}
Furthermore, it will be interesting to consider extraction from states other than the ground state of the source system. Of particular relevance for practical implementations should be general mixed states and thermal states. 
There, we may encounter situations with vanishing or even negative energy cost. Similarly, in quantum systems with spontaneously broken symmetries, \ie with several ground states, one might be able to extract entanglement by moving from one ground state to another.

In the context of mixed states, but also in general, it should be relevant to consider more general methods of entanglement extraction. Instead of replacing the entangled state of a subsystem by a product state in one discrete step, one could consider more general local operations on the subsystem such as the  local operations and classical communication (LOCC) envisioned in the general framework of~\cite{beny_energy_2018}. This may yield an infinitesimal relation between entropy and energy increase around a given system state.

For these and other future research directions, the advancement of the partner mode formula provided in this work should prove instrumental because it can yield better insight into the correlation structure of mixed Gaussian states.
In particular, we are interested in extending the partner mode formula to mixed states. This allows, \eg for the decomposition of a mixed Gaussian states into chains of $k$-tuples of modes where the modes in each $k$-tuple are not correlated with each other, and only correlated with modes in the preceding and subsequent tuple. In fact, this construction does not only apply to mixed states but to bi-linear forms in general. Thus, it may be used to further develop techniques such as \cite{chin_exact_2010}  to efficiently describe the interaction of several quantum systems with a given bath of harmonic oscillators, with applications ranging from open system dynamics to fundamental quantum  communication between local observers via relativistic fields \cite{mancini_preserving_2014,jonsson_information_2015,jonsson_information_2016,jonsson_quantum_2017,jonsson_quantum_2014,jonsson_transmitting_2018}.

\section*{Acknowledgements}
The authors thank Dmytro Bondarenko, Tommaso Guaita, Terry Farrelly, and Tommaso Roscilde for inspiring discussions.
LH is supported by the the Max Planck Harvard Research Center for Quantum Optics and the Deutsche Forschungsgemeinschaft (DFG, German Research Foundation) under Germany’s Excellence Strategy – EXC-2111 – 39081486.
LH thanks the QMATH center for their hospitality during several visits related to this project.
RHJ acknowledges support by ERC Advanced grant 321029 and VILLUM FONDEN via the QMATH center of excellence (grant no.10059).

\appendix

\section{Standard forms of Gaussian states and the complex structure}\label{app:normalform}

The linear complex structure associated to a quantum state is a powerful method to parametrize Gaussian states. This is because most properties of the partial, mixed state  of subsystems,  such as the von Neumann entropy and the full entanglement spectrum, are encoded in the restriction of the complex structure to the relevant sub-phase space.
Therefore, this appendix uses the complex structure to rephrase some known properties of Gaussian states  in a largely basis-independent way. 
The key properties of the complex structure are reviewed and derived based on compatibility with the symplectic form (for bosons) or the positive-definite metric (for fermions). This appendix is in parts based on  the respective appendix in~\cite{hackl_aspects_2018}.

\begin{definition}
	Given a finite-dimensional real vector space $V$, we refer to a linear map $J: V\to V$ as complex structure if it satisfies $J^2=-\mathbb{1}$.
\end{definition}
The spectrum of a linear complex structure takes the following simple form.
\begin{proposition}
	A complex structure is diagonalizable over the complex numbers with eigenvalues that come in conjugate pairs as $\pm \ii$ with equal multiplicity. This also implies that $V$ must be even dimensional.
\end{proposition}
\begin{proof}
The equality $J^2=-\id$ implies that $J$ is diagonalizable and every eigenvalue squares to $-1$. Moreover, $J$ is a real map which means that all non-real eigenvalues come in complex conjugate pairs of equal multiplicity. Thus, the spectrum of $J$ is given by $\pm \ii$ with equal multiplicity equal to half of the dimension of $V$. The last part implies that such a linear map $J$ can only exist in an even dimensional real vector space $V$.
\end{proof}
We will be particularly interested in the restriction of linear complex structures.
\begin{proposition}\label{prop:J_decomp}
	Given a complex structure $J$ over $V$ and a direct sum decomposition of $V=A\oplus B$, we have the unique decomposition
	\begin{align}
	\begin{split}
	J=\left(\begin{array}{c|c}
	J_A & J_{AB}\\
	\hline
	J_{BA} & J_B
	\end{array}\right)\quad\text{with}\quad\\
	\begin{array}{rcc}
	J_A:& A\to A:& a\mapsto \pro_A(Ja)\,,\\
	J_B:& B\to B:& b\mapsto \pro_B(Jb)\,,\\
	J_{AB}:& B\to A:& b\mapsto \pro_A(Jb)\,,\\
	J_{BA}:& A\to B:& a\mapsto \pro_B(Ja)\,,
	\end{array}
	\end{split}
	\end{align}
	where the projections $\pro_A: V\to A$ and $\pro_B: V\to B$ with $\pro_A+\pro_B=\id$ provide the unique decomposition of a vector $v=a+b$ into its part in $a=p_A(v)\in A$ and $b=p_B(v)\in B$.\\
	Given such a decomposition, the spectrum of $J^2_A$ and $J^2_B$ is the same except for the number of eigenvalues equal to $-1$. Put differently, given an eigenvalue $\lambda$ of $J^2_A$, we have $\lambda=-1$ or $J^2_B$ also has $\lambda$ as eigenvalue. 
\end{proposition}
\begin{proof}
	We write $J^2=-\mathbb{1}$ in blocks to find
	\begin{align}\label{eq:J2_as_block_matrix}
	J^2&=\left(\begin{array}{c|c}
	J_A^2+J_{AB}J_{BA} & J_AJ_{AB}+J_{AB}J_B\\
	\hline
	J_BJ_{BA}+J_{BA}J_A & J_B^2+J_{BA}J_{AB}
	\end{array}\right)\nonumber\\
	&=\left(\begin{array}{c|c}
	-\mathbb{1}_A & 0\\
	\hline
	0 & -\mathbb{1}_B
	\end{array}\right)\,.
	\end{align}
	Let us consider the eigenvector $a\in A$ with $J_A^2\,a=\alpha\, a$. From the first diagonal block equation, we find that $a$ must also be eigenvector of $J_{AB}J_{BA}$ with $J_{AB}J_{BA}\,a=-(1+\alpha)a$. Let us define $b:=J_{BA}\,a$. Now there are two possibilities: Either $b=0$ which implies $\alpha=-1$ or $b\neq 0$, in which case we must have $J_{BA}J_{AB}\,b=-(1+\alpha)J_{BA}\,a=-(1+\alpha)\,b$. The second diagonal block equation then implies $J_B^2\,b=\alpha\,b$.
\end{proof}

\subsection{Compatible symplectic form (bosons)}\label{app:standardform-boson}
For bosonic systems, we are interested in linear complex structures that are compatible with the underlying symplectic form or vice versa.
\begin{definition}
	An antisymmetric, non-degenerate bilinear form $\omega:V\times V\to \mathbb{R}$ is called compatible to a complex structure $J: V\to V$ if and only if it satisfies the following conditions:
	\begin{itemize}
		\item[(I)] Skew-symmetry: The symplectic form satisfies $\omega(Jv,w)=-\omega(v,Jw)$.
		\item[(II)] Taming: The bilinear form $g(v,w):=\omega(v,Jw)$ is positive definite.
	\end{itemize}
	We will  refer to $\omega$ as a $J$-compatible symplectic form and to $J$ as a $\omega$-compatible complex structure.
\end{definition}
The compatibility conditions implies the following invariance property.
\begin{proposition}
	A $J$-compatible symplectic form $\omega$ is invariant under $J$ meaning
	\begin{align}
	    \omega(Jv,Jw)=\omega(v,w)
	\end{align}
	for all $v,w\in V$, thus  $J\in \mathrm{Sp}(V,\omega)$. 
	The bilinear form 
	\begin{align}\label{eq:g_defn_boson_app}
	    g(v,w):=\omega(v,Jw)
	\end{align}
	is symmetric and positive definite and thus a proper metric.
\end{proposition}
\begin{proof}
	Straightforward computation leads to
	\begin{align}
	\omega(Jv,Jw)%=-\omega(v,J^2w)
	=-\omega(v,-\mathbb{1}w)=\omega(v,w)\,.
	\end{align}
	Every linear map $M: V\to V$ that preserves the symplectic form $\omega$ in the sense of $\omega(Mv,Mw)=\omega(v,w)$ is part of the symplectic group $\mathrm{Sp}(V,\omega)$. 
	The form $g(v,w)=\omega(v,Jw)$ is symmetric due to
	\begin{align}
	g(v,w)=\omega(v,Jw)%=-\omega(Jw,v)
	=\omega(w,Jv)=g(w,v)\,.
	\end{align}
	The taming property (II) of $\omega$ ensures that $g$ is positive definite, and thus a proper positive definite metric.
\end{proof}

The compatibility conditions  characterize  the possible spectra of  restrictions of the  linear complex structure to symplectic subspaces.
\begin{proposition}\label{prop:Jomega_compatible}
	Let the  complex structure $J$ and the symplectic form $\omega$  on a  vector space $V$ be compatible, and let  $V=A\oplus B$ a decomposition of $V$ into a direct sum of even dimensional symplectic complements. 
	Then, the restricted complex structure $J_A$ is diagonalizable and its spectrum consists of complex conjugate pairs $\pm i c_i$ with $c_i\in [1,\infty)$. Both $J_A$ and $\omega_A$ can simultaneously be brought into the  block diagonal form
	\begin{align}
	J_A\equiv\bigoplus^{N_A}_{i=1}\left(\begin{array}{cc}
	0 & c_i\\
	-c_i & 0
	\end{array}\right)\,\text{and}\,\,\,
	\omega_A\equiv\bigoplus^{N_A}_{i=1}\left(\begin{array}{cc}
	0 & -1\\
	1 & 0
	\end{array}\right)\,. \label{eq:JA_omegaA_block_form}
	\end{align}
\end{proposition}
\begin{proof}
    The restricted linear complex structure $J_A$ is diagonalizable due to the fact that it squares to the diagonalizable matrix $J_A^2=-\mathbb{1}_A-J_{AB}J_{BA}$.
	Given a non-zero eigenvector $a\in A$ with $J_A^2\,a=\alpha\, a$, we have $\alpha\in (-\infty,-1]$ by the following argument: From
	\begin{align}
	\alpha \underbrace{g(a,a)}_{>0}=g(a,J_A^2a)=-\underbrace{g(J_Aa,J_Aa)}_{\geq 0}
	\end{align}
	follows $\alpha\leq 0$. Moreover, we can compute
	\begin{align}
	\begin{split}
    &g(J_Aa,J_Aa)=g(Ja-J_{BA}a,Ja-J_{BA}a)\\
	&=\underbrace{g(Ja,Ja)}_{=g(a,a)\geq 0}+\underbrace{g(J_{BA}a,J_{BA}a)}_{\geq0}-2\underbrace{g(Ja,J_{BA}a)}_{=0}\,,
	\end{split}
	\end{align}
	where we used in the last step that $A$ and $B$ are symplectic complements with $g(Ja,J_{BA}a)=\omega(Ja,JJ_{BA}a)=-\omega(J^2a,J_{BA}a)=\omega(a,J_{BA}a)=0$. We thus find  $g(J_Aa,J_Aa)\geq g(a,a)$ which implies $\alpha\in (-\infty,-1]$.\\
	From this, we find that $J_A$ must have eigenvalues appearing in complex conjugate pairs $\pm \ii c$ with $c=\sqrt{-\alpha}$, meaning that every eigenvalue $\alpha$ of $J_A^2$ appears with even multiplicity.
	Since $J_A$ is a linear real map with complex conjugated eigenvalues $\pm \ii c$, it can can be brought into the block diagonal form of  \eqref{eq:JA_omegaA_block_form}. In particular, in such a basis, $J_A^2$ is diagonal. Next, we show that $\omega_A$ simultaneously takes the form claimed above.\\
	While the full $J$ satisfies $\omega(Jv,Jw)=\omega(v,w)$ for $v,w\in V$, the restricted $J_A$ does in general not satisfy $\omega_A(J_A a,J_Aa')=\omega_A(a,a')$ for $a,a'\in A$. However, we can use the relation 
	$\omega(Jv,w)=-\omega(v,Jw)$ for $v,w\in V$ to derive the relation
	\begin{align}
	\begin{split}
	    \omega_A(J_Aa,a')&=\omega(Ja,a')=-\omega(a,Ja')\\
	    &\quad=-\omega_A(a,J_Aa')\,,
	\end{split}
	\end{align}
	which is the condition for $J_A$ to represent a symplectic algebra element, \ie $J_A\in \mathfrak{sp}(2N_A,\mathbb{R})$. For a diagonalizable symplectic algebra element with purely imaginary eigenvalues, we can always choose a basis, \eg see proposition~3.1.18 in~\cite{abraham1978foundations}, such that the matrix representations of $J_A$ and $\omega_A$ are given by~\eqref{eq:JA_omegaA_block_form}.
\end{proof}

This result is central for computing the entanglement spectrum of bosonic systems. In particular, the fact that the magnitude of possible eigenvalues lies in the interval $[1,\infty)$ indicates that the entanglement entropy in bosonic systems can be arbitrarily large. This can also be related to the non-compactness of symplectic group whose group elements relate different linear complex structures.

Before we construct the standard form for the matrix representation of $J$ for a general symplectic decomposition $V=A\oplus B$, we show that if $J_A^2=-\id_A$, \ie when the restricted complex structure defines a complex structure on the subspace $A$, then the complex structure leaves $A$ and $B$ invariant.
\begin{proposition}\label{prop:short_on_JA2_iscomplex}
    In the setup of the previous Proposition~\ref{prop:Jomega_compatible}, if $J_A^2=-\id$, then $J_{BA}=0$ vanishes, furthermore $J_B^2=-\id$ and $J_{AB}=0$.
\end{proposition}
\begin{proof}
    Because $V=A\oplus B$ is a symplectic decomposition,  for all $a\in A,b\in B$
    \begin{align}
        0&=\omega(a,b)=\omega(Ja,Jb)\nonumber\\
        &=\omega_A(Ja,Jb)+\omega_B(Ja,Jb)\nonumber\\
        &= \omega_A(J_A a,J_{AB}b)+\omega_B(J_{BA}a,J_{B} b)\,,
    \end{align}
    hence $\omega_A(J_{A}a,J_{AB}b)=-\omega_B(J_{BA}a,J_{B}b)$.
    For $J_A^2=-\id_A$, we find $J_{AB}J_{BA}=0$ due to $-\id=J^2$ and $-\id_A=J_A^2+J_{AB}J_{BA}$.
    To show that $J_{BA}=0$, let $b:=J_{BA}a\in B$ leading to
    \begin{align}
        &g(b,b)=\omega(b,Jb)=\omega_B(J_{BA}a,JJ_{BA}a)\nonumber\\
        &=\omega_B(J_{BA}a,J_{B}J_{BA}a)=-\omega_A(J_A a, J_{AB}J_{BA} a)\,.
    \end{align}
     Hence $g(b,b)=0$ which implies $b=0$ and since $a\in A$ was arbitrary, this means $J_{BA}=0$.\\
     Since $J_A^2=-\id$, by Proposition \ref{prop:J_decomp} we have $J_B^2=-\id$. Hence, by repeating the argument above, we obtain $J_{AB}=0$.
\end{proof}
The derived form for $J_A$ above induces a standard representation of $J$ which yields the standard form of the bosonic covariance matrix presented in~\eqref{eq:sta}.

\begin{proposition}
As before, let the complex structure $J$ and the symplectic form $\omega$ on a vector space $V$ be compatible, and let $V=A\oplus B$ be a decomposition of $V$ into a direct sum of even dimensional symplectic complements. Without loss of generality assume that the dimension of $A$ is lower than, or equal to the  dimension of $B$. Then there exists a symplectic basis with respect to which $J$ takes the form
\begin{widetext}
\begin{align}
    J\equiv \left( 
        \begin{array}{ccc|cccccc}
            \cosh(2r_1)\,\mathbb{A}_2 & \cdots & 0 & \sinh(2r_1)\,\mathbb{S}_2 & \cdots & 0 & 0& \cdots & 0 \\
             \vdots & \ddots & \vdots & \vdots &\ddots & \vdots &\vdots &\ddots & \vdots\\
             0 & \cdots & \cosh(2r_{N_A})\,\mathbb{A}_2 & 0 & \cdots & \sinh(2_{N_A})\,\mathbb{S}_2 & 0 &\cdots & 0\\ \hline
             \sinh(2r_1)\,\mathbb{S}_2 & \cdots & 0 & \cosh(2r_1)\,\mathbb{A}_2 & \cdots & 0 & 0 &\cdots & 0 \\
             \vdots & \ddots & \vdots & \vdots & \ddots &\vdots & \vdots & \ddots & \vdots \\
             0 & \cdots & \sinh(2r_{N_A})\,\mathbb{S}_2 & 0 & \cdots & \cosh(2 r_{N_A})\,\mathbb{A}_2 & 0 & \cdots & 0 \\
             0 & \cdots & 0 & 0 & \cdots & 0 &\mathbb{A}_2 & \cdots & 0\\
             \vdots & \ddots & \vdots & \vdots & \ddots & \vdots & \vdots & \ddots & \vdots\\
             0 & \cdots & 0 & 0 & \cdots & 0 & 0 &\cdots &   \mathbb{A}_2
        \end{array}
    \right)\,,
\end{align}
\end{widetext}
for some $r_i\geq0$, where $\mathbb{A}_2$ and $\mathbb{S}_2$ are given by
\begin{align}
    \mathbb{A}_2 &=\left(\begin{array}{cc}
	0 & 1\\
	-1 & 0
	\end{array}\right)\,,\quad 	\mathbb{S}_2 =\left(\begin{array}{cc}
	0 & 1\\
	1 & 0
	\end{array}\right)\,.\label{eq:A2andS2}
\end{align}
\end{proposition}
\begin{proof}
    The form of the upper diagonal block, representing $J_A$, follows from the  previous proposition, by choosing $\cosh(2r_i)=c_i$.
    Now for each $i\in\{1,\dots,N_A\}$, we denote the basis vectors bringing $J_A$ into this standard form by $a^{(i)}_1,a^{(i)}_2\in A$, \ie
    \begin{align}
        J_A a^{(i)}_2 = \cosh(2r_i) a_1\,,\quad  J_A a^{(i)}_1 = -\cosh(2r_i) a_2.
    \end{align}
    Now, if $r_i>0$, define
    \begin{align}
        b^{(i)}_1= \frac1{\sinh(2r_i)} J_{BA} a_2\,,\quad b^{(i)}_2=\frac1{\sinh(2r_i)}J_{BA}a_1\,.
    \end{align}
    These vectors lie in the eigenspace of $J_B^2$ with eigenvalue $-\cosh(2r_i)^2$. By definition, they bring the lower left off-diagonal block, representing $J_{BA}$  into the claimed form.
    Using $J_BJ_{BA}=-J_{BA}J_A$ which follows from \eqref{eq:J2_as_block_matrix}, we find that
    \begin{align}
         J_B b^{(i)}_1&= \frac{J_B J_{BA}}{\sinh(2r_i)} a_2 =-\frac{J_{BA}J_A a_2}{\sinh(2r_i)} \nonumber\\
         & =- \cosh(2r_i) \frac{J_{BA}a_1}{\sinh(2r_i)}=- \cosh(2r_i) b^{(i)}_2\,, \\
         J_B b^{(i)}_2& = %\frac1{\sinh(2r_i)}J_B J_{BA} a_1 =\frac{-1}{\sinh(2r_i)} J_{BA}J_A a_1= \cosh(2r_i) \frac1{\sinh(2r_i)}J_{BA}a_2  \nonumber\\&=
         \cosh(2r_i) b^{(i)}_1\,.
    \end{align}
    This means that the vectors also bring $J_B$ into standard form on the eigenspaces of $J_B^2$ of eigenvalue $-\cosh(2r_i)^2$. 
    They also bring $\omega$ into standard form, which follows from
    \begin{align}
        &\omega(b^{(i)}_1,b^{(i)}_2) = \frac{\omega\left(J_{BA}a_2,J_{BA}a_1\right)}{\sinh^2(2 r_i)} \nonumber\\
        & = \frac{\omega\left(J a_2,J_{BA}a_1\right)}{\sinh^2(2 r_i)}= \frac{-\omega\left(a_2,JJ_{BA}a_1\right)}{\sinh^2(2 r_i)}\nonumber\\
        &= \frac{-\omega\left(a_2,J_{AB}J_{BA}a_1\right)}{\sinh^2(2 r_i)}=-\omega(a_2,a_1)=-1\,,
    \end{align}
    where we used that $J_{AB}J_{BA}a_i=\sinh^2{2r_i}a_i$  which follows from $-\id= J_A^2+J_{AB}J_{BA}$ in \eqref{eq:J2_as_block_matrix}.
    This relation we also use to confirm that the defined vectors bring the lower off-diagonal block, representing $J_{BA}$ into the claimed standard form:
    \begin{align}
        &J_{AB}b^{(i)}_1%= \frac{J_{AB}J_{BA}a^{(i)}_2}{\sinh(2r_i)}
        = \sinh(2r_i) a_2\,,
        %\nonumber\\&
        \quad J_{AB}b_2=\sinh(2r_i) a_1\,.
    \end{align}
    In the case where $r_i=0$, Proposition \ref{prop:short_on_JA2_iscomplex} ensures that the matrix entries of off-diagonal blocks vanish.
    Hence in the subspaces of $B$ corresponding to $r_i=0$, and in the remaining dimensions of $B$, we can choose any basis which brings $J_B$ in its standard form. Thus, we obtain the claimed matrix representation for $J$.
    Finally, this result implies the standard form~\eqref{eq:sta} of the bosonic covariance matrix $G$ by applying $G=-J\Omega$ with $\Omega\equiv\mat\Omega$ from~\eqref{eq:staOmega2}.
\end{proof}

To conclude this subsection, let us comment on the connection to the standard form of the bosonic covariance matrix and partner formula of Section \ref{sec:boson_review}.
The symplectic form $\omega$ and the metric $g$ used in this section act on the phase space $V$. In fact, they are the inverses of $\Omega^{ab}$ and $G^{ab}$ used in the main body of the article which act on $V^*$, the co-vector space of the phase space.

With respect to a basis in which $J$ is represented by the matrix form which we just derived, $G$ is then represented, due to equation \eqref{eq:boson_J_defn}, by the matrix $\mat G=- \mat J\mat \Omega$ which is exactly as in  \eqref{eq:sta}.

Each pair of vectors $a_1^{(i)},a_1^{(i)}\in V$, appearing  in the proof above, is dual to a pair of covectors $x,k\in V^*$ which defines a mode. From the proof it follows that the mode's partner mode is defined by the covectors $(\bar x,\bar k)$ given by
\begin{align}
    \bar x&= \frac1{\sinh(2r_i)}J^\intercal_{BA} k = \frac1{\sinh(2r_i)}\left( -J_A^\intercal+ J^\intercal\right) k \nonumber\\
    &= \frac1{\sinh(2r_i)} \left(\cosh(2 r_i) x+J^\intercal k\right)\,, \\
    \bar k%= \frac1{\sinh(2r_i)}J^\intercal_{BA} x = \frac1{\sinh(2r_i)}\left( -J_A^\intercal+ J^\intercal\right) x \nonumber\\&
    &= \frac1{\sinh(2r_i)} \left(-\cosh(2 r_i) k+J^\intercal x\right)\,,
\end{align}
which is precisely the partner formula \eqref{eq:partnermodeformula}.

\subsection{Compatible metric (fermions)}\label{app:standardform-fermions}
For fermionic systems the analysis of the complex structure is very similar. Here, we study the properties of a linear complex structure which is compatible with a positive definite metric.
\begin{definition}
	A positive definite symmetric bilinear form $g:V\times V\to \mathbb{R}$ is called compatible to a complex structure $J: V\to V$ if and only if it satisfies
	\begin{align}
	g(Jv,w)=-g(v,Jw)\quad\text{for all}\quad v,w\in V\,.
	\end{align}
	We will refer to $g$ as a $J$-compatible metric and to $J$ as a $g$-compatible complex structure.
\end{definition}
The compatibility conditions implies the following invariance property.
\begin{proposition}
	A $J$-compatible metric $g$ is invariant under $J$ meaning $g(Jv,Jw)=g(v,w)$ for all $v,w\in V$ which implies  $J\in O(V,g)$. Moreover, we can define the antisymmetric bilinear form $\omega(v,w)=g(Jv,w)$ which is a well-defined symplectic form on $V$.
\end{proposition}
\begin{proof}
	Straightforward computation gives
	\begin{align}
	&g(Jv,Jw)=-g(v,J^2w)=-g(v,-\mathbb{1}w)\nonumber\\
	&=g(v,w)\,.
	\end{align}
	Every linear map $M: V\to V$ that preserves metric $g$ in the sense of $g(Mv,Mw)=g(v,w)$ is an element of the orthogonal group $O(V,g)$.	A symplectic form on a vector space $V$ is required to be (a) antisymmetric and (b) non-degenerate:
	\begin{itemize}
		\item[(a)] Antisymmetry: We compute $\omega(v,w)=g(Jv,w)=g(w,Jv)=-g(Jw,v)=-\omega(w,v)$.
		\item[(b)] Non-Degeneracy: We need to show that for every non-zero vector $v$, there is at least one vector $w$, such that $\omega(v,w)\neq0$. This can easily be done by choosing $w=Jv$, where we find $\omega(v,w)=g(Jv,Jv)>0$.
	\end{itemize}
	This proves that every compatible every pair of a complex structure $J$ and a metric $g$ defines a symplectic form $\omega$.
\end{proof}
Again, the compatibility condition characterizes the possible spectra of restrictions of the linear complex structure  to even dimensional subspaces.
\begin{proposition}\label{prop:JA_standard_fermion}
	Given a compatible pair of a complex structure $J$ and a metric $g$, we can decompose the vector space $V$ into a direct sum of orthogonal sum of orthogonal complements $V=A\oplus B$. In this case, the restricted complex structure $J_A$ is diagonalizable and its spectrum consists of complex conjugate pairs $\pm i c_i$ with $c_i\in [0,1]$. If $A$ is odd-dimensional, $J_A$ also has the eigenvalue $0$. Provided that we choose a decomposition into even-dimensional orthogonal complements with dimensions $2N_A$, we can bring $J_A$ and $g_A$ simultaneously in the following block diagonal form:
	\begin{align}
	J_A\equiv\bigoplus^{N_A}_{i=1}\left(\begin{array}{cc}
	0 & c_i\\
	-c_i & 0
	\end{array}\right)\,\text{and}\,\,\, g_A\equiv\bigoplus^{N_A}_{i=1}\left(\begin{array}{cc}
	1 & 0\\
	0 & 1
	\end{array}\right)\,.\label{eq:JA_gA_block_diag_form}
	\end{align}
\end{proposition}
\begin{proof}
	The restricted complex structure $J_A$ is itself anti-hermitian with respect to the metric $g_A$ because we have for $a_1,a_2\in A$
	\begin{align}
	g_A(J_Aa_1,a_2)&%=g(Ja_1,a_2)
	=-g(a_1,Ja_2)=-g_A(a_1,J_Aa_2)\,.
	\end{align}
	Therefore, $J_A$ is normal with respect to $g_A$. This ensures that $J_A$ is diagonalizable (with a complete set of eigenvectors) and so is $J_A^2$.
	Given a non-zero eigenvector $a\in A$ with $J_A^2\,a=\alpha\, a$, we have $\alpha\in [-1,0]$ due to the following argument. Let us compute
	\begin{align}
	\alpha \underbrace{g(a,a)}_{>0}=g(a,J_A^2a)=-\underbrace{g(J_Aa,J_Aa)}_{\geq 0}\,.
	\end{align}
	This already implies that $\alpha\leq 0$. Moreover, we can also compute
	\begin{align}
	\begin{split}
	g(a,a)%=g(Ja,Ja)
	&=g(J_Aa+J_{BA}a,J_Aa+J_{BA}a)\\
	&=g(J_Aa,J_Aa)+g(J_{BA}a,J_{BA}a)\,,
	\end{split}
	\end{align}
	where we used in the last step that $A$ and $B$ are orthogonal which elimates crossing terms. This equation implies the inequality $g(a,a)\geq g(J_Aa,J_Aa)=|\alpha|g(a,a)$ leading to $\alpha\in [-1,0]$. From here, we find that $J_A$ must have eigenvalues appearing in complex conjugate pairs $\pm \ii c$ with $c=\sqrt{-\alpha}$ meaning that also every eigenvalue $\alpha$ of $J_A^2$ appears with even multiplicity unless $\alpha=0$.\\ 
	While the full $J$ satisfies $g(Jv,Jw)=g(v,w)$ for $v,w\in V$, the restricted $J_A$ does in general not satisfy $g_A(J_A a,J_Aa')=\omega_A(a,a')$ for $a,a'\in A$. However, we can use the relation $g(Jv,w)=-g(v,Jw)$ for $v,w\in V$ to derive the relation
	\begin{align}
	\begin{split}
	    g_A(J_Aa,a')&=g(Ja,a')=-g(a,Ja')\\
	    &=g_A(a,J_Aa')\,,
	\end{split}
	\end{align}
	which is well-known to be the condition for $J_A$ to represent an orthogonal algebra element, \ie $J_A\in \mathfrak{so}(2N_A)$. An orthogonal algebra element is antisymmetric with respect to $g_A$ and has purely imaginary eigenvalues. It is well-known that we can always choose an orthonormal basis, such that the matrix representations of $J_A$ and $g_A$ are given by~\eqref{eq:JA_gA_block_diag_form}.
\end{proof}

\begin{proposition}
In the setup of Proposition \ref{prop:JA_standard_fermion}, if $J_A^2=-\id$ this implies $J_{BA}=0$, furthermore $J_B^2=-\id$ and $J_{AB}=0$.
\end{proposition}
\begin{proof}
For any two vectors $a,a'\in A$, we have 
\begin{align}
g(J_{BA} a, J_{BA} a')= -g(a, J_{AB}J_{BA} a')\,,
\end{align}
which is shown using $g(v,Jw)=-g(Jv,w)$ and the fact that $A$ and $B$ are orthogonal subspaces:
\begin{align}
    &g(J_{BA}a,J_{BA}a')=g(Ja,J_{BA}a')\nonumber\\
    &=-g(a,JJ_{BA}a')=-g(a,J_{AB}J_{BA}a')\,.
\end{align}
If $J_A^2=-\id$, then from \eqref{eq:J2_as_block_matrix} follows $J_{AB}J_{BA}=-\id-J_A^2=0$.
This however means that for all $a\in A$ we have
\begin{align}
    &g(J_{BA}a,J_{BA}a)= -g(a,J_{AB}J_{BA}a)=0.
\end{align}
Since $g$ is non-degenerate this implies $J_{BA}=0$.
By Proposition \ref{prop:J_decomp}, $J_B^2=-\id$ follows from $J_A^2=-\id$, then repeating the argument above shows $J_{AB}=0$.
\end{proof}

This result is central for computing the entanglement spectrum of fermionic systems. In particular, the fact that the magnitude of possible eigenvalues lies in the interval $[0,1]$ implies that the entanglement entropy of fermionic systems is bounded by $\log{2}$ per fermionic degree of freedom.

\begin{proposition}
As before, let the complex structure $J$ and the metric $g$ on a vector space $V$ be compatible, and let $V=A\oplus B$ be a decomposition of $V$ into even-dimensional orthogonal complements. Without loss of generality, assume that the dimension of $A$ is  lower than, or equal to the  dimension of $B$.
Then there exists an orthonormal basis with respect to which $J$ takes the form
\begin{widetext}
\begin{align}
	J\equiv\left(\begin{array}{ccc|cccccc}
	\cos(2r_1)\,\mathbb{A}_2 & \cdots & 0 & \sin(2r_{1})\,\mathbb{S}_2 & \cdots & 0 & 0 &\cdots & 0\\
	\vdots & \ddots & \vdots & \vdots & \ddots & \vdots & \vdots & \ddots & \vdots \\
	0 & \cdots & \cos(2r_{N_A})\,\mathbb{A}_2 & 0 & \cdots & \sin(2r_{N_A})\,\mathbb{S}_2 & 0 & \cdots & 0\\
	\hline
	-\sin(2r_1)\,\mathbb{S}_2 & \cdots & 0 & \cos(2r_1)\,\mathbb{A}_2 & \cdots & 0 & 0 &\cdots & 0\\
	\vdots & \ddots & \vdots & \vdots & \ddots & \vdots & \vdots & \ddots & \vdots\\
	0 & \cdots & -\sin(2r_{N_A})\,\mathbb{S}_2 & 0 & \cdots & \cos(2r_{N_A})\,\mathbb{A}_2 & 0 &\cdots & 0\\
	0 & \cdots & 0 & 0 & \cdots & 0 & \mathbb{A}_2 & \cdots & 0\\
	\vdots & \ddots & \vdots & \vdots & \ddots & \vdots & \vdots & \ddots & \vdots \\
	0 & \cdots & 0 & 0 & \cdots & 0 & 0 & \cdots & \mathbb{A}_2
	\end{array}\right)\,,
\end{align}
\end{widetext}
with $r_i\in[0,\frac\pi4]$ and the two-by-two matrices $\mathbb{A}_2$ and $\mathbb{S}_2$ defined in~\eqref{eq:A2andS2}.
\end{proposition}
\begin{proof}
By Proposition \ref{prop:JA_standard_fermion} the upper left diagonal block, representing $J_A$, can be brought into the claimed form by choosing  $\cos2r_i=c_i$. Let $a^{(i)}_1,\,a^{(i)}_2$ denote the pair of basis vectors spanning the space corresponding to $r_i$. Then, if $r_i>0$, define the vectors
\begin{align}
    b^{(i)}_1=\frac1{\sin2r_i}J_{BA}a^{(i)}_2\,,\quad b^{(i)}_2=\frac1{\sin2r_i}J_{BA}a^{(i)}_1\,.
\end{align}
These vectors are orthonormal
\begin{align}
    g\left(b^{(i)}_1,b^{(i)}_2\right)&=\frac{-g\left(a_2^{(i)}, J_{AB}J_{BA} a_1^{(i)}\right)}{\sin^22r_i}\nonumber\\
    &=g\left(a_2^{(i)}, a_1^{(i)}\right)=0\,,\\
    g\left(b^{(i)}_1,b^{(i)}_1\right)&=\frac{-g\left(a_2^{(i)}, J_{AB}J_{BA} a_2^{(i)}\right)}{\sin^22r_i}\nonumber\\
    &=g\left(a_2^{(i)}, a_2^{(i)}\right)=1\,,\\
    g(b^{(i)}_2,b^{(i)}_2)&=1\,,
\end{align}
where we used $J_{AB}J_{BA}a=(-1+\cos^22r_i)a=-\sin^2(2r_i) a$ which follows from~%$-\id-J_A^2=J_{AB}J_{BA}$ due to
\eqref{eq:J2_as_block_matrix}. Similarly, one also shows that $g(b^{(i)}_n,b^{(j)}_m)=0$ if $i\neq j$.
By definition, these vectors bring the lower left diagonal block, representing $J_{BA}$ into the claimed standard form.
To see that the vectors also bring the diagonal blocks of $J_B$ into the claimed form, we calculate
\begin{align}
    J_B b_1^{(i)} &= \frac{J_BJ_{BA} a_2^{(i)}}{\sin2r_i}=\frac{-J_{BA}J_{A} a_2^{(i)}}{\sin2r_i}\nonumber\\
    &=\frac{-\cos(2r_i) J_{BA}a_1^{(i)}}{\sin2r_i}=-\cos2r_i b^{(i)}_2\,,\\
    J_B b_2^{(i)} %= \frac{J_BJ_{BA} a_1^{(i)}}{\sin2r_i}=\frac{-J_{BA}J_{A} a_1^{(i)}}{\sin2r_i}\nonumber\\&=\frac{\cos(2r_i) J_{BA}a_2^{(i)}}{\sin2r_i}
    &=\cos2r_i b^{(i)}_1\,.
\end{align}
Hence, by extending the set of all these vectors with an orthonormal set of vectors which brings the eigenspace of $J_B^2$ with eigenvalue $-1$ into the standard form, we obtain an orthonormal basis of $B$ which brings the  matrix representing $J$ into the claimed form.\\
Finally, this result implies the standard form~\eqref{eq:standardform-fermions} of the fermionic covariance matrix $\Omega$ by applying $\Omega=JG$ with $G\equiv\mathbb{1}$.
\end{proof}
The propositions of these section imply both the normal form of the covariance matrix for fermions \eqref{eq:sta_ferm}, as well as the fermionic partner mode. In fact, the matrix $\mat \Omega$ representing the fermionic covariance matrix $\Omega^{ab}$ and the matrix $\mat J$ representing the complex structure $J^a{}_B$ are in fact identical. This is due to $\mat \Omega=\mat J\mat G$ where the metric $G$, which is the inverse of the metric $g$ used in this subsection, is represented by the identity matrix $\mat G=\id$.
The partner formula for fermions \eqref{eq:partnermode_formula_ferm} follows by the same argument as discussed above for bosons.

\section{Spectra of restricted Hamiltonian}\label{app:spectrumproof}

Here, we show how the spectrum of the restriction of a Hamiltonian to a subspace of partner modes, is related to the spectrum of the unrestricted Hamiltonian operator, as a consequence of the Courant–Fischer–Weyl min-max principle.

\begin{proposition}
Let $\hat H$ be a quadratic Hamiltonian on a (bosonic or fermionic) system of $N$ modes, and denote the excitation energies of $\hat H$ by $0<\omega_1\leq\omega_2\leq\dots\leq\omega_N$. 
Assuming that $\hat H_A$ is the restriction of $\hat H$ onto a subsystem $A$ spanned by $N_A<N$ modes such that the ground state $\ket0$ of $\hat H$ is a product state
\begin{align}
    \ket0=\ket\psi_A\otimes\ket\phi_B
\end{align}
between $A$ and its complement $B$, the excitation energies $\epsilon_1\leq\dots\leq\epsilon_{N_A}$ of $\hat H_A$ fulfill 
\begin{align}
    \omega_{i}\leq \epsilon_i\leq \omega_{N-N_A+i}\,.
\end{align}
\end{proposition}

\begin{proof}
We note that the excitation energies of a quadratic Hamiltonian $\hat H=\frac12h_{ab}\hat\xi^a\hat\xi^b+f_a\hat\xi^a$ are given by the eigenvalues of the linear map $L^a{}_c=G^{ab}h_{bc}$ obtained by composing the covariance matrix $G^{ab}$ of the Hamiltonian's ground state with the Hamiltonian bilinear form $h_{ab}$.
This is easily seen by observing that with respect to a basis of energy eigenmodes, where $G^{ab}\equiv\id$ is represented by the identity matrix, both $L$ and $h$ are represented by the same diagonal matrix.

In general, in a basis which is adapted to the bipartition of the $N$ modes into subsystems $A$ and $B$, the map $L$ may not be represented by a diagonal matrix. However, since we assume the ground state of $\hat H$ to be a product state between $A$ and $B$, the covariance matrix takes block diagonal form, and we find that $L$ is represented by the matrix
\begin{align}
    L\equiv \mat L= \mat G\mat h&= \left( \begin{array}{c|c}
         \mat G_A& 0  \\ \hline
         0 & \mat G_B
    \end{array}\right)
    \left( \begin{array}{c|c}
         \mat h_A& \mat h_{AB}  \\ \hline
         \mat h_{AB}^\intercal & \mat h_B
    \end{array}\right)\nonumber\\
    &=
    \left( \begin{array}{c|c}
         \mat G_A \mat h_A& \mat G_A\mat h_{AB}  \\ \hline
         \mat G_B\mat h_{AB}^\intercal & \mat G_B\mat h_B
    \end{array}\right)\,.
\end{align}
The min-max principle now warrants that the eigenvalues $\omega_1\leq\dots\leq\omega_N$ of $\mat L$ (with double multiplicity) bound the eigenvalues of $\mat G_A\mat h_A$, which we denote by $\alpha_1\leq\dots\leq\alpha_{2N_A}$, as
\begin{align}
    \omega_1\leq \alpha_1\leq \omega_{N-N_A+1}\,,\quad \omega_{N_A} \leq {\alpha_{2N_A}} \leq\omega_{N}\,.
\end{align}
The $\alpha_i$ with double multiplicity represent the excitation energies of $\hat H_A$. We can see this from the ground state of $\hat H$, that is a product state between $A$ and $B$ such that $\ket \psi_A$ must be the ground state of $\hat H_A$ and its covariance matrix is represented by the matrix $\mat G_A$. Therefore, the spectrum of $L_A=\mat G_A\mat h_A$ yields the excitation energies of $\hat H_A$.

The equivalent statement for a fermionic Hamiltonian $\hat H=\frac\ii2 h_{ab}\hat\xi^a\hat\xi^b$ follows from the same line of argument, but the linear map $L$ is now defined by $L^a{}_c=-\Omega^{ab} h_{bc}$, where $\Omega^{ab}$ represents the covariance matrix of the ground state of the fermionic Hamiltonian $\hat{H}=\frac{\ii}{2}h_{ab}\hat{\xi}^a\hat{\xi}^b$.
\end{proof}

\bibliographystyle{unsrtnat}
\bibliography{minimalenergycost_refs}

\begin{thebibliography}{89}
\providecommand{\natexlab}[1]{#1}
\providecommand{\url}[1]{\texttt{#1}}
\expandafter\ifx\csname urlstyle\endcsname\relax
  \providecommand{\doi}[1]{doi: #1}\else
  \providecommand{\doi}{doi: \begingroup \urlstyle{rm}\Url}\fi

\bibitem[Einstein et~al.(1935)Einstein, Podolsky, and Rosen]{einstein_can_1935}
A.~Einstein, B.~Podolsky, and N.~Rosen.
\newblock Can {{Quantum}}-{{Mechanical Description}} of {{Physical Reality Be
  Considered Complete}}?
\newblock \emph{Phys. Rev.}, 47\penalty0 (10):\penalty0 777--780, May 1935.
\newblock \doi{10.1103/PhysRev.47.777}.
\newblock URL \url{https://link.aps.org/doi/10.1103/PhysRev.47.777}.

\bibitem[Bell(1964)]{bell_einstein_1964}
J.~S. Bell.
\newblock On the {{Einstein Podolsky Rosen}} paradox.
\newblock \emph{Physics Physique Fizika}, 1\penalty0 (3):\penalty0 195--200,
  November 1964.
\newblock \doi{10.1103/PhysicsPhysiqueFizika.1.195}.
\newblock URL
  \url{https://link.aps.org/doi/10.1103/PhysicsPhysiqueFizika.1.195}.

\bibitem[Bell(2001)]{bell_einstein_2001}
J.~S. Bell.
\newblock On the einstein podolsky rosen paradox.
\newblock In \emph{John {{S Bell}} on the {{Foundations}} of {{Quantum
  Mechanics}}}, pages 7--12. {WORLD SCIENTIFIC}, August 2001.
\newblock ISBN 978-981-02-4687-7.
\newblock \doi{10.1142/9789812386540_0002}.
\newblock URL
  \url{https://www.worldscientific.com/doi/abs/10.1142/9789812386540_0002}.

\bibitem[Hensen et~al.(2015)Hensen, Bernien, Dr{\'e}au, Reiserer, Kalb, Blok,
  Ruitenberg, Vermeulen, Schouten, Abell{\'a}n, Amaya, Pruneri, Mitchell,
  Markham, Twitchen, Elkouss, Wehner, Taminiau, and
  Hanson]{hensen_loophole-free_2015}
B.~Hensen, H.~Bernien, A.~E. Dr{\'e}au, A.~Reiserer, N.~Kalb, M.~S. Blok,
  J.~Ruitenberg, R.~F.~L. Vermeulen, R.~N. Schouten, C.~Abell{\'a}n, W.~Amaya,
  V.~Pruneri, M.~W. Mitchell, M.~Markham, D.~J. Twitchen, D.~Elkouss,
  S.~Wehner, T.~H. Taminiau, and R.~Hanson.
\newblock Loophole-free {{Bell}} inequality violation using electron spins
  separated by 1.3 kilometres.
\newblock \emph{Nature}, 526\penalty0 (7575):\penalty0 682--686, October 2015.
\newblock ISSN 1476-4687.
\newblock \doi{10.1038/nature15759}.
\newblock URL \url{https://www.nature.com/articles/nature15759}.

\bibitem[Giustina et~al.(2015)Giustina, Versteegh, Wengerowsky, Handsteiner,
  Hochrainer, Phelan, Steinlechner, Kofler, Larsson, Abell{\'a}n, Amaya,
  Pruneri, Mitchell, Beyer, Gerrits, Lita, Shalm, Nam, Scheidl, Ursin,
  Wittmann, and Zeilinger]{giustina_significant-loophole-free_2015}
Marissa Giustina, Marijn A.~M. Versteegh, S{\"o}ren Wengerowsky, Johannes
  Handsteiner, Armin Hochrainer, Kevin Phelan, Fabian Steinlechner, Johannes
  Kofler, Jan-{\AA}ke Larsson, Carlos Abell{\'a}n, Waldimar Amaya, Valerio
  Pruneri, Morgan~W. Mitchell, J{\"o}rn Beyer, Thomas Gerrits, Adriana~E. Lita,
  Lynden~K. Shalm, Sae~Woo Nam, Thomas Scheidl, Rupert Ursin, Bernhard
  Wittmann, and Anton Zeilinger.
\newblock Significant-{{Loophole}}-{{Free Test}} of {{Bell}}'s {{Theorem}} with
  {{Entangled Photons}}.
\newblock \emph{Phys. Rev. Lett.}, 115\penalty0 (25):\penalty0 250401, December
  2015.
\newblock \doi{10.1103/PhysRevLett.115.250401}.
\newblock URL \url{https://link.aps.org/doi/10.1103/PhysRevLett.115.250401}.

\bibitem[Shalm et~al.(2015)Shalm, {Meyer-Scott}, Christensen, Bierhorst, Wayne,
  Stevens, Gerrits, Glancy, Hamel, Allman, Coakley, Dyer, Hodge, Lita, Verma,
  Lambrocco, Tortorici, Migdall, Zhang, Kumor, Farr, Marsili, Shaw, Stern,
  Abell{\'a}n, Amaya, Pruneri, Jennewein, Mitchell, Kwiat, Bienfang, Mirin,
  Knill, and Nam]{shalm_strong_2015}
Lynden~K. Shalm, Evan {Meyer-Scott}, Bradley~G. Christensen, Peter Bierhorst,
  Michael~A. Wayne, Martin~J. Stevens, Thomas Gerrits, Scott Glancy, Deny~R.
  Hamel, Michael~S. Allman, Kevin~J. Coakley, Shellee~D. Dyer, Carson Hodge,
  Adriana~E. Lita, Varun~B. Verma, Camilla Lambrocco, Edward Tortorici, Alan~L.
  Migdall, Yanbao Zhang, Daniel~R. Kumor, William~H. Farr, Francesco Marsili,
  Matthew~D. Shaw, Jeffrey~A. Stern, Carlos Abell{\'a}n, Waldimar Amaya,
  Valerio Pruneri, Thomas Jennewein, Morgan~W. Mitchell, Paul~G. Kwiat,
  Joshua~C. Bienfang, Richard~P. Mirin, Emanuel Knill, and Sae~Woo Nam.
\newblock Strong {{Loophole}}-{{Free Test}} of {{Local Realism}}.
\newblock \emph{Phys. Rev. Lett.}, 115\penalty0 (25):\penalty0 250402, December
  2015.
\newblock \doi{10.1103/PhysRevLett.115.250402}.
\newblock URL \url{https://link.aps.org/doi/10.1103/PhysRevLett.115.250402}.

\bibitem[Horodecki et~al.(2009)Horodecki, Horodecki, Horodecki, and
  Horodecki]{horodecki_quantum_2009}
Ryszard Horodecki, Pawe{\l} Horodecki, Micha{\l} Horodecki, and Karol
  Horodecki.
\newblock Quantum entanglement.
\newblock \emph{Rev. Mod. Phys.}, 81\penalty0 (2):\penalty0 865--942, June
  2009.
\newblock \doi{10.1103/RevModPhys.81.865}.
\newblock URL \url{https://link.aps.org/doi/10.1103/RevModPhys.81.865}.

\bibitem[Bennett and Wiesner(1992)]{bennett_communication_1992}
Charles~H. Bennett and Stephen~J. Wiesner.
\newblock Communication via one- and two-particle operators on
  {{Einstein}}-{{Podolsky}}-{{Rosen}} states.
\newblock \emph{Phys. Rev. Lett.}, 69\penalty0 (20):\penalty0 2881--2884,
  November 1992.
\newblock \doi{10.1103/PhysRevLett.69.2881}.
\newblock URL \url{https://link.aps.org/doi/10.1103/PhysRevLett.69.2881}.

\bibitem[Bennett et~al.(1993)Bennett, Brassard, Cr{\'e}peau, Jozsa, Peres, and
  Wootters]{bennett_teleporting_1993}
Charles~H. Bennett, Gilles Brassard, Claude Cr{\'e}peau, Richard Jozsa, Asher
  Peres, and William~K. Wootters.
\newblock Teleporting an unknown quantum state via dual classical and
  {{Einstein}}-{{Podolsky}}-{{Rosen}} channels.
\newblock \emph{Phys. Rev. Lett.}, 70\penalty0 (13):\penalty0 1895--1899, March
  1993.
\newblock \doi{10.1103/PhysRevLett.70.1895}.
\newblock URL \url{https://link.aps.org/doi/10.1103/PhysRevLett.70.1895}.

\bibitem[Bennett and Brassard(2014)]{bennett_quantum_2014}
Charles~H. Bennett and Gilles Brassard.
\newblock Quantum cryptography: {{Public}} key distribution and coin tossing.
\newblock \emph{Theoretical Computer Science}, 560:\penalty0 7--11, December
  2014.
\newblock ISSN 0304-3975.
\newblock \doi{10.1016/j.tcs.2014.05.025}.
\newblock URL
  \url{http://www.sciencedirect.com/science/article/pii/S0304397514004241}.

\bibitem[Srednicki(1993)]{srednicki_entropy_1993}
Mark Srednicki.
\newblock Entropy and area.
\newblock \emph{Phys. Rev. Lett.}, 71\penalty0 (5):\penalty0 666--669, August
  1993.
\newblock \doi{10.1103/PhysRevLett.71.666}.
\newblock URL \url{https://link.aps.org/doi/10.1103/PhysRevLett.71.666}.

\bibitem[Vidal et~al.(2003)Vidal, Latorre, Rico, and
  Kitaev]{vidal_entanglement_2003}
G.~Vidal, J.~I. Latorre, E.~Rico, and A.~Kitaev.
\newblock Entanglement in {{Quantum Critical Phenomena}}.
\newblock \emph{Phys. Rev. Lett.}, 90\penalty0 (22):\penalty0 227902, June
  2003.
\newblock \doi{10.1103/PhysRevLett.90.227902}.
\newblock URL \url{https://link.aps.org/doi/10.1103/PhysRevLett.90.227902}.

\bibitem[Osborne and Nielsen(2002)]{osborne_entanglement_2002}
Tobias~J. Osborne and Michael~A. Nielsen.
\newblock Entanglement in a simple quantum phase transition.
\newblock \emph{Phys. Rev. A}, 66\penalty0 (3):\penalty0 032110, September
  2002.
\newblock \doi{10.1103/PhysRevA.66.032110}.
\newblock URL \url{https://link.aps.org/doi/10.1103/PhysRevA.66.032110}.

\bibitem[Eisert et~al.(2010)Eisert, Cramer, and Plenio]{eisert_colloquium_2010}
J.~Eisert, M.~Cramer, and M.~B. Plenio.
\newblock Colloquium: {{Area}} laws for the entanglement entropy.
\newblock \emph{Rev. Mod. Phys.}, 82\penalty0 (1):\penalty0 277--306, February
  2010.
\newblock \doi{10.1103/RevModPhys.82.277}.
\newblock URL \url{https://link.aps.org/doi/10.1103/RevModPhys.82.277}.

\bibitem[Amico et~al.(2008)Amico, Fazio, Osterloh, and
  Vedral]{amico_entanglement_2008}
Luigi Amico, Rosario Fazio, Andreas Osterloh, and Vlatko Vedral.
\newblock Entanglement in many-body systems.
\newblock \emph{Rev. Mod. Phys.}, 80\penalty0 (2):\penalty0 517--576, May 2008.
\newblock \doi{10.1103/RevModPhys.80.517}.
\newblock URL \url{https://link.aps.org/doi/10.1103/RevModPhys.80.517}.

\bibitem[Islam et~al.(2015)Islam, Ma, Preiss, Eric~Tai, Lukin, Rispoli, and
  Greiner]{islam_measuring_2015}
Rajibul Islam, Ruichao Ma, Philipp~M. Preiss, M.~Eric~Tai, Alexander Lukin,
  Matthew Rispoli, and Markus Greiner.
\newblock Measuring entanglement entropy in a quantum many-body system.
\newblock \emph{Nature}, 528\penalty0 (7580):\penalty0 77--83, December 2015.
\newblock ISSN 1476-4687.
\newblock \doi{10.1038/nature15750}.
\newblock URL \url{https://www.nature.com/articles/nature15750}.

\bibitem[Ryu and Takayanagi(2006)]{ryu_holographic_2006}
Shinsei Ryu and Tadashi Takayanagi.
\newblock Holographic {{Derivation}} of {{Entanglement Entropy}} from the
  anti--de {{Sitter Space}}/{{Conformal Field Theory Correspondence}}.
\newblock \emph{Phys. Rev. Lett.}, 96\penalty0 (18):\penalty0 181602, May 2006.
\newblock \doi{10.1103/PhysRevLett.96.181602}.
\newblock URL \url{https://link.aps.org/doi/10.1103/PhysRevLett.96.181602}.

\bibitem[Summers and Werner(1985)]{summers_vacuum_1985}
Stephen~J. Summers and Reinhard Werner.
\newblock The vacuum violates {{Bell}}'s inequalities.
\newblock \emph{Physics Letters A}, 110\penalty0 (5):\penalty0 257--259, July
  1985.
\newblock ISSN 0375-9601.
\newblock \doi{10.1016/0375-9601(85)90093-3}.
\newblock URL
  \url{http://www.sciencedirect.com/science/article/pii/0375960185900933}.

\bibitem[Summers and Werner(1987{\natexlab{a}})]{summers_bells_1987}
Stephen~J. Summers and Reinhard Werner.
\newblock Bell's inequalities and quantum field theory. {{II}}. {{Bell}}'s
  inequalities are maximally violated in the vacuum.
\newblock \emph{J. Math. Phys.}, 28\penalty0 (10):\penalty0 2448--2456, October
  1987{\natexlab{a}}.
\newblock ISSN 0022-2488, 1089-7658.
\newblock \doi{10.1063/1.527734}.
\newblock URL
  \url{http://scitation.aip.org/content/aip/journal/jmp/28/10/10.1063/1.527734}.

\bibitem[Summers and Werner(1987{\natexlab{b}})]{summers_bells_1987a}
Stephen~J. Summers and Reinhard Werner.
\newblock Bell's inequalities and quantum field theory. {{I}}. {{General}}
  setting.
\newblock \emph{Journal of Mathematical Physics}, 28\penalty0 (10):\penalty0
  2440--2447, October 1987{\natexlab{b}}.
\newblock ISSN 0022-2488.
\newblock \doi{10.1063/1.527733}.
\newblock URL \url{http://aip.scitation.org/doi/abs/10.1063/1.527733}.

\bibitem[Valentini(1991)]{valentini_non-local_1991}
Antony Valentini.
\newblock Non-local correlations in quantum electrodynamics.
\newblock \emph{Physics Letters A}, 153\penalty0 (6):\penalty0 321--325, March
  1991.
\newblock ISSN 0375-9601.
\newblock \doi{10.1016/0375-9601(91)90952-5}.
\newblock URL
  \url{http://www.sciencedirect.com/science/article/pii/0375960191909525}.

\bibitem[Reznik et~al.(2005)Reznik, Retzker, and Silman]{reznik_violating_2005}
Benni Reznik, Alex Retzker, and Jonathan Silman.
\newblock Violating {{Bell}}'s inequalities in vacuum.
\newblock \emph{Phys. Rev. A}, 71\penalty0 (4):\penalty0 042104, April 2005.
\newblock \doi{10.1103/PhysRevA.71.042104}.
\newblock URL \url{http://link.aps.org/doi/10.1103/PhysRevA.71.042104}.

\bibitem[Salton et~al.(2015)Salton, Mann, and
  Menicucci]{salton_acceleration-assisted_2015}
Grant Salton, Robert~B. Mann, and Nicolas~C. Menicucci.
\newblock Acceleration-assisted entanglement harvesting and rangefinding.
\newblock \emph{New J. Phys.}, 17\penalty0 (3):\penalty0 035001, 2015.
\newblock ISSN 1367-2630.
\newblock \doi{10.1088/1367-2630/17/3/035001}.
\newblock URL \url{http://stacks.iop.org/1367-2630/17/i=3/a=035001}.

\bibitem[Steeg and Menicucci(2009)]{steeg_entangling_2009}
Greg~Ver Steeg and Nicolas~C. Menicucci.
\newblock Entangling power of an expanding universe.
\newblock \emph{Phys. Rev. D}, 79\penalty0 (4):\penalty0 044027, February 2009.
\newblock \doi{10.1103/PhysRevD.79.044027}.
\newblock URL \url{http://link.aps.org/doi/10.1103/PhysRevD.79.044027}.

\bibitem[Cliche and Kempf(2011)]{cliche_vacuum_2011}
M.~Cliche and A.~Kempf.
\newblock Vacuum entanglement enhancement by a weak gravitational field.
\newblock \emph{Phys. Rev. D}, 83\penalty0 (4):\penalty0 045019, February 2011.
\newblock \doi{10.1103/PhysRevD.83.045019}.
\newblock URL \url{http://link.aps.org/doi/10.1103/PhysRevD.83.045019}.

\bibitem[{Mart{\'i}n-Mart{\'i}nez} et~al.(2016){Mart{\'i}n-Mart{\'i}nez},
  Smith, and Terno]{martin-martinez_spacetime_2016}
Eduardo {Mart{\'i}n-Mart{\'i}nez}, Alexander R.~H. Smith, and Daniel~R. Terno.
\newblock Spacetime structure and vacuum entanglement.
\newblock \emph{Phys. Rev. D}, 93\penalty0 (4):\penalty0 044001, February 2016.
\newblock \doi{10.1103/PhysRevD.93.044001}.
\newblock URL \url{https://link.aps.org/doi/10.1103/PhysRevD.93.044001}.

\bibitem[Braun(2002)]{braun_creation_2002}
Daniel Braun.
\newblock Creation of {{Entanglement}} by {{Interaction}} with a {{Common Heat
  Bath}}.
\newblock \emph{Phys. Rev. Lett.}, 89\penalty0 (27):\penalty0 277901, December
  2002.
\newblock \doi{10.1103/PhysRevLett.89.277901}.
\newblock URL \url{https://link.aps.org/doi/10.1103/PhysRevLett.89.277901}.

\bibitem[Braun(2005)]{braun_entanglement_2005}
Daniel Braun.
\newblock Entanglement from thermal blackbody radiation.
\newblock \emph{Phys. Rev. A}, 72\penalty0 (6):\penalty0 062324, December 2005.
\newblock \doi{10.1103/PhysRevA.72.062324}.
\newblock URL \url{https://link.aps.org/doi/10.1103/PhysRevA.72.062324}.

\bibitem[Brown(2013)]{brown_thermal_2013}
Eric~G. Brown.
\newblock Thermal amplification of field-correlation harvesting.
\newblock \emph{Phys. Rev. A}, 88\penalty0 (6):\penalty0 062336, December 2013.
\newblock \doi{10.1103/PhysRevA.88.062336}.
\newblock URL \url{http://link.aps.org/doi/10.1103/PhysRevA.88.062336}.

\bibitem[Sachs et~al.(2017)Sachs, Mann, and
  {Mart{\'i}n-Mart{\'i}nez}]{sachs_entanglement_2017}
Allison Sachs, Robert~B. Mann, and Eduardo {Mart{\'i}n-Mart{\'i}nez}.
\newblock Entanglement harvesting and divergences in quadratic
  {{Unruh}}-{{DeWitt}} detector pairs.
\newblock \emph{Phys. Rev. D}, 96\penalty0 (8):\penalty0 085012, October 2017.
\newblock \doi{10.1103/PhysRevD.96.085012}.
\newblock URL \url{https://link.aps.org/doi/10.1103/PhysRevD.96.085012}.

\bibitem[Simidzija et~al.(2018)Simidzija, Jonsson, and
  {Mart{\'i}n-Mart{\'i}nez}]{simidzija_general_2018}
Petar Simidzija, Robert~H. Jonsson, and Eduardo {Mart{\'i}n-Mart{\'i}nez}.
\newblock General no-go theorem for entanglement extraction.
\newblock \emph{Phys. Rev. D}, 97\penalty0 (12):\penalty0 125002, June 2018.
\newblock \doi{10.1103/PhysRevD.97.125002}.
\newblock URL \url{https://link.aps.org/doi/10.1103/PhysRevD.97.125002}.

\bibitem[Simidzija and {Martin-Martinez}(2018)]{simidzija_harvesting_2018}
Petar Simidzija and Eduardo {Martin-Martinez}.
\newblock Harvesting correlations from thermal and squeezed coherent states.
\newblock \emph{Phys. Rev. D}, 98\penalty0 (8):\penalty0 085007, October 2018.
\newblock ISSN 2470-0010, 2470-0029.
\newblock \doi{10.1103/PhysRevD.98.085007}.
\newblock URL \url{http://arxiv.org/abs/1809.05547}.

\bibitem[Olson and Ralph(2011)]{olson_entanglement_2011}
S.~Jay Olson and Timothy~C. Ralph.
\newblock Entanglement between the {{Future}} and the {{Past}} in the {{Quantum
  Vacuum}}.
\newblock \emph{Phys. Rev. Lett.}, 106\penalty0 (11):\penalty0 110404, March
  2011.
\newblock \doi{10.1103/PhysRevLett.106.110404}.
\newblock URL \url{https://link.aps.org/doi/10.1103/PhysRevLett.106.110404}.

\bibitem[Olson and Ralph(2012)]{olson_extraction_2012}
S.~Jay Olson and Timothy~C. Ralph.
\newblock Extraction of timelike entanglement from the quantum vacuum.
\newblock \emph{Phys. Rev. A}, 85\penalty0 (1):\penalty0 012306, January 2012.
\newblock \doi{10.1103/PhysRevA.85.012306}.
\newblock URL \url{https://link.aps.org/doi/10.1103/PhysRevA.85.012306}.

\bibitem[Sab{\'i}n et~al.(2012)Sab{\'i}n, Peropadre, {del Rey}, and
  {Mart{\'i}n-Mart{\'i}nez}]{sabin_extracting_2012}
Carlos Sab{\'i}n, Borja Peropadre, Marco {del Rey}, and Eduardo
  {Mart{\'i}n-Mart{\'i}nez}.
\newblock Extracting {{Past}}-{{Future Vacuum Correlations Using Circuit QED}}.
\newblock \emph{Phys. Rev. Lett.}, 109\penalty0 (3):\penalty0 033602, July
  2012.
\newblock \doi{10.1103/PhysRevLett.109.033602}.
\newblock URL \url{https://link.aps.org/doi/10.1103/PhysRevLett.109.033602}.

\bibitem[{Pozas-Kerstjens} and
  {Mart{\'i}n-Mart{\'i}nez}(2016)]{pozas-kerstjens_entanglement_2016}
Alejandro {Pozas-Kerstjens} and Eduardo {Mart{\'i}n-Mart{\'i}nez}.
\newblock Entanglement harvesting from the electromagnetic vacuum with
  hydrogenlike atoms.
\newblock \emph{Phys. Rev. D}, 94\penalty0 (6):\penalty0 064074, September
  2016.
\newblock \doi{10.1103/PhysRevD.94.064074}.
\newblock URL \url{https://link.aps.org/doi/10.1103/PhysRevD.94.064074}.

\bibitem[Galve and Lutz(2009)]{galve_energy_2009}
Fernando Galve and Eric Lutz.
\newblock Energy cost and optimal entanglement production in harmonic chains.
\newblock \emph{Phys. Rev. A}, 79\penalty0 (3):\penalty0 032327, March 2009.
\newblock \doi{10.1103/PhysRevA.79.032327}.
\newblock URL \url{https://link.aps.org/doi/10.1103/PhysRevA.79.032327}.

\bibitem[{Mart{\'i}n-Mart{\'i}nez} et~al.(2013){Mart{\'i}n-Mart{\'i}nez},
  Brown, Donnelly, and Kempf]{martin-martinez_sustainable_2013}
Eduardo {Mart{\'i}n-Mart{\'i}nez}, Eric~G. Brown, William Donnelly, and Achim
  Kempf.
\newblock Sustainable entanglement production from a quantum field.
\newblock \emph{Phys. Rev. A}, 88\penalty0 (5):\penalty0 052310, November 2013.
\newblock \doi{10.1103/PhysRevA.88.052310}.
\newblock URL \url{http://link.aps.org/doi/10.1103/PhysRevA.88.052310}.

\bibitem[Huber et~al.(2015)Huber, {Perarnau-Llobet}, Hovhannisyan, Skrzypczyk,
  Kl{\"o}ckl, Brunner, and Ac{\'i}n]{huber_thermodynamic_2015}
Marcus Huber, Mart{\'i} {Perarnau-Llobet}, Karen~V. Hovhannisyan, Paul
  Skrzypczyk, Claude Kl{\"o}ckl, Nicolas Brunner, and Antonio Ac{\'i}n.
\newblock Thermodynamic cost of creating correlations.
\newblock \emph{New J. Phys.}, 17\penalty0 (6):\penalty0 065008, June 2015.
\newblock ISSN 1367-2630.
\newblock \doi{10.1088/1367-2630/17/6/065008}.
\newblock URL \url{https://doi.org/10.1088\%2F1367-2630\%2F17\%2F6\%2F065008}.

\bibitem[Friis et~al.(2016)Friis, Huber, and
  {Perarnau-Llobet}]{friis_energetics_2016}
Nicolai Friis, Marcus Huber, and Mart{\'i} {Perarnau-Llobet}.
\newblock Energetics of correlations in interacting systems.
\newblock \emph{Phys. Rev. E}, 93\penalty0 (4):\penalty0 042135, April 2016.
\newblock \doi{10.1103/PhysRevE.93.042135}.
\newblock URL \url{https://link.aps.org/doi/10.1103/PhysRevE.93.042135}.

\bibitem[Das et~al.(2017)Das, Kumar, Kumar~Pal, Shukla, Sen(De), and
  Sen]{das_canonical_2017}
Tamoghna Das, Asutosh Kumar, Amit Kumar~Pal, Namrata Shukla, Aditi Sen(De), and
  Ujjwal Sen.
\newblock Canonical distillation of entanglement.
\newblock \emph{Physics Letters A}, 381\penalty0 (41):\penalty0 3529--3535,
  November 2017.
\newblock ISSN 0375-9601.
\newblock \doi{10.1016/j.physleta.2017.08.065}.
\newblock URL
  \url{http://www.sciencedirect.com/science/article/pii/S0375960117308393}.

\bibitem[Chiribella and Yang(2017)]{chiribella_optimal_2017}
Giulio Chiribella and Yuxiang Yang.
\newblock Optimal quantum operations at zero energy cost.
\newblock \emph{Phys. Rev. A}, 96\penalty0 (2):\penalty0 022327, August 2017.
\newblock \doi{10.1103/PhysRevA.96.022327}.
\newblock URL \url{https://link.aps.org/doi/10.1103/PhysRevA.96.022327}.

\bibitem[Vitagliano et~al.(2018)Vitagliano, Kl{\"o}ckl, Huber, and
  Friis]{vitagliano_trade-off_2018}
Giuseppe Vitagliano, Claude Kl{\"o}ckl, Marcus Huber, and Nicolai Friis.
\newblock Trade-{{Off Between Work}} and {{Correlations}} in {{Quantum
  Thermodynamics}}.
\newblock In Felix Binder, Luis~A. Correa, Christian Gogolin, Janet Anders, and
  Gerardo Adesso, editors, \emph{Thermodynamics in the {{Quantum Regime}}:
  {{Fundamental Aspects}} and {{New Directions}}}, Fundamental {{Theories}} of
  {{Physics}}, pages 731--750. {Springer International Publishing}, {Cham},
  2018.
\newblock ISBN 978-3-319-99046-0.
\newblock \doi{10.1007/978-3-319-99046-0_30}.
\newblock URL \url{https://doi.org/10.1007/978-3-319-99046-0_30}.

\bibitem[Hotta(2008)]{hotta_quantum_2008}
Masahiro Hotta.
\newblock Quantum measurement information as a key to energy extraction from
  local vacuums.
\newblock \emph{Phys. Rev. D}, 78\penalty0 (4):\penalty0 045006, August 2008.
\newblock \doi{10.1103/PhysRevD.78.045006}.
\newblock URL \url{http://link.aps.org/doi/10.1103/PhysRevD.78.045006}.

\bibitem[Hotta(2010{\natexlab{a}})]{hotta_controlled_2010}
Masahiro Hotta.
\newblock Controlled {{Hawking}} process by quantum energy teleportation.
\newblock \emph{Phys. Rev. D}, 81\penalty0 (4):\penalty0 044025, February
  2010{\natexlab{a}}.
\newblock \doi{10.1103/PhysRevD.81.044025}.
\newblock URL \url{http://link.aps.org/doi/10.1103/PhysRevD.81.044025}.

\bibitem[Hotta(2010{\natexlab{b}})]{hotta_energy_2010}
Masahiro Hotta.
\newblock Energy entanglement relation for quantum energy teleportation.
\newblock \emph{Physics Letters A}, 374\penalty0 (34):\penalty0 3416--3421,
  July 2010{\natexlab{b}}.
\newblock ISSN 0375-9601.
\newblock \doi{10.1016/j.physleta.2010.06.058}.
\newblock URL
  \url{http://www.sciencedirect.com/science/article/pii/S0375960110007723}.

\bibitem[Hotta et~al.(2014)Hotta, Matsumoto, and Yusa]{hotta_quantum_2014}
Masahiro Hotta, Jiro Matsumoto, and Go~Yusa.
\newblock Quantum energy teleportation without a limit of distance.
\newblock \emph{Phys. Rev. A}, 89\penalty0 (1):\penalty0 012311, January 2014.
\newblock \doi{10.1103/PhysRevA.89.012311}.
\newblock URL \url{http://link.aps.org/doi/10.1103/PhysRevA.89.012311}.

\bibitem[B{\'e}ny et~al.(2018)B{\'e}ny, Chubb, Farrelly, and
  Osborne]{beny_energy_2018}
C{\'e}dric B{\'e}ny, Christopher~T. Chubb, Terry Farrelly, and Tobias~J.
  Osborne.
\newblock Energy cost of entanglement extraction in complex quantum systems.
\newblock \emph{Nat. Commun.}, 9\penalty0 (1):\penalty0 3792, September 2018.
\newblock ISSN 2041-1723.
\newblock \doi{10.1038/s41467-018-06153-w}.
\newblock URL \url{https://www.nature.com/articles/s41467-018-06153-w}.

\bibitem[Jacobson(2016)]{jacobson_entanglement_2016}
Ted Jacobson.
\newblock Entanglement {{Equilibrium}} and the {{Einstein Equation}}.
\newblock \emph{Phys. Rev. Lett.}, 116\penalty0 (20), May 2016.
\newblock ISSN 0031-9007, 1079-7114.
\newblock \doi{10.1103/PhysRevLett.116.201101}.
\newblock URL \url{http://arxiv.org/abs/1505.04753}.

\bibitem[Hotta et~al.(2015)Hotta, Sch{\"u}tzhold, and
  Unruh]{hotta_partner_2015}
M.~Hotta, R.~Sch{\"u}tzhold, and W.~G. Unruh.
\newblock On the partner particles for moving mirror radiation and black hole
  evaporation.
\newblock \emph{ArXiv150306109 Gr-Qc Physicsquant-Ph}, March 2015.
\newblock URL \url{http://arxiv.org/abs/1503.06109}.

\bibitem[Trevison et~al.(2018)Trevison, Yamaguchi, and
  Hotta]{trevison_pure_2018}
Jose Trevison, Koji Yamaguchi, and Masahiro Hotta.
\newblock Pure state entanglement harvesting in quantum field theory.
\newblock \emph{Prog Theor Exp Phys}, 2018\penalty0 (10), October 2018.
\newblock \doi{10.1093/ptep/pty109}.
\newblock URL \url{https://arxiv.org/abs/1808.01764}.

\bibitem[Trevison et~al.(2019)Trevison, Yamaguchi, and
  Hotta]{trevison_spatially_2019}
Jose Trevison, Koji Yamaguchi, and Masahiro Hotta.
\newblock Spatially overlapped partners in quantum field theory.
\newblock \emph{J. Phys. A: Math. Theor.}, 52\penalty0 (12):\penalty0 125402,
  February 2019.
\newblock ISSN 1751-8121.
\newblock \doi{10.1088/1751-8121/ab065b}.
\newblock URL \url{https://doi.org/10.1088\%2F1751-8121\%2Fab065b}.

\bibitem[{Ashtekar A.} et~al.(1975){Ashtekar A.}, {Magnon Anne}, and {Penrose
  Roger}]{ashtekara._quantum_1975}
{Ashtekar A.}, {Magnon Anne}, and {Penrose Roger}.
\newblock Quantum fields in curved space-times.
\newblock \emph{Proceedings of the Royal Society of London. A. Mathematical and
  Physical Sciences}, 346\penalty0 (1646):\penalty0 375--394, November 1975.
\newblock \doi{10.1098/rspa.1975.0181}.
\newblock URL
  \url{https://royalsocietypublishing.org/doi/abs/10.1098/rspa.1975.0181}.

\bibitem[Wald(1994)]{wald_quantum_1994}
Robert~M. Wald.
\newblock \emph{Quantum Field Theory in Curved Spacetime and Black Hole
  Thermodynamics}.
\newblock Chicago Lectures in Physics. {University of Chicago Press},
  {Chicago}, 1994.
\newblock ISBN 0-226-87025-1.

\bibitem[Botero and Reznik(2003)]{botero_mode-wise_2003}
Alonso Botero and Benni Reznik.
\newblock Mode-{{Wise Entanglement}} of {{Gaussian States}}.
\newblock \emph{Phys. Rev. A}, 67\penalty0 (5), May 2003.
\newblock ISSN 1050-2947, 1094-1622.
\newblock \doi{10.1103/PhysRevA.67.052311}.
\newblock URL \url{http://arxiv.org/abs/quant-ph/0209026}.

\bibitem[Wolf(2008)]{wolf_not-so-normal_2008}
Michael~M. Wolf.
\newblock Not-{{So}}-{{Normal Mode Decomposition}}.
\newblock \emph{Phys. Rev. Lett.}, 100\penalty0 (7):\penalty0 070505, February
  2008.
\newblock \doi{10.1103/PhysRevLett.100.070505}.
\newblock URL \url{https://link.aps.org/doi/10.1103/PhysRevLett.100.070505}.

\bibitem[Bianchi et~al.(2016)Bianchi, Guglielmon, Hackl, and
  Yokomizo]{bianchi_squeezed_2016}
Eugenio Bianchi, Jonathan Guglielmon, Lucas Hackl, and Nelson Yokomizo.
\newblock Squeezed vacua in loop quantum gravity.
\newblock \emph{ArXiv160505356 Gr-Qc Physicshep-Th}, May 2016.
\newblock URL \url{http://arxiv.org/abs/1605.05356}.

\bibitem[Hackl et~al.(2018)Hackl, Bianchi, Modak, and
  Rigol]{hackl_entanglement_2018}
Lucas Hackl, Eugenio Bianchi, Ranjan Modak, and Marcos Rigol.
\newblock Entanglement production in bosonic systems: {{Linear}} and
  logarithmic growth.
\newblock \emph{Phys. Rev. A}, 97\penalty0 (3):\penalty0 032321, March 2018.
\newblock \doi{10.1103/PhysRevA.97.032321}.
\newblock URL \url{https://link.aps.org/doi/10.1103/PhysRevA.97.032321}.

\bibitem[Hackl(2018)]{hackl_aspects_2018}
Lucas~Fabian Hackl.
\newblock \emph{Aspects of {{Gaussian States}}: {{Entanglement}}, {{Squeezing}}
  and {{Complexity}}}.
\newblock PhD thesis, Pennsylvania State University, July 2018.
\newblock URL \url{https://etda.libraries.psu.edu/catalog/15815lfh109}.

\bibitem[Weedbrook et~al.(2012)Weedbrook, Pirandola, {Garc{\'i}a-Patr{\'o}n},
  Cerf, Ralph, Shapiro, and Lloyd]{weedbrook_gaussian_2012}
Christian Weedbrook, Stefano Pirandola, Ra{\'u}l {Garc{\'i}a-Patr{\'o}n},
  Nicolas~J. Cerf, Timothy~C. Ralph, Jeffrey~H. Shapiro, and Seth Lloyd.
\newblock Gaussian quantum information.
\newblock \emph{Rev. Mod. Phys.}, 84\penalty0 (2):\penalty0 621--669, May 2012.
\newblock \doi{10.1103/RevModPhys.84.621}.
\newblock URL \url{http://link.aps.org/doi/10.1103/RevModPhys.84.621}.

\bibitem[Plenio et~al.(2005)Plenio, Eisert, Drei{\ss}ig, and
  Cramer]{plenio_entropy_2005}
M.~B. Plenio, J.~Eisert, J.~Drei{\ss}ig, and M.~Cramer.
\newblock Entropy, {{Entanglement}}, and {{Area}}: {{Analytical Results}} for
  {{Harmonic Lattice Systems}}.
\newblock \emph{Phys. Rev. Lett.}, 94\penalty0 (6):\penalty0 060503, February
  2005.
\newblock \doi{10.1103/PhysRevLett.94.060503}.
\newblock URL \url{https://link.aps.org/doi/10.1103/PhysRevLett.94.060503}.

\bibitem[Casini and Huerta(2009)]{casini_entanglement_2009}
H.~Casini and M.~Huerta.
\newblock Entanglement entropy in free quantum field theory.
\newblock \emph{J. Phys. A: Math. Theor.}, 42\penalty0 (50):\penalty0 504007,
  2009.
\newblock ISSN 1751-8121.
\newblock \doi{10.1088/1751-8113/42/50/504007}.
\newblock URL \url{http://stacks.iop.org/1751-8121/42/i=50/a=504007}.

\bibitem[Woit(2017)]{woit_quantum_2017}
Peter Woit.
\newblock \emph{Quantum {{Theory}}, {{Groups}} and {{Representations}}: {{An
  Introduction}}}.
\newblock {Springer International Publishing}, 2017.
\newblock ISBN 978-3-319-64610-7.
\newblock \doi{10.1007/978-3-319-64612-1}.
\newblock URL \url{https://www.springer.com/gp/book/9783319646107}.

\bibitem[Hudson(1974)]{hudson_when_1974}
R.~L. Hudson.
\newblock When is the wigner quasi-probability density non-negative?
\newblock \emph{Reports on Mathematical Physics}, 6\penalty0 (2):\penalty0
  249--252, October 1974.
\newblock ISSN 0034-4877.
\newblock \doi{10.1016/0034-4877(74)90007-X}.
\newblock URL
  \url{http://www.sciencedirect.com/science/article/pii/003448777490007X}.

\bibitem[Soto and Claverie(1983)]{soto_when_1983}
Francisco Soto and Pierre Claverie.
\newblock When is the {{Wigner}} function of multidimensional systems
  nonnegative?
\newblock \emph{Journal of Mathematical Physics}, 24\penalty0 (1):\penalty0
  97--100, January 1983.
\newblock ISSN 0022-2488.
\newblock \doi{10.1063/1.525607}.
\newblock URL \url{https://aip.scitation.org/doi/10.1063/1.525607}.

\bibitem[Wick(1950)]{wick_evaluation_1950}
G.~C. Wick.
\newblock The {{Evaluation}} of the {{Collision Matrix}}.
\newblock \emph{Phys. Rev.}, 80\penalty0 (2):\penalty0 268--272, October 1950.
\newblock \doi{10.1103/PhysRev.80.268}.
\newblock URL \url{https://link.aps.org/doi/10.1103/PhysRev.80.268}.

\bibitem[Vidmar et~al.(2017)Vidmar, Hackl, Bianchi, and
  Rigol]{vidmar_entanglement_2017}
Lev Vidmar, Lucas Hackl, Eugenio Bianchi, and Marcos Rigol.
\newblock Entanglement {{Entropy}} of {{Eigenstates}} of {{Quadratic Fermionic
  Hamiltonians}}.
\newblock \emph{Phys. Rev. Lett.}, 119\penalty0 (2):\penalty0 020601, July
  2017.
\newblock \doi{10.1103/PhysRevLett.119.020601}.
\newblock URL \url{https://link.aps.org/doi/10.1103/PhysRevLett.119.020601}.

\bibitem[Vidmar et~al.(2018)Vidmar, Hackl, Bianchi, and
  Rigol]{vidmar_volume_2018}
Lev Vidmar, Lucas Hackl, Eugenio Bianchi, and Marcos Rigol.
\newblock Volume {{Law}} and {{Quantum Criticality}} in the {{Entanglement
  Entropy}} of {{Excited Eigenstates}} of the {{Quantum Ising Model}}.
\newblock \emph{Phys. Rev. Lett.}, 121\penalty0 (22):\penalty0 220602, November
  2018.
\newblock \doi{10.1103/PhysRevLett.121.220602}.
\newblock URL \url{https://link.aps.org/doi/10.1103/PhysRevLett.121.220602}.

\bibitem[Hackl et~al.(2019)Hackl, Vidmar, Rigol, and
  Bianchi]{hackl_average_2019}
Lucas Hackl, Lev Vidmar, Marcos Rigol, and Eugenio Bianchi.
\newblock Average eigenstate entanglement entropy of the {{XY}} chain in a
  transverse field and its universality for translationally invariant quadratic
  fermionic models.
\newblock \emph{Phys. Rev. B}, 99\penalty0 (7):\penalty0 075123, February 2019.
\newblock \doi{10.1103/PhysRevB.99.075123}.
\newblock URL \url{https://link.aps.org/doi/10.1103/PhysRevB.99.075123}.

\bibitem[Chapman et~al.(2019)Chapman, Eisert, Hackl, Heller, Jefferson,
  Marrochio, and Myers]{chapman_complexity_2019}
Shira Chapman, Jens Eisert, Lucas Hackl, Michal Heller, Ro~Jefferson, Hugo
  Marrochio, and Robert Myers.
\newblock Complexity and entanglement for thermofield double states.
\newblock \emph{SciPost Phys.}, 6\penalty0 (3):\penalty0 034, March 2019.
\newblock ISSN 2542-4653.
\newblock \doi{10.21468/SciPostPhys.6.3.034}.
\newblock URL \url{https://scipost.org/10.21468/SciPostPhys.6.3.034}.

\bibitem[Sorkin(2014)]{sorkin_entropy_2014}
Rafael~D. Sorkin.
\newblock On the {{Entropy}} of the {{Vacuum}} outside a {{Horizon}}.
\newblock \emph{ArXiv14023589 Cond-Mat Physicsgr-Qc Physicshep-Th
  Physicsquant-Ph}, February 2014.
\newblock URL \url{http://arxiv.org/abs/1402.3589}.

\bibitem[Bombelli et~al.(1986)Bombelli, Koul, Lee, and
  Sorkin]{bombelli_quantum_1986}
Luca Bombelli, Rabinder~K. Koul, Joohan Lee, and Rafael~D. Sorkin.
\newblock Quantum source of entropy for black holes.
\newblock \emph{Phys. Rev. D}, 34\penalty0 (2):\penalty0 373--383, July 1986.
\newblock \doi{10.1103/PhysRevD.34.373}.
\newblock URL \url{https://link.aps.org/doi/10.1103/PhysRevD.34.373}.

\bibitem[Ba{\~n}uls et~al.(2007)Ba{\~n}uls, Cirac, and
  Wolf]{banuls_entanglement_2007}
Mari-Carmen Ba{\~n}uls, J.~Ignacio Cirac, and Michael~M. Wolf.
\newblock Entanglement in fermionic systems.
\newblock \emph{Phys. Rev. A}, 76\penalty0 (2), August 2007.
\newblock ISSN 1050-2947, 1094-1622.
\newblock \doi{10.1103/PhysRevA.76.022311}.
\newblock URL \url{http://arxiv.org/abs/0705.1103}.

\bibitem[Schwabl(2008)]{schwabl_advanced_2008}
Franz Schwabl.
\newblock \emph{Advanced {{Quantum Mechanics}}}.
\newblock {Springer-Verlag}, {Berlin Heidelberg}, 4 edition, 2008.
\newblock ISBN 978-3-540-85061-8.
\newblock \doi{10.1007/978-3-540-85062-5}.
\newblock URL \url{https://www.springer.com/gp/book/9783540850618}.

\bibitem[Jordan and Wigner(1993)]{jordan_ueber_1993}
P.~Jordan and E.~P. Wigner.
\newblock {{\"U}ber das Paulische {\"A}quivalenzverbot}.
\newblock In Arthur~S. Wightman, editor, \emph{{The Collected Works of Eugene
  Paul Wigner: Part A: The Scientific Papers}}, {The Collected Works of Eugene
  Paul Wigner}, pages 109--129. {Springer Berlin Heidelberg}, {Berlin,
  Heidelberg}, 1993.
\newblock ISBN 978-3-662-02781-3.
\newblock \doi{10.1007/978-3-662-02781-3_9}.
\newblock URL \url{https://doi.org/10.1007/978-3-662-02781-3_9}.

\bibitem[Lieb et~al.(1961)Lieb, Schultz, and Mattis]{lieb_two_1961}
Elliott Lieb, Theodore Schultz, and Daniel Mattis.
\newblock Two soluble models of an antiferromagnetic chain.
\newblock \emph{Annals of Physics}, 16\penalty0 (3):\penalty0 407--466,
  December 1961.
\newblock ISSN 0003-4916.
\newblock \doi{10.1016/0003-4916(61)90115-4}.
\newblock URL
  \url{http://www.sciencedirect.com/science/article/pii/0003491661901154}.

\bibitem[Holstein and Primakoff(1940)]{holstein_field_1940}
T.~Holstein and H.~Primakoff.
\newblock Field {{Dependence}} of the {{Intrinsic Domain Magnetization}} of a
  {{Ferromagnet}}.
\newblock \emph{Phys. Rev.}, 58\penalty0 (12):\penalty0 1098--1113, December
  1940.
\newblock \doi{10.1103/PhysRev.58.1098}.
\newblock URL \url{https://link.aps.org/doi/10.1103/PhysRev.58.1098}.

\bibitem[Cazalilla et~al.(2011)Cazalilla, Citro, Giamarchi, Orignac, and
  Rigol]{cazalilla_one_2011}
M.~A. Cazalilla, R.~Citro, T.~Giamarchi, E.~Orignac, and M.~Rigol.
\newblock One dimensional bosons: {{From}} condensed matter systems to
  ultracold gases.
\newblock \emph{Rev. Mod. Phys.}, 83\penalty0 (4):\penalty0 1405--1466,
  December 2011.
\newblock \doi{10.1103/RevModPhys.83.1405}.
\newblock URL \url{https://link.aps.org/doi/10.1103/RevModPhys.83.1405}.

\bibitem[Pfeuty(1970)]{pfeuty_one-dimensional_1970}
Pierre Pfeuty.
\newblock The one-dimensional {{Ising}} model with a transverse field.
\newblock \emph{Annals of Physics}, 57\penalty0 (1):\penalty0 79--90, March
  1970.
\newblock ISSN 0003-4916.
\newblock \doi{10.1016/0003-4916(70)90270-8}.
\newblock URL
  \url{http://www.sciencedirect.com/science/article/pii/0003491670902708}.

\bibitem[Dahlsten et~al.(2014)Dahlsten, Lupo, Mancini, and
  Serafini]{dahlsten_entanglement_2014}
Oscar C.~O. Dahlsten, Cosmo Lupo, Stefano Mancini, and Alessio Serafini.
\newblock Entanglement typicality.
\newblock \emph{J. Phys. A: Math. Theor.}, 47\penalty0 (36):\penalty0 363001,
  August 2014.
\newblock ISSN 1751-8121.
\newblock \doi{10.1088/1751-8113/47/36/363001}.
\newblock URL \url{https://doi.org/10.1088\%2F1751-8113\%2F47\%2F36\%2F363001}.

\bibitem[Hackl and Jonsson()]{hackl_jonsson_2019b}
Lucas Hackl and Robert~H. Jonsson.
\newblock \emph{in preparation}.

\bibitem[Chin et~al.(2010)Chin, Rivas, Huelga, and Plenio]{chin_exact_2010}
Alex~W. Chin, {\'A}ngel Rivas, Susana~F. Huelga, and Martin~B. Plenio.
\newblock Exact mapping between system-reservoir quantum models and
  semi-infinite discrete chains using orthogonal polynomials.
\newblock \emph{Journal of Mathematical Physics}, 51\penalty0 (9):\penalty0
  092109, September 2010.
\newblock ISSN 0022-2488.
\newblock \doi{10.1063/1.3490188}.
\newblock URL \url{https://aip.scitation.org/doi/abs/10.1063/1.3490188}.

\bibitem[Mancini et~al.(2014)Mancini, Pierini, and
  Wilde]{mancini_preserving_2014}
Stefano Mancini, Roberto Pierini, and Mark~M. Wilde.
\newblock Preserving information from the beginning to the end of time in a
  {{Robertson}}\textendash{{Walker}} spacetime.
\newblock \emph{New J. Phys.}, 16\penalty0 (12):\penalty0 123049, December
  2014.
\newblock ISSN 1367-2630.
\newblock \doi{10.1088/1367-2630/16/12/123049}.
\newblock URL \url{https://doi.org/10.1088\%2F1367-2630\%2F16\%2F12\%2F123049}.

\bibitem[Jonsson et~al.(2015)Jonsson, {Mart{\'i}n-Mart{\'i}nez}, and
  Kempf]{jonsson_information_2015}
Robert~H. Jonsson, Eduardo {Mart{\'i}n-Mart{\'i}nez}, and Achim Kempf.
\newblock Information {{Transmission Without Energy Exchange}}.
\newblock \emph{Phys. Rev. Lett.}, 114\penalty0 (11):\penalty0 110505, March
  2015.
\newblock \doi{10.1103/PhysRevLett.114.110505}.
\newblock URL \url{http://link.aps.org/doi/10.1103/PhysRevLett.114.110505}.

\bibitem[Jonsson(2016)]{jonsson_information_2016}
Robert~H. Jonsson.
\newblock Information travels in massless fields in 1+1 dimensions where energy
  cannot.
\newblock \emph{J. Phys. A: Math. Theor.}, 49\penalty0 (44):\penalty0 445402,
  2016.
\newblock ISSN 1751-8121.
\newblock \doi{10.1088/1751-8113/49/44/445402}.
\newblock URL \url{http://stacks.iop.org/1751-8121/49/i=44/a=445402}.

\bibitem[Jonsson(2017)]{jonsson_quantum_2017}
Robert~H. Jonsson.
\newblock Quantum signaling in relativistic motion and across acceleration
  horizons.
\newblock \emph{J. Phys. A: Math. Theor.}, 50\penalty0 (35):\penalty0 355401,
  2017.
\newblock ISSN 1751-8121.
\newblock \doi{10.1088/1751-8121/aa7d3c}.
\newblock URL \url{http://stacks.iop.org/1751-8121/50/i=35/a=355401}.

\bibitem[Jonsson et~al.(2014)Jonsson, {Mart{\'i}n-Mart{\'i}nez}, and
  Kempf]{jonsson_quantum_2014}
Robert~H. Jonsson, Eduardo {Mart{\'i}n-Mart{\'i}nez}, and Achim Kempf.
\newblock Quantum signaling in cavity {{QED}}.
\newblock \emph{Phys. Rev. A}, 89\penalty0 (2):\penalty0 022330, February 2014.
\newblock \doi{10.1103/PhysRevA.89.022330}.
\newblock URL \url{http://link.aps.org/doi/10.1103/PhysRevA.89.022330}.

\bibitem[Jonsson et~al.(2018)Jonsson, Ried, {Mart{\'i}n-Mart{\'i}nez}, and
  Kempf]{jonsson_transmitting_2018}
Robert~H. Jonsson, Katja Ried, Eduardo {Mart{\'i}n-Mart{\'i}nez}, and Achim
  Kempf.
\newblock Transmitting qubits through relativistic fields.
\newblock \emph{J. Phys. A: Math. Theor.}, 51\penalty0 (48):\penalty0 485301,
  2018.
\newblock ISSN 1751-8121.
\newblock \doi{10.1088/1751-8121/aae78a}.
\newblock URL \url{http://stacks.iop.org/1751-8121/51/i=48/a=485301}.

\bibitem[Abraham et~al.(1978)Abraham, Marsden, and
  Marsden]{abraham1978foundations}
Ralph Abraham, Jerrold~E Marsden, and Jerrold~E Marsden.
\newblock \emph{Foundations of mechanics}, volume~36.
\newblock Benjamin/Cummings Publishing Company Reading, Massachusetts, 1978.

\end{thebibliography}

\end{document}